\documentclass[paper=a4, fontsize=11pt]{scrartcl}
\usepackage[T1]{fontenc}
\usepackage{fourier}
\usepackage[english]{babel}															
\usepackage{amssymb}
\usepackage[protrusion=true,expansion=true]{microtype}	
\usepackage{amsmath,amsfonts,amsthm} 
\usepackage[pdftex]{graphicx}	
\usepackage{url}
\usepackage[title,titletoc,page]{appendix} 
\usepackage{listings} 
\usepackage{color}
\usepackage{comment}
\usepackage{caption}
\usepackage{subcaption}
\usepackage[ruled,vlined]{algorithm2e}
\usepackage{sectsty}
\allsectionsfont{\centering \normalfont\scshape}

\usepackage{fancyhdr}
\numberwithin{equation}{section}		
\numberwithin{figure}{section}			
\numberwithin{table}{section}				

\newtheorem{defi}{Definition}[section]
\newtheorem{lem}{Lemma}[section]
\newtheorem{thm}{Theorem}[section]

\newtheorem{remark}{Remark}[section]

\newtheorem{prop}{Proposition}[section]

\newcommand{\babs}[1]{\left|{#1}\right|}
\newcommand{\bnorm}[1]{\left|\left|{#1}\right|\right|}
\newcommand{\btwonorm}[1]{\left|\left|{#1}\right|\right|_2}

\newcommand{\vect}[1]{\boldsymbol{\mathbf{#1}}}

\newcommand{\R}{\mathbb{R}}
\newcommand{\N}{\mathbb{N}}

\title{A mathematical theory of resolution limits for super-resolution of positive sources
	\thanks{\footnotesize This work was supported in part by the Swiss National Science Foundation grant number
		200021--200307.}}
\author{Ping Liu\thanks{\footnotesize Department of Mathematics, ETH Z\"urich, R\"amistrasse 101, CH-8092 Z\"urich, Switzerland (ping.liu@sam.math.ethz.ch, yanchen.he@sam.math.ethz.ch, habib.ammari@math.ethz.ch).} \and Yanchen He\footnotemark[2] \and Habib Ammari\footnotemark[2]}

\begin{document}
\maketitle

\begin{abstract}
A priori information on the positivity of source intensities is ubiquitous in imaging fields, and is also important for a multitude of super-resolution and deconvolution algorithms.  But the fundamental resolution limit of positive sources is still unknown, and research in this field is very limited indeed. In this work, we analyze the super-resolving capacity for number and location recoveries in the super-resolution of positive sources and aim to answer the resolution limit problem in a rigorous manner. Specifically, we introduce the computational resolution limit for respectively the number detection and location recovery in the one-dimensional super-resolution problem and quantitatively characterize their dependency on the cutoff frequency, signal-to-noise ratio, and the sparsity of the sources. As a direct consequence, we show that targeting at the sparest positive solution in the super-resolution already provides the optimal resolution order. These results are generalized to multi-dimensional spaces. Our estimates indicate that there exist phase transitions in the corresponding reconstructions, which are confirmed by numerical experiments. On the other hand, despite the fact that positivity plays important roles in improving the resolution of certain super-resolution algorithms, our theory has made several different but significant discoveries: i) The a priori information of positivity cannot further improve the order of the resolution limit; ii) The positivity of the source sometimes deteriorates the resolution limit instead of enhancing it. In particular, under certain the signal-to-noise ratio, two point sources with different phases actually have better resolution limit than those with the same one. 
\end{abstract}

\vspace{0.5cm}
		
\vspace{0.2cm}
\vspace{0.5cm}	

\section{Introduction}
In recent years, the development of super-resolution optical microscopy led to a revolutionary improvement of resolution through the use of different technical approaches. This impressive success has generated significant interest in studying the super-resolution algorithms and the fundamental super-resolving capability. In this paper, we aim to study the super-resolving capacity of the number and locations recovery in the super-resolution of positive sources. To be more specific, we consider the following mathematical model. Let $\mu=\sum_{j=1}^{n}a_{j}\delta_{y_j}, a_j>0$ be a positive discrete measure, where $y_j \in \mathbb R,j=1,\cdots,n$, represent the location of the point sources and $a_j>0, j=1,\cdots,n,$ their amplitudes. Noting that $y_j$'s are the supports of the Dirac masses in $\mu$.  In this paper we will use support recovery instead of location reconstruction. We denote by
\[
m_{\min}=\min_{j=1,\cdots,n}a_j,
\quad 	d_{\min}=\min_{p\neq j}|y_p-y_j|.
\] 
The measurement is the noisy Fourier data of $\mu$ in a bounded interval, that is, 
\begin{equation}\label{equ:modelsetting1}
\mathbf Y(\omega) = \mathcal F [\mu] (\omega) + \mathbf W(\omega)= \sum_{j=1}^{n}a_j e^{i y_j \omega} + \mathbf W(\omega),\quad  \omega \in [-\Omega, \Omega],
\end{equation}
with $\mathbf W(\omega)$ being the noise and $\Omega$ the cutoff frequency of the imaging system. We assume that \[
|\mathbf W(\omega)|< \sigma,\quad  \omega \in [-\Omega, \Omega],
\]
with $\sigma$ being the noise level. The above measurement model is chosen for convenience. All the results in this paper also hold for the case when taking measurement at a sufficient number of evenly-spaced points as what was considered in \cite{liu2021theorylse}. Thus our results can be applied to practical situations and real-world problems.  

The super-resolution problem we are interested in is to recover the positive discrete measure $\mu$ from the above noisy measurement $\mathbf Y$.  We note that the super-resolution problem is closely related to the line spectral estimation problem \cite{liu2021theorylse} which is at the core of diverse fields such as wireless communications and array processing. It should also be pointed out that the measurement discussed in equation (\ref{equ:modelsetting1}) pertains to imaging the convolution of positive point sources with a band-limited or general point spread function $f$. Therefore, with slight modifications, our findings can also be employed to analyze the stability of deconvolving positive sources, which is a crucial problem in various fields.

\subsection{Literature review}
\textbf{Fundamental limits.} In \cite{shahram2004imaging, shahram2005resolvability, helstrom1964detection}, the authors analyzed the resolution limit in the detection of two closely-spaced point sources based on the statistical inference theory, but their theory has not been generalized to the case when there are more than two sources in the signal. The mathematical theory for analyzing the fundamental limit in super-resolving multiple point sources was pioneered by Donoho \cite{donoho1992superresolution}  in 1992. In that work, he considered a grid setting where a discrete measure is supported on a lattice with spacing $\Delta$ and regularized by a so-called  "Rayleigh index". The problem is to reconstruct the amplitudes of the grid points from their noisy Fourier data in $[-\Omega, \Omega]$ with $\Omega$ being the band limit. His main contribution is estimating the corresponding minimax error in the recovery, which emphasizes the importance of the sparsity of sources for the super-resolution. It was improved in recent years for the case when only $n$ point sources are presented. In \cite{demanet2015recoverability}, the authors considered resolving $n$-sparse point sources supported on a grid and showed that the minimax error of amplitude recovery in the presence of noise with magnitude $\sigma$ scales exactly as $SRF^{2n-1}\sigma$, where $SRF:=\frac{1}{\Delta \Omega}$ is the super-resolution factor. The case of multi-clustered point sources was considered in \cite{li2021stable, batenkov2020conditioning} and similar minimax error estimates were derived. Moreover, in \cite{akinshin2015accuracy, batenkov2019super} the authors considered the minimax error for recovering the amplitudes and locations of off-the-grid point sources. They showed that  for $\sigma \lessapprox (SRF)^{-2p+1}$, where $p$ is the number of point sources in a cluster, the minimax error for the amplitude and the location recoveries scale respectively as $(SRF)^{2p-1}\sigma$, $(SRF)^{2p-2} {\sigma}/{\Omega}$, while for the single non-clustered source away from other sources, the corresponding minimax error for the amplitude and the location recoveries scale respectively as $\sigma$ and ${\sigma}/{\Omega}$. We also refer the readers to \cite{moitra2015super, chen2020algorithmic} for understanding the resolution limit from the perspective of sample complexity and to  \cite{tang2015resolution, da2020stable} for the resolving limit of some algorithms. 

On the other hand, in order to characterize the exact resolution in resolving multiple point sources like the classical Rayleigh limit, in the earlier works \cite{liu2021mathematicaloned, liu2021mathematicalhighd, liu2021theorylse, liu2022nearly, liu2021mathematical} we defined the concept of "computational resolution limit" as the minimum required distance between point sources so that their number and locations can be stably resolved under certain noise level. By developing a non-linear approximation theory in a so-called Vandermonde space, we derived sharp bounds for computational resolution limits in one- and multi-dimensional super-resolution problems. In particular, we showed that the computational resolution limits for number and location recoveries should be respectively $\frac{C_{\mathrm{num}}(n,k)}{\Omega}\left(\frac{\sigma}{m_{\min}}\right)^{\frac{1}{2n-2}}$ and  $\frac{C_{\mathrm{supp}}(n,k)}{\Omega}\left(\frac{\sigma}{m_{\min}}\right)^{\frac{1}{2n-1}}$, where $C_{\mathrm{num}}(n,k), C_{\mathrm{supp}}(n,k)$ are constants depending only on source number $n$ and space dimensionality $k$. In this paper, we will generalize these results to the super-resolution problem of positive sources. 

\medskip
\noindent\textbf{Reconstruction algorithms.}
Due to the importance of super-resolution in applications, a number of sophisticated super-resolution algorithms have been developed over the years. Among those algorithms, a class of algorithms called subspace methods have exhibited favourable performance and have been used frequently in engineering applications. Specific examples include MUltiple SIgnal Classification (MUSIC) \cite{schmidt1986multiple}, Estimation of Signal Parameters via Rotational Invariance Technique (ESPRIT) \cite{roy1989esprit}, and Matrix Pencil Method \cite{hua1990matrix}. Note that these algorithms date back to the work of Prony \cite{Prony-1795}. Despite the appealing performance of the subspace methods in practical applications, their stability properties are not yet well-understood. The asymptotic results on the stability of MUSIC algorithm in the presence of Gaussian noise were derived just slightly after its emerging \cite{stoica1989music, clergeot1989performance, stoica1991statistical}. But only until recent years, some steps towards understanding the stability of MUSIC, ESPRIT and Matrix-Pencil Method in the non-asymptotic regime were taken in \cite{liao2016music}, \cite{li2020super} and \cite{moitra2015super}, respectively. Nevertheless, the error tolerance derived in these papers are not as strong as our estimates here for the resolution limit in location reconstruction. On the other hand, it was shown numerically in \cite{liao2016music, batenkov2019super, li2020super} that these subspace methods actually achieve the optimal resolution order. Thus the theoretical demonstrations for the performance limits of subspace methods in the non-asymptotic regime are still important open problems.

In recent years, inspired by the idea of sparse modeling and compressed sensing, many sparsity promoting algorithms have been proposed for the super-resolution problem. For example, in the groundbreaking work of Cand\`es and Fernandez-Granda \cite{candes2014towards}, it was proved that off-the-grid sources can be exactly recovered from their low-frequency measurements by a TV minimization under a minimum separation condition. It invokes active researches in the off-the-grid algorithms, among them we would like to mention the BLASSO \cite{azais2015spike, duval2015exact, poon2019} and the atomic norm minimization method \cite{tang2013compressed,tang2014near}. Both methods were proved to be able to stably recover the source under a minimum separation condition or a non-degeneracy condition. 


\medskip
\noindent\textbf{Super-resolution of positive sources.}
To the best of our knowledge, the theoretical possibility for the super-resolution of positive sources was first considered in \cite{donoho1992maximum}. Specifically, the authors defined
$$
\omega(\sigma ; \mathbf{x})=\sup \left\{\left\|\mathbf{x}^{\prime}-\mathbf{x}\right\|_1:\left\|K \mathbf{x}^{\prime}-K \mathbf{x}\right\|_2 \leqslant \sigma \text { and } \mathbf{x}^{\prime} \geqslant 0\right\}, 
$$
where $\vect x, \vect x^{\prime}$ are vectors of length $M$ (sources are on a grid) and $K$ consists of the first $m$ rows of the $M\times M$ discrete Fourier transform, and said that $\mathbf{x}$ admits super-resolution if
$$
\omega(\sigma ; \mathbf{x}) \rightarrow 0 \quad \text { as } \sigma \rightarrow 0.
$$
They demonstrated for $\vect x\geq 0$ that: 
(a) If $\mathbf{x}$ has $\frac{1}{2}(m-1)$ or fewer non-zero elements then $\vect x$ admits super-resolution; (b) If $\frac{1}{2}(m+1)$ divides $M$, there exists $\mathbf{x}$ with $\frac{1}{2}(m+1)$ non-zero elements yet does not admit super-resolution; (c) If $\mathbf{x}$ has more than $m$ non-zero elements then $\vect x$ does not admit super-resolution. Their definition and results focused on the possibility of overcoming Rayleigh limit in the presence of sufficient small noise and hence demonstrated the possibility of super-resolution. See also \cite{fuchs2005sparsity} for a shorter exposition of the same idea. 

In recent years, some researchers analyzed the stability of specific super-resolution algorithms in a non-asymptotic regime \cite{morgenshtern2016super, morgenshtern2020super, denoyelle2017support}. To be more specific, it was shown in \cite{morgenshtern2016super} that a simple convex optimization program can already superresolve the positive sources (on a grid) to nearly optimal. The authors of 
\cite{morgenshtern2016super} demonstrated that under certain conditions the deviation between the algorithm's output $\hat{\mathbf{x}}$ and the ground truth $\mathbf{x}$ obeys the following relation 
\[
\|\hat{\mathbf{x}}-\mathbf{x}\|_1  \approx C \cdot \mathrm{SRF}^{2 r} \cdot\sigma,
\]
where $r$ is the Rayleigh index. The theory was later generalized to the off-the-grid setting in \cite{morgenshtern2020super} where the authors analyzed the stability of the reconstruction of high frequency information. In a different line of research, the authors studied in \cite{denoyelle2017support} the amplitude and support recoveries of positive discrete measures for a so-called BLASSO convex program. They demonstrated that when $\sigma/ \lambda, \sigma/ d_{\min}^{2n-1}$ and $\lambda / d_{\min}^{2n-1}$ are sufficiently small (with $\lambda$ being the regularization parameter, $\sigma$ the noise level, $n$ the source number, $d_{\min}$ the minimum separation distance between two sources), there exists a unique solution to the BLASSO program consisting of exactly $n$ point sources. The amplitudes and locations of the solution both converge toward those of the ground truth when the noise and the regularization parameter decay to zero faster than $d_{\min}^{2n-1}$. Note that this result is consistent with our estimates in the current paper, showing that the BLASSO achieves the optimal resolution order, which is quite impressive. We also refer the readers to many other algorithms \cite{bendory2017robust, eftekhari2021sparse, kurmanbek2022multivariate, garcia2021approximate, da2020compressed} leveraging the a priori knowledge of positivity and especially the well-known maximum entropy method \cite{gull1978image, donoho1992maximum}.

\subsection{Main contribution}
The main contribution of this paper is quantitative characterizations of the resolution limits to number detection and location recovery in the super-resolution of positive sources. Accurate detection of the source number (model order) is important in the super-resolution problem and many parametric estimation methods require the model order as a priori information. But there are few theoretical results which address the issue when the number of underlying sources is greater than two. In \cite{liu2021mathematicaloned, liu2021theorylse}, the first results for capacity of the number detection in the super-resolution of general sources are derived. Here we generalize the estimates to the case of positive sources, which is also the first result for understanding the capacity of super-resolving positive sources. Specifically, we introduce the computational resolution limit $\mathcal{D}_{num}^+$ for the detection of $n$ point sources (see Definition \ref{defi:computresolutionlimit}), and derive the following sharp bounds:  
\begin{equation}\label{equ:numbcrlbounds}
\frac{2e^{-1}}{\Omega}\Big(\frac{\sigma}{m_{\min}}\Big)^{\frac{1}{2n-2}}<\mathcal{D}_{num}^+\leq  \frac{2e\pi}{\Omega }\Big(\frac{\sigma}{m_{\min}}\Big)^{\frac{1}{2n-2}},
\end{equation}
where $\frac{\sigma}{m_{\min}}$ is viewed as the inverse of the signal-to-noise ration (SNR). It follows that exact detection of the source number is possible when the minimum separation distance of point sources $d_{\min}$ is greater than $\frac{2e\pi}{\Omega }\Big(\frac{\sigma}{m_{\min}}\Big)^{\frac{1}{2n-2}}$,  and impossible  without additional a priori information when $d_{\min}$ is less than $\frac{2e^{-1}}{\Omega}\Big(\frac{\sigma}{m_{\min}}\Big)^{\frac{1}{2n-2}}$ in the worst-case scenario.  

Following the same line of argument for the number detection problem, we also consider the location recovery in the super-resolution of positive sources. We introduce the computational resolution limit $\mathcal D_{supp}^{+}$ for the support recovery (see Definition \ref{defi:computresolutionlimit2}) and derive the following bounds:
\begin{equation}\label{equ:suppcrlbounds}
\frac{2e^{-1}}{\Omega} \Big(\frac{\sigma}{m_{\min}}\Big)^{\frac{1}{2n-1}}<\mathcal{D}_{supp}^{+}\leq \frac{2.36e\pi }{\Omega }\Big(\frac{\sigma}{m_{\min}}\Big)^{\frac{1}{2n-1}}.
\end{equation}
As a consequence, the resolution limit $\mathcal D_{supp}^+$ is of the order $O(\frac{1}{\Omega}\Big(\frac{\sigma}{m_{\min}}\Big)^{\frac{1}{2n-1}})$. It follows that stable recovery (in certain sense) of the source locations is possible when the minimum separation distance of point sources $d_{\min}$ is greater than $\frac{2.36e\pi}{\Omega }\Big(\frac{\sigma}{m_{\min}}\Big)^{\frac{1}{2n-1}}$,  and impossible without additional a priori information when $d_{\min}$ is less than $\frac{2e^{-1}}{\Omega}\Big(\frac{\sigma}{m_{\min}}\Big)^{\frac{1}{2n-1}}$ in the worst-case scenario. To further emphasize that the separation distance $O\left(\frac{1}{\Omega}\Big(\frac{\sigma}{m_{\min}}\Big)^{\frac{1}{2n-1}}\right)$ is necessary for a stable location reconstruction, we construct an example showing that if the sources are separated below the $\frac{c}{\Omega}\Big(\frac{\sigma}{m_{\min}}\Big)^{\frac{1}{2n-1}}$ for certain constant $c$, the recovered locations can be very unstable. 

As a direct consequence of our estimates, we analyze the stability for a sparsity-promoting algorithm ($l_0$ minimization) in super-resolving positive sources and show that it already achieves the optimal order of the resolution. These estimates for the resolution limits are also generalized to multi-dimensional spaces.   

The quantitative characterizations of the resolution limits $\mathcal{D}_{num}^+$ and $\mathcal{D}_{supp}^+$ imply phase transition phenomena in the corresponding reconstructions, which are confirmed here by numerical experiments.

In addition, our results reveal that a priori knowledge of the positivity of the source does not improve the resolution limit order compared to resolving complex sources \cite{liu2021theorylse,liu2022mathematical}. The question of whether positivity can indeed enhance the resolution limit arises. The answer is no. Specifically, we demonstrate that under certain noise level, the computational resolution limit for distinguishing two sources (number detection) with a phase difference $\theta$ (in amplitude) is 
	\begin{equation*}
	\mathcal D_{k, num} =  \frac{4\arcsin\left(\left(\frac{\sigma}{m_{\min}}\right)^{\frac{1}{2}}\right)-\theta}{\Omega}.
\end{equation*}
This shows that the positivity of two sources actually deteriorates the resolution limit rather than enhancing it. As another discovery, we see that achieving super-resolution in distinguishing images generated from one or two sources is quite possible, especially when the source amplitudes differ in phases. 

On the other hand, our techniques offer a way to analyze the capability of resolving positive sources, which could inspire future work.

\subsection{Organization of the paper}
The paper is organized in the following way. Section 2 presents the estimates for the resolution limit in the one-dimensional super-resolution of positive sources. Section 3 extends the estimates to multi-dimensional spaces.  In Section 4, we discuss if the positivity can indeed enhance the resolution limit.  In Sections 5 and 6, we verify the phase transition in respectively the number detection and location recovery problems. The purpose of Section 7 is to make a few final remarks. Section 8 and Section 9 prove respectively the results in Section 2 and Section 4. In Appendix \ref{appendixA}, we prove several auxiliary lemmas and useful inequalities.

\section{Resolution limits for super-resolution in one-dimensional space}\label{section:onedresolutionlimit}
We present in this section our main results on the resolution limit for the super-resolution of one-dimensional positive sources. All the results shall be proved in Section \ref{section:proofofmainresult}.
We consider the case when the point sources are tightly spaced and form a cluster. To be more specific, we define the interval 
$$
I(n,\Omega):=\Big(-\frac{(n-1)\pi}{2\Omega}, \quad  \frac{(n-1)\pi}{2\Omega}\Big),
$$ 
which is of length of several Rayleigh limits and assume that $y_j\in I(n,\Omega), 1\leq j\leq n$. The reconstruction process is usually targeting at some specific solutions in a so-called admissible set, which comprises of discrete measures whose Fourier data are sufficiently close to $\vect Y$. In our problem, we introduce the following concept of positive $\sigma$-admissible discrete measures. We denote in this section $\bnorm{f}_{\infty} = \max_{\omega\in[-\Omega, \Omega]}|f(\omega)|$. 

\begin{defi}{\label{determinecriterion1}}
	Given measurement $\mathbf Y$, we say that $\hat \mu=\sum_{j=1}^{k} \hat a_j \delta_{\hat y_j}, \hat a_j>0$ is a positive $\sigma$-admissible discrete measure of $\ \mathbf Y$ if
	\[
	\bnorm{\mathcal F[\hat \mu]-\mathbf Y}_{\infty}< \sigma.
	\]
\end{defi}

The set of positive $\sigma$-admissible measures of $\mathbf Y$ characterizes all possible solutions to our super-resolution problem with the given measurement $\mathbf Y$. Following similar definitions in \cite{liu2021theorylse, liu2021mathematicaloned, liu2021mathematicalhighd, liu2022mathematical}, we define the following computational resolution limit for the number detection in the super-resolution of positive sources. The reason for the definition is the fact that detecting the correct source number in $\mu$ is impossible without additional a prior information when there exists one positive $\sigma$-admissible measure with less than $n$ supports.  

\begin{defi}\label{defi:computresolutionlimit}
The computational resolution limit to the number detection problem in the super-resolution of one-dimensional positive source is defined as the smallest nonnegative number $\mathcal D_{num}^{+}$ such that for all positive $n$-sparse measure $\sum_{j=1}^{n}a_{j}\delta_{y_j}, a_j>0, y_j \in I(n, \Omega)$ and the associated measurement $\vect Y$ in (\ref{equ:modelsetting1}), if 
	\[
	\min_{p\neq j} |y_j-y_p| \geq \mathcal D_{num}^{+},
	\]
then there does not exist any positive $\sigma$-admissible measure of $\ \mathbf Y$ with less than $n$ supports.
\end{defi}

The notion of ``computational resolution limit'' emphasizes the essential impossibility of correct number detection for very close source by any means. Also, this notion depends crucially on the signal-to-noise ratio and the sparsity of the source, which is different from all classical resolution limits \cite{abbe1873beitrage, volkmann1966ernst, rayleigh1879xxxi, schuster1904introduction, sparrow1916spectroscopic} that depend only on the cutoff frequency. We now present sharp bounds for this computational resolution limit $\mathcal D_{num}^{+}$. The following upper bound for it is a direct consequence of  \cite[Theorem 3.1]{liu2022mathematical}. 

\begin{thm}\label{thm:upperboundnumberlimithm0}
	Let $\mathbf Y$ be a measurement generated by a positive measure $\mu =\sum_{j=1}^{n}a_j\delta_{y_j}$, which is supported on $I(n,\Omega)$. Let $n\geq 2$ and assume that  the following separation condition is satisfied 
	\begin{equation}\label{upperboundnumberlimithm0equ0}
	\min_{p\neq j}\Big|y_p-y_j\Big|\geq \frac{2 e \pi}{\Omega }\Big(\frac{\sigma}{m_{\min}}\Big)^{\frac{1}{2n-2}}.
	\end{equation}
	Then there do not exist any positive $\sigma$-admissible measures of \,$\mathbf Y$ with less than $n$ supports.
\end{thm}
Theorem \ref{thm:upperboundnumberlimithm0} gives an upper bound for the computational resolution limit $\mathcal D_{num}^{+}$. This upper bound is shown to be tight for the super-resolution of general discrete source (not positive) by a lower bound derived in \cite{liu2021theorylse}, but the result is unknown for the case of resolving positive sources. We next present a lower bound of $\mathcal D_{num}^{+}$ which is the main result of this paper. 
\begin{thm}\label{thm:numberlowerboundthm0}
	For given $0<\sigma\leq m_{\min}$ and integer $n\geq 2$, there exist positive measures $\mu=\sum_{j=1}^{n}a_j\delta_{y_j}$ with $n$ supports and $\hat \mu=\sum_{j=1}^{n-1}\hat a_j \delta_{\hat y_j}$ with $(n-1)$ supports such that 
	$\bnorm{\mathcal F[\hat \mu]-\mathcal F [\mu]}_{\infty}< \sigma$. Moreover,
	\[
	\min_{1\leq j\leq n}\babs{a_j}= m_{\min}, \quad \min_{p\neq j}\babs{y_p-y_j}= \frac{2e^{-1}}{\Omega}\Big(\frac{\sigma}{m_{\min}}\Big)^{\frac{1}{2n-2}}.
	\]
\end{thm}
The above result gives a lower bound for the computational resolution limit $\mathcal D_{num}^{+}$ to the number detection problem. Combined with Theorem \ref{thm:upperboundnumberlimithm0}, it reveals that the computational resolution limit for number detection satisfies 
\[
\frac{2e^{-1}}{\Omega}\Big(\frac{\sigma}{m_{\min}}\Big)^{\frac{1}{2n-2}}<\mathcal D_{num}^{+}\leq \frac{2 e\pi }{\Omega}\Big(\frac{\sigma}{m_{\min}}\Big)^{\frac{1}{2n-2}}. 
\]

We remark that similar to the results of \cite{akinshin2015accuracy, batenkov2019super,liu2021theorylse}, our bounds are the worst-case bounds, and one may achieve better bounds for the case of random noise.  

\medskip	
We now consider the location (support) recovery problem in the super-resolution of positive sources. We first introduce the following concept of $\delta$-neighborhood of a discrete measure. 
\begin{defi}\label{defi:deltaneighborhood}
	Let  $\mu=\sum_{j=1}^{n}a_j \delta_{y_j}$ be a discrete measure and let $0<\delta$ be such that the $n$ intervals $(y_k- \delta, y_k + \delta), 1\leq k \leq n$ are pairwise disjoint. We say that 
	$\hat \mu=\sum_{j=1}^{n}\hat a_{j}\delta_{\hat y_j}$ is within $\delta$-neighborhood of $\mu$ if each $\hat y_j$ is contained in one and only one of the n intervals $(y_k- \delta, y_k + \delta), 1\leq k \leq n$. 
\end{defi}

According to the above definition, a measure in a $\delta$-neighbourhood preserves the inner structure of the real source. For any stable support recovery algorithm, the output should be a measure in some $\delta$-neighborhood, otherwise it is impossible to distinguish which is the reconstructed location of some $y_j$'s. We now introduce the computational resolution limit for stable support recoveries. For ease of exposition, we only consider measures supported in $I(n, \Omega)$, where $n$ is the number of supports. 

\begin{defi}\label{defi:computresolutionlimit2}
The computational resolution limit to the stable support recovery problem in the super-resolution of one-dimensional positive sources is defined as the smallest nonnegative number $\mathcal D_{supp}^{+}$ such that for all positive $n$-sparse measures $\sum_{j=1}^{n}a_{j}\delta_{y_j}, a_j>0, y_j \in I(n, \Omega)$ and the associated measurement $\vect Y$ in (\ref{equ:modelsetting1}), if 
	\[
	\min_{p\neq j} |y_j-y_p| \geq \mathcal D_{supp}^{+},
	\]
	then there exists $\delta>0$ such that any positive $\sigma$-admissible measure for $\mathbf Y$ with $n$ supports in $I(n, \Omega)$ is within a $\delta$-neighbourhood of $\mu$.  
\end{defi}

To state the results on the resolution limit to stable support recovery, we introduce the super-resolution factor which is defined as the ratio between Rayleigh limit $\frac{\pi}{\Omega}$ (for point spread function $sinc(x)^2$) and the minimum separation distance of sources $d_{\min}:= \min_{p\neq j}|y_p-y_j|$:
\[
SRF:= \frac{\pi}{\Omega d_{\min}}.
\]
As a direct consequence of  \cite[Theorem 3.2]{liu2022mathematical}, we have the following theorem giving the upper bound of $\mathcal D_{supp}^+$.
\begin{thm}\label{thm:upperboundsupportlimithm0}
	Let $n\geq 2$, assume that the positive measure $\mu=\sum_{j=1}^{n}a_j \delta_{y_j}$ is supported on $I(n, \Omega)$ and that 
	\begin{equation}\label{supportlimithm0equ0}
	\min_{p\neq j}\left|y_p-y_j\right|\geq \frac{2.36e\pi }{\Omega }\Big(\frac{\sigma}{m_{\min}}\Big)^{\frac{1}{2n-1}}. \end{equation}
	If $\hat \mu=\sum_{j=1}^{n}\hat a_{j}\delta_{\hat y_j}, \hat a_j>0$ supported on $I(n,\Omega)$ is a positive $\sigma$-admissible measure for the measurement generated by $\mu$, then $\hat \mu$ is within the $\frac{d_{\min}}{2}$-neighborhood of $\mu$. Moreover, after reordering the $\hat y_j$'s, we have 
	\begin{equation}\label{supportlimithm0equ2}
	\left|\hat{y}_j-y_j\right|<\frac{C(n)}{\Omega} S R F^{2 n-2} \frac{\sigma}{m_{\min }}, \quad 1 \leq j \leq n,
	\end{equation}
	where $C(n)=2^{2 n-\frac{3}{2}} e^{2 n-1}(\max (\sqrt{n-2}, 1) \pi)^{-\frac{1}{2}}$. 
\end{thm}

Theorem \ref{thm:upperboundsupportlimithm0} gives an upper bound to the computational resolution limit $\mathcal{D}_{supp}^{+}$. We next show that the order of the upper bound is optimal.

\begin{thm}\label{thm:supportlowerboundthm0}
	For given $0<\sigma\leq m_{\min}$ and integer $n\geq 2$, let 
	\begin{equation}\label{supportlowerboundequ0}
	\tau=\frac{e^{-1}}{\Omega}\ \Big(\frac{\sigma}{m_{\min}}\Big)^{\frac{1}{2n-1}}.
	\end{equation}
	Then there exist a positive measure $\mu=\sum_{j=1}^{n}a_j \delta_{y_j}$ with $n$ supports at $\{-(n-\frac{1}{2})\tau, -(n-\frac{5}{2})\tau, \cdots, (n-\frac{3}{2})\tau\}$ and a positive measure $\hat \mu=\sum_{j=1}^{n}\hat a_j \delta_{\hat y_j}$ with $n$ supports at  $\{-(n-\frac{3}{2})\tau, -(n-\frac{7}{2})\tau,\cdots, (n-\frac{1}{2})\tau\}$ such that
	\[
	||\mathcal F[\hat \mu]- \mathcal F[\mu]||_{\infty}< \sigma, \quad \min_{1\leq j\leq n}|a_j|= m_{\min}.
	\]
\end{thm}

Since the minimum distance between $y_j$'s is $d_{\min} = 2\tau$, thus for the positive $\sigma$-admissible measure $\hat \mu$, it is obviously that the $\hat y_j$'s are not in any $\delta$-neighborhood of $y_j$'s for $\delta\leq \frac{d_{\min}}{2}$ (for $\delta>\frac{d_{\min}}{2}$ the intervals in Definition \ref{defi:deltaneighborhood} are overlapped). According to Definition \ref{defi:computresolutionlimit2}, Theorem \ref{thm:supportlowerboundthm0} implies $\mathcal D_{supp}^{+}> \frac{2e^{-1}}{\Omega}\ \Big(\frac{\sigma}{m_{\min}}\Big)^{\frac{1}{2n-1}}$. Thus we conclude that 
\[
\frac{2e^{-1}}{\Omega}\ \Big(\frac{\sigma}{m_{\min}}\Big)^{\frac{1}{2n-1}} < \mathcal D_{supp}^{+}\leq \frac{2.36 e\pi }{\Omega}\ \Big(\frac{\sigma}{m_{\min}}\Big)^{\frac{1}{2n-1}}.
\]
To further demonstrate that the order $O\left(\frac{1}{\Omega}\left(\frac{\sigma}{m_{\min}}\right)^{\frac{1}{2n-1}}\right)$ is essentially optimal for stable location reconstruction, we present an example with a new distribution of the source locations as  follows.  

\begin{prop}\label{thm:supportlowerboundthm1}
    For given $0<\sigma<m_{\min}$ and integer $n\geq 2$, let 
	\begin{equation}\label{supportlowerboundequ1}
	\tau=\frac{0.2e^{-1}}{\Omega s^{\frac{2n+1}{2n-1}}}\ \Big(\frac{\sigma}{m_{\min}}\Big)^{\frac{1}{2n-1}}.
	\end{equation}
	Then there exist a positive measure $\mu=\sum_{j=1}^{n}a_j \delta_{y_j}$ with $n$ supports at $\{t_j=-\frac{sn-2}{2}\tau+\frac{(j-2)s}{2}\tau, j=2,4, \cdots, 2n\}$ and a positive measure $\hat \mu=\sum_{j=1}^{n}\hat a_j \delta_{\hat y_j}$ with $n$ supports at $\{t_j=t_{4\lceil\frac{j+1}{4}\rceil-2}+(-1)^{\frac{j+1}{2}}\tau, j =1,3,5,\cdots, 2n-1\}$ such that
	\[
	||\mathcal F[\hat \mu]- \mathcal F[\mu]||_{\infty}< \sigma, \quad \min_{1\leq j\leq n}|a_j|= m_{\min}.
	\]
\end{prop}

Note that the $n$ underlying sources in $\mu$ are spaced by 
\[
s\tau =  \frac{0.4e^{-1}}{\Omega s^{\frac{2}{2n-1}}}\ \Big(\frac{\sigma}{m_{\min}}\Big)^{\frac{1}{2n-1}}.
\]
Proposition \ref{thm:supportlowerboundthm1} reveals that when the $n$ point sources are separated by $\frac{c}{\Omega} \Big(\frac{\sigma}{m_{\min}}\Big)^{\frac{1}{2n-1}}$ for some constant $c$, the recovered source locations from the $\sigma$-admissible measures can be very unstable; see Figure \ref{fig:examplelocation}. 

\begin{figure}[!h]
	\includegraphics[width=0.7\textwidth]{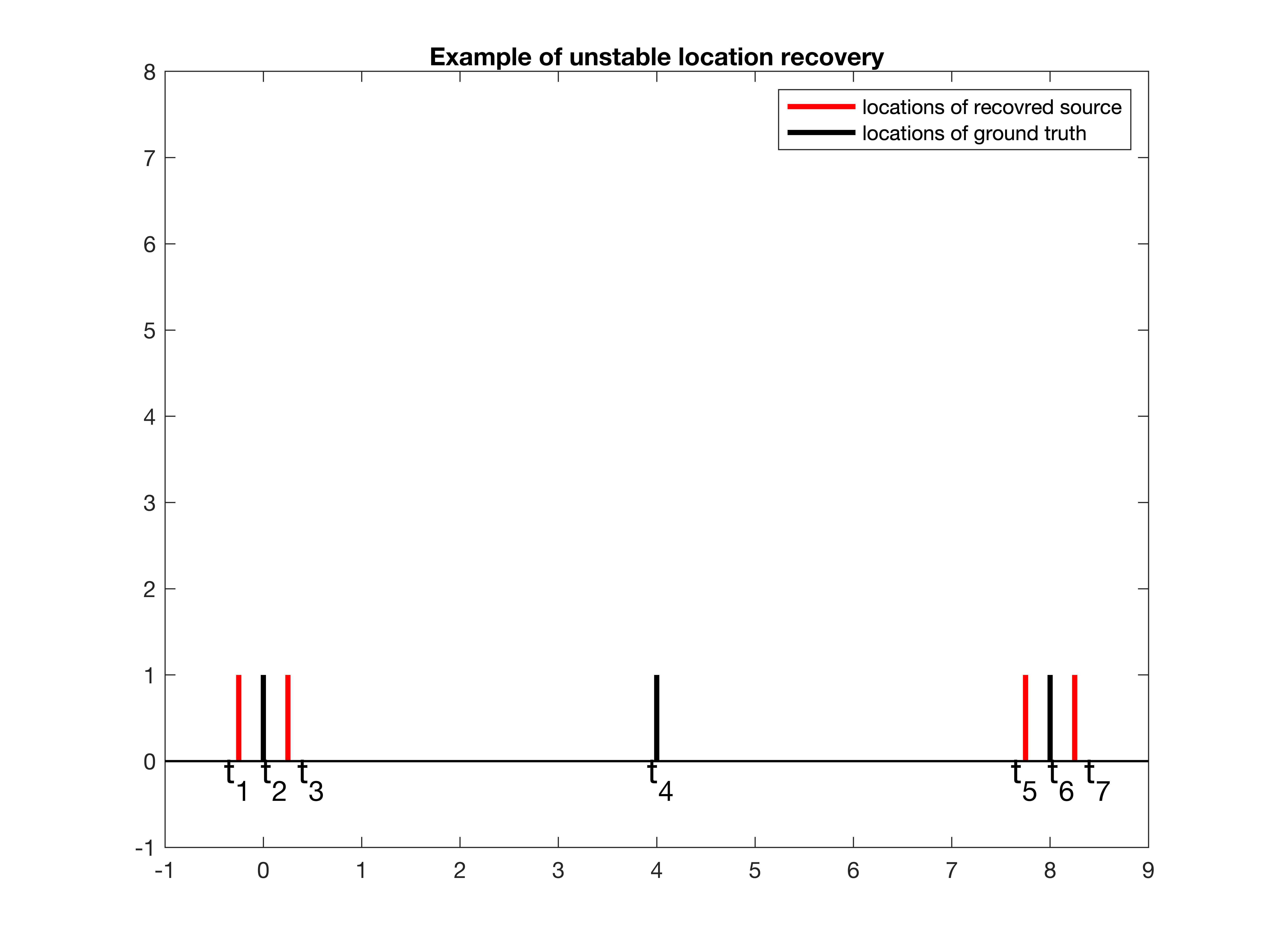}
	\centering
	\caption{An example of unstable location recovery.}
	\label{fig:examplelocation}
\end{figure}

\begin{remark}
Note that all of our results hold for the case when the sources are supported on a grid. Specifically, we consider the grid points $t_j=j\Delta, j=1,\cdots,N,$ where $N$ and $\Delta$ are the number and spacing of grid points, respectively, and assume the sources are supported on the grid. Assume also that the grid spacing $\Delta\leq \frac{e^{-1}}{\Omega}\left(\frac{\sigma}{m_{\min}}\right)^{\frac{1}{2n-1}}$ for fixed $n$ and $m_{\min}$. By Theorem \ref{thm:supportlowerboundthm0}, we can construct two positive measures $\mu = \sum_{q=1}^n a_q \delta_{t_{j_q}}$ and $\hat \mu = \sum_{q=1}^n \hat a_{q} \delta_{\hat t_{j_q}}$ supported on the grid with completely different supports such that the difference of their Fourier data is less than the noise level and the minimum separation of sources is equal or less than $\frac{2e^{-1}}{\Omega}\left(\frac{\sigma}{m_{\min}}\right)^{\frac{1}{2n-1}}$ with $\min_{q=1,\cdots,n}a_q =m_{\min}$.  
\end{remark}

\begin{remark}
Note that our estimates for both the resolution limits in the number detection and support recovery already improve the estimates in \cite{liu2021theorylse} for the case of general sources. 
\end{remark}

\subsection{Stability analysis of sparsity-promoting algorithms}
Nowadays, sparsity-promoting algorithms are popular methods in image processing, signal processing and many other fields. By our results for the resolution limits, we can derive a sharp stability result for the $l_0$ minimization in the super-resolution of positive sources. We consider the following $l_0$-minimization problem:
\begin{equation}\label{equ:l0minimization}
\min_{\text{$\rho$ supported on $\mathcal O$, $\rho$ is a positive discrete measure}} \bnorm{\rho}_{0} \quad \text{subject to} \quad |\mathcal F\rho(\omega) -\vect Y(\omega)|< \sigma, \quad \omega \in[-\Omega,\Omega], 
\end{equation}	
where $||\rho||_{0}$ is the number of Dirac masses representing the discrete measure $\rho$. As a corollary of Theorems \ref{thm:upperboundnumberlimithm0} and \ref{thm:upperboundsupportlimithm0}, we have the following theorem for its stability.

\begin{thm}\label{thm:sparspromstabilitythm0}
	Let $n\geq 2$ and $\sigma \leq m_{\min}$. Let the measurement $\vect Y$ in (\ref{equ:modelsetting1}) be generated by a positive $n$-sparse measure $\mu=\sum_{j=1}^{n}a_j \delta_{y_j}, y_j \in I(n,\Omega)$. Assume that
	\begin{equation}\label{equ:sparspromstabilitythm0equ1}
	d_{\min}:=\min_{p\neq j}\left| y_p-y_j\right|\geq \frac{2.36 e\pi }{\Omega }\Big(\frac{\sigma}{m_{\min}}\Big)^{\frac{1}{2n-1}}. \end{equation}
	Let $\mathcal O$ in the minimization problem (\ref{equ:l0minimization}) be (or be included in) $I(n,\Omega)$ , then the solution to (\ref{equ:l0minimization}) contains exactly $n$ point sources. For any solution $\hat \mu=\sum_{j=1}^{n}\hat a_{j}\delta_{\hat y_j}$, it is in a $\frac{d_{\min}}{2}$-neighborhood of $\mu$. Moreover, after reordering the ${\hat y}_j$'s, we have 
	\begin{equation}
	\left|\hat y_j-y_j\right|\leq \frac{C(n)}{\Omega}SRF^{2n-2}\frac{\sigma}{m_{\min}}, \quad 1\leq j\leq n,
	\end{equation}
	where $C(n)=C(n)=2^{2 n-\frac{3}{2}} e^{2 n-1}(\max (\sqrt{n-2}, 1) \pi)^{-\frac{1}{2}}$. 
\end{thm}

Theorem \ref{thm:sparspromstabilitythm0} reveals that sparsity promoting over admissible solutions can resolve the source locations to the resolution limit level. It provides an insight that theoretically sparsity-promoting algorithms would have excellent performance on the super-resolution of positive sources, which already have been corroborated by \cite{morgenshtern2016super, morgenshtern2020super, denoyelle2017support}. Especially, under the separation condition (\ref{equ:sparspromstabilitythm0equ1}), any tractable sparsity-promoting algorithms (such as total variation minimization algorithms \cite{candes2014towards}) rendering the sparsest solution
could stably reconstruct  all the source locations.

\section{Resolution limits for super-resolution in multi-dimensional spaces}
In this section, combining the estimates in Section \ref{section:onedresolutionlimit} and \cite{liu2021mathematicalhighd, liu2022nearly}, we present our main results on the resolution limits to the super-resolution of positive sources in multi-dimensional spaces. Let us first introduce the model setting. We consider the source as the $n$-sparse positive measure 
\[
\mu=\sum_{j=1}^{n}a_{j}\delta_{\vect y_j},
\]
where $\delta$ denotes Dirac's $\delta$-distribution in $\mathbb R^k$, $\vect y_j \in \mathbb R^k,1\leq j\leq n$, represent the locations of the point sources and $a_j>0, 1\leq j \leq n$ are their amplitudes. Denote by
\begin{equation}\label{equ:intendisset}
m_{\min}=\min_{j=1,\cdots,n}|a_j|,
\quad 	d_{\min}=\min_{p\neq j}||\vect y_p-\vect y_j||_2.
\end{equation}
The available measurement is the noisy Fourier data of $\mu$ in a bounded region, that is, 
\begin{equation}\label{equ:modelsetting2}
\mathbf Y(\vect{\omega}) = \mathcal F [\mu] (\vect{\omega}) + \mathbf W(\vect{\omega})= \sum_{j=1}^{n}a_j e^{i \vect{y}_j\cdot \vect{\omega}} + \mathbf W(\vect{\omega}), \ \vect \omega \in \mathbb R^k,
\end{equation}
where with slight abuse of notation $\mathcal F \mu$ denotes the Fourier transform of $\mu$ in the $k$-dimensional space, 
$\Omega$ is the cut-off frequency, and $\mathbf W$ is the noise. We assume that 
\[
 ||\mathbf W(\vect{\omega})||_\infty< \sigma,\ ||\vect{\omega}||_2\leq \Omega,
\]
where $\sigma$ is the noise level and $\bnorm{f(\vect \omega)}_{\infty}:=\max_{\vect \omega\in\mathbb R^k, ||\vect \omega||_2\leq \Omega}|f(\vect \omega)|$ in this section. We are interested in the resolution limit for resolving a cluster of tightly-spaced point sources. Thus, we denote by
\[
B_{\delta}^k(\vect x) := \Big\{ \mathbf y \ \Big|\ \mathbf y\in \mathbb R^k,\  ||\vect y- \vect x||_2<\delta \Big\},
\] 
and assume that $\vect y_j \in B_{\frac{(n-1)\pi}{2\Omega}}^k(\vect 0), j=1,\cdots,n$, or equivalently $||\vect y_j||_2<\frac{(n-1)\pi}{2\Omega}$.  

We then define positive $\sigma$-admissible measures and computational resolution limits in the $k$-dimensional space analogously to those in the one-dimensional case. 

\begin{defi}{\label{def:sigmaadmissiblemeasure}}
	Given measurement $\mathbf Y$, we say that the positive measure $\hat \mu=\sum_{j=1}^{m} \hat a_j \delta_{ \mathbf{\hat y}_j}, \ \mathbf{\hat y}_j\in \mathbb R^k$, is a positive $\sigma$-admissible discrete measure of $\ \mathbf{Y}$ if
	\begin{equation*}
	\bnorm{\mathcal F[\hat \mu] (\vect{\omega})-\vect Y(\vect{\omega})}_\infty< \sigma, \ \text{for all}\ ||\vect{\omega}||_2\leq \Omega,\  \vect \omega \in \mathbb R^k.
	\end{equation*}
In particular, without the constraint on the positivity of the amplitudes, $\hat \mu$ is called a $\sigma$-admissible discrete measure of $\ \mathbf{Y}$.
\end{defi}

\begin{defi}\label{def:computresolutionlimitnumber}
The computational resolution limit to the number detection problem in $k$-dimensional space is defined as the smallest nonnegative number $\mathcal D_{k,num}^+$ such that for all positive $n$-sparse measures $\sum_{j=1}^{n}a_{j}\delta_{\mathbf y_j}, \vect y_j \in B_{\frac{(n-1)\pi}{2\Omega}}^{k}(\vect 0)$ and the associated measurement $\vect Y$ in (\ref{equ:modelsetting2}), if 
	\[
	\min_{p\neq j} ||\mathbf y_j-\mathbf y_p||_2 \geq \mathcal D_{k, num}^+,
	\]
then there does not exist any positive $\sigma$-admissible measure with less than $n$ supports for $\mathbf Y$. Note that if we remove the constraint of positivity on the source and the $\sigma$-admissible measure, the computational resolution limit is denoted by $\mathcal D_{k, num}$. 
\end{defi}

As a consequence of  \cite[Theorem 2.3]{liu2021mathematicalhighd}, we have the following result for the upper bound of the $\mathcal D_{k, num}^+$. 

\begin{thm}\label{thm:highdupperboundnumberlimit0}
	Let $n\geq 2$ and the measurement $\mathbf Y$ in (\ref{equ:modelsetting2}) be generated by a positive $n$-sparse measure $\mu =\sum_{j=1}^{n}a_j\delta_{\mathbf y_j}, \vect y_j \in B_{\frac{(n-1)\pi}{2\Omega}}^{k}(\vect 0)$.  There is a constant $C_{num}(k,n)$ which has an explicit form such that if 
		\begin{equation}\label{equ:highdupperboundnumberlimit1}
		\min_{p\neq j, 1\leq p, j\leq n}\btwonorm{\mathbf y_p- \mathbf y_j}\geq \frac{C_{num}(k,n)}{\Omega }\Big(\frac{\sigma}{m_{\min}}\Big)^{\frac{1}{2n-2}}
	\end{equation}
holds, then there do not exist any positive $\sigma$-admissible measures of \,$\mathbf Y$ with less than $n$ supports.
\end{thm}

We next show that the above upper bound is optimal in terms of the signal-to-noise ratio.	
\begin{thm}\label{thm:highdnumberlowerbound0}
	For given $0<\sigma\leq m_{\min}$ and integer $n\geq 2$, there exist positive measures $\mu=\sum_{j=1}^{n}a_j\delta_{\vect y_j}$ with $n$ supports, and $\hat \mu=\sum_{j=1}^{n-1}\hat a_j \delta_{\mathbf{\hat y}_j}$ with $(n-1)$ supports such that 
	$||\mathcal F[\hat \mu]-\mathcal F [\mu]||_{\infty}< \sigma$. Moreover,
	\[
	\min_{1\leq j\leq n}|a_j|= m_{\min}, \quad \min_{p\neq j}\bnorm{\vect y_p-\vect y_j}_2= \frac{2e^{-1}}{\Omega}\Big(\frac{\sigma}{m_{\min}}\Big)^{\frac{1}{2n-2}}.
	\]
\end{thm}
\begin{proof}
Consider $\gamma =\sum_{j=1}^{2n-1} a_j \delta_{\vect t_j}$ with $\vect t_1 = (-(n-1)\tau, 0, \cdots, 0), \vect t_2 = (-(n-2)\tau, 0, \cdots, 0), \cdots, \vect t_{2n-1}= ((n-1)\tau, 0, \cdots, 0)$ and $\tau = \frac{e^{-1}}{\Omega}\Big(\frac{\sigma}{m_{\min}}\Big)^{\frac{1}{2n-2}}$. For every $\vect \omega =(\omega_1, \omega_2, \cdots, \omega_k)^{\top}$, $$\mathcal F \gamma(\vect \omega) = \sum_{j=1}^{2n-1}a_je^{i \vect t_j \cdot \vect \omega} = \sum_{j=1}^{2n-1}a_je^{i (-n+j)\tau \omega_1}, |\omega_1|\leq \Omega.$$ This reduces the estimation of $\mathcal F \gamma(\vect \omega)$ to the one-dimensional case. Combined with Theorem \ref{thm:numberlowerboundthm0}, there exist $a_{2j-1}>0, a_{2j}<0,  1\leq j\leq n, \min_{j=1, \cdots, n}|a_{2j-1}|= m_{\min}$, so that $\bnorm{\mathcal F \gamma(\vect \omega)}_{\infty}<\sigma$. As a consequence,
\[
\mu=\sum_{j=1}^{n}a_{2j-1} \delta_{\vect t_{2j-1}},\quad  \hat \mu=\sum_{j=1}^{n-1}-a_{2j}\delta_{\vect t_{2j}}
\]
satisfy all the conditions of the theorem. 
\end{proof}

The above results indicate that
\[
\frac{C_{1,k}(n)}{\Omega}\Big(\frac{\sigma}{m_{\min}}\Big)^{\frac{1}{2n-2}} <  \mathcal D_{k,num}^+ \leq \frac{C_{2,k}(n)}{\Omega }\Big(\frac{\sigma}{m_{\min}}\Big)^{\frac{1}{2n-2}}, 
\]
with $C_{1,k}(n), C_{2,k}(n)$ being certain constants. An interesting open problem is to improve these constants. Two of the authors of this paper have made a progress in this direction \cite{liu2022nearly}. 

To state the estimates for the resolution limits to the location recovery, we introduce the following concepts which are analogue to those in the one-dimensional case. 
\begin{defi}\label{deltaneighborhood}
	Let  $\mu=\sum_{j=1}^{n}a_j \delta_{\vect y_j}$ be a positive $n$-sparse discrete measure in $\mathbb R^k$ and let $\delta>0$ be such that the $n$ balls $B_{\delta}^k(\vect y_j), 1\leq j \leq n,$ are pairwise disjoint. We say that 
	$\hat \mu=\sum_{j=1}^{n}\hat a_{j}\delta_{\mathbf{\hat y}_j}$ is within $\delta$-neighborhood of $\mu$ if each $\mathbf {\hat y}_j$ is contained in one and only one of the $n$ balls $B_{\delta}^k(\vect y_j), 1\leq j \leq n$.
\end{defi}

\begin{defi}\label{def:highdcomputresolutionlimitsupport}
	The computational resolution limit to the stable support recovery problem in $k$-dimensional space is defined as the smallest non-negative number $\mathcal D_{k,supp}^+$ such that for
	any positive $n$-sparse measure $\mu=\sum_{j=1}^{n}a_j \delta_{\mathbf y_j}, \vect y_j \in B_{\frac{(n-1)\pi}{2\Omega}}^{k}(\vect 0)$ and the associated measurement $\vect Y$ in (\ref{equ:modelsetting2}), 
	if
	\[
	\min_{p\neq j, 1\leq p,j \leq n} \btwonorm{\mathbf y_p-\mathbf y_j}\geq \mathcal{D}_{k,supp}^+,
	\]  
	then there exists $\delta>0$ such that any $\sigma$-admissible measure of $\mathbf Y$ with $n$ supports in $B_{\frac{(n-1)\pi}{2\Omega}}^k(\mathbf 0)$ is within a $\delta$-neighbourhood of $\mu$.  
\end{defi}

As a consequence of  \cite[Theorem 2.7]{liu2021mathematicalhighd}, we have the following result on the characterization of $\mathcal{D}_{k,supp}^+$.

\begin{thm}\label{thm:highdupperboundsupportlimit0}
	Let $n\geq 2$. Let the measurement $\vect Y$ in (\ref{equ:modelsetting2}) be generated by a positive $n$-sparse measure $\mu=\sum_{j=1}^{n}a_j \delta_{\vect y_j}, \vect y_j \in B_{\frac{(n-1)\pi}{2\Omega}}^{k}(\vect 0)$ in the $k$-dimensional space. There is a constant $C_{supp}(k,n)$ which has an explicit form such that if 
	\begin{equation}\label{equ:highdsupportlimithm0equ0}
	d_{\min}:=\min_{p\neq j}\Big|\Big|\mathbf y_p-\mathbf y_j\Big|\Big|_2\geq \frac{C_{supp}(k,n)}{\Omega }\Big(\frac{\sigma}{m_{\min}}\Big)^{\frac{1}{2n-1}}
	\end{equation}
 holds, then  for any $\hat \mu=\sum_{j=1}^{n}\hat a_{j}\delta_{\mathbf{\hat y}_j}, \mathbf {\hat y}_j \in B_{\frac{(n-1)\pi}{2\Omega}}^{k}(\vect 0)$ being a positive $\sigma$-admissible measure of $\vect Y$,  $\hat \mu$ is within the $\frac{d_{\min}}{2}$-neighborhood of $\mu$. Moreover, after reordering the $\mathbf{\hat y}_j$'s, we have 
	\begin{equation}\label{equ:highdsupportlimithm0equ1}
	\Big|\Big|\mathbf {\hat y}_j-\mathbf y_j\Big|\Big|_2\leq \frac{C(k, n)}{\Omega}SRF^{2n-2}\frac{\sigma}{m_{\min}}, \quad 1\leq j\leq n,
	\end{equation}
	where $SRF:=\frac{\pi}{\Omega}$ is the super-resolution factor and $C(k,n)$ has an explicit form. 
\end{thm}

Theorem \ref{thm:highdupperboundsupportlimit0} gives an upper bound for the computational resolution limit for the stable support recovery in the $k$-dimensional space. This bound is optimal in terms of the order of the signal-to-noise ratio, as is shown by the theorem below. 

\begin{thm}\label{thm:highdsupportlowerboundthm0}
	For given $0<\sigma\leq m_{\min}$ and integer $n\geq 2$, let 
	\begin{equation}\label{equ:highdsupportlowerboundequ0}
	\tau=\frac{e^{-1}}{\Omega}\ \Big(\frac{\sigma}{m_{\min}}\Big)^{\frac{1}{2n-1}}.
	\end{equation}
	Then there exist a positive measure $\mu=\sum_{j=1}^{n}a_j \delta_{\vect y_j}, \vect y_j \in \mathbb R^k,$ with $n$ supports at $\{(-(n-\frac{1}{2})\tau, 0, \cdots, 0), (-(n-\frac{5}{2})\tau, 0, \cdots, 0), \cdots, ((n-\frac{3}{2})\tau, 0, \cdots, 0)\}$ and a positive measure $\hat \mu=\sum_{j=1}^{n}\hat a_j \delta_{\mathbf{\hat y}_j}, \mathbf{\hat y}_j\in \mathbb R^k,$ with $n$ supports at  $\{(-(n-\frac{3}{2})\tau, 0, \cdots, 0), (-(n-\frac{7}{2})\tau, 0, \cdots, 0),\cdots, ((n-\frac{1}{2})\tau, 0,\cdots, 0)\}$ such that
	\[
	\bnorm{\mathcal F[\hat \mu]- \mathcal F[\mu]}_{\infty}< \sigma, \quad \min_{1\leq j\leq n}|a_j|= m_{\min}.
	\]
\end{thm}
\begin{proof}
Similar to the discussions in the proof of Theorem \ref{thm:highdnumberlowerbound0}, the problem can be reduced to the one-dimensional case. Then leveraging Theorem \ref{thm:supportlowerboundthm0} proves the result. 
\end{proof}

Theorem \ref{thm:highdsupportlowerboundthm0} provides a lower bound to the computational resolution limit $\mathcal{D}_{k,supp}^+$. Combined with Theorem \ref{thm:highdupperboundsupportlimit0}, it reveals that
\[
\frac{C_{3,k}(n)}{\Omega}\ \Big(\frac{\sigma}{m_{\min}}\Big)^{\frac{1}{2n-1}}< \mathcal D_{k, supp}^+ \leq \frac{C_{4,k}(n)}{\Omega }\Big(\frac{\sigma}{m_{\min}}\Big)^{\frac{1}{2n-1}}
\]
for certain constants $C_{3,k}(n), C_{4,k}(n)$.

\begin{remark}
Compared to the one-dimensional case in Section \ref{section:onedresolutionlimit}, the upper bounds of multi-dimensional computational resolution limits for the number detection and location recovery in the super-resolution of positive sources has the same dependence on the signal-to-noise ratio and cutoff frequency. Moreover, their dependence on the dimensionality are indicated by the constant factors in the upper bound. We conjecture that the optimal constants may be independent of the source number $n$. Note that the constant factors in the bounds have been improved in \cite{liu2022nearly} to nearly optimal for the two-dimensional case. 
\end{remark}

\section{Does the positivity enhace the resolution limit?}\label{section:positivitynotenhance}

It is indicated in previous theorems that the a priori information of positivity cannot further improve the resolution order as compared to the case of complex sources \cite{liu2021theorylse, liu2022mathematical}. A further question is whether the positivity can indeed improve the resolution limit or not. To answer the question, in this section we demonstrate that, in certain scenarios, positivity actually deteriorates the resolution limit. Specifically, we show that two point sources with different phases can have a better resolution limit than those with the same phase.

We first consider a generalized diffraction limit problem. Note that the classic diffraction limit problem considers distinguishing two positive sources with identical intensities \cite{abbe1873beitrage, rayleigh1879xxxi, schuster1904introduction, sparrow1916spectroscopic}. However,  in order to highlight the effect of the phase difference between two sources, we introduce a generalized diffraction limit that examines the ability to resolve two sources with equal magnitudes but varying phases.

\begin{defi}\label{defi:twopointresolimit}
	The generalized two-point diffraction limit is defined as the largest nonnegative number $\mathcal{R}(\theta)$ such that for all measures $\mu =\sum_{j=1}^{2}a_j\delta_{\vect y_j}, \vect y_j \in \mathbb R^k$ with $|a_1|=|a_2|>0$ and the phase difference of $a_1, a_2$ being $\theta$, if 
	\[
	\bnorm{\vect y_1-\vect y_2}_2 < \mathcal R(\theta),
	\]
	then, for some image $\vect Y$ in the model (\ref{equ:modelsetting2}), it is impossible to determine whether the image $\vect Y$ is generated from one or two sources from the $\sigma$-admissible measures defined in Definition \ref{def:sigmaadmissiblemeasure}. In other words, there exists a $\sigma$-admissible measure of some $\vect Y$ with only one point source. 
\end{defi}

By the above definition, when $\bnorm{\vect y_1-\vect y_2}_2 \geq \mathcal R(\theta)$, one can definitely distinguish two points with amplitudes of the same magnitude and a phase difference $\theta$  from their image.  Conversely, if the separation condition fails to hold, in some cases it is impossible to determine if the image is generated from one or two sources. We have the following theorem for the exact characterization of this generalized two-point diffraction.

\begin{thm} \label{thm:twopointresolution0}
	Consider the collection of two sources $\mu = a_1 \delta_{\vect y_1}+a_2 \delta_{\vect y_2}, y_j \in \mathbb R^k, j=1,2$ with $a_1= m_{\min}e^{i\theta_1}, a_2 =m_{\min}e^{i\theta_2}$. Denote by $\theta = \babs{\theta_1-\theta_2}$ and assume that $\theta \in \left[0, \frac{\pi}{2}\right]$.  When $\sin \left(\frac{\theta}{2}\right)^2\leq \frac{\sigma}{m_{\min}}\leq \frac{1}{2}$, the generalized two-point diffraction limit $\mathcal R(\theta)$ in a space of general dimensionality is given by
	\begin{equation}\label{equ:twopointresolution1}
	\mathcal R(\theta) =  \frac{4\arcsin\left(\left(\frac{\sigma}{m_{\min}}\right)^{\frac{1}{2}}\right)-\theta}{\Omega}.
	\end{equation}
   When $\frac{\sigma}{m_{\min}} >\frac{1}{2}$, no matter what the separation distance is, there are always some $\sigma$-admissible measures of some image $\vect Y$ with only one point source. 
\end{thm}

Theorem \ref{thm:twopointresolution0} precisely characterizes the resolution of two points with the same magnitude but with a phase difference in $[0, \frac{\pi}{2}]$ under certain noise level. In particular, it reveals an interesting fact that sometimes two sources with identical amplitudes have the worst diffraction limit. Therefore, the phase difference actually improves the resolution limit. 

However, Theorem \ref{thm:twopointresolution0} has a specific setting, and it is more crucial to determine the resolution limit for general scenarios. In the ensuing theorem, we investigate the computational resolution limit while considering  phase differences between sources. It is worth noting that the original definition of the computational resolution limit does not include the phase difference. Still, we have utilized the same notion to avoid any confusion.

	\begin{thm}\label{thm:computatwopointresolution0}
		For $\theta \in [0, \frac{\pi}{2}]$ and $\sin \left(\frac{\theta}{2}\right)^2\leq \frac{\sigma}{m_{\min}} \leq \frac{1}{2}$, the resolution limit $\mathcal D_{k, num}$ for resolving two sources with phase difference $\theta$ in $\mathbb R^k$ is given by 
		\begin{equation}\label{equ:computatwopointresolu0}
			\mathcal D_{k, num} =  \frac{4\arcsin\left(\left(\frac{\sigma}{m_{\min}}\right)^{\frac{1}{2}}\right)-\theta}{\Omega}.
		\end{equation}
		It can be attained if $|a_1|=|a_2|$. When $\frac{\sigma}{m_{\min}} >\frac{1}{2}$, no matter what the separation distance is, there are always some $\sigma$-admissible measures of some $\vect Y$ with only one point source.
	\end{thm}

Theorem \ref{thm:computatwopointresolution0} establishes that, under certain noise levels, when there is a phase difference between two sources, their computational resolution limit is strictly better than that of two positive sources; this means that positivity actually impairs the resolution limit rather than improving it. It is also worth noting that, as per Theorem \ref{thm:computatwopointresolution0}, super-resolution can already be achieved to some extent provided that $\frac{\sigma}{m_{\min}}< \frac{1}{2}$, especially when the source amplitudes differ in phase.

\section{Phase transition in the number detection}\label{section:numberphasetrasition}
In this section, employing the sweeping singular-value-thresholding number detection algorithm introduced in \cite{liu2021theorylse}, we verify the phase transition phenomenon for the number detection in the super-resolution of positive sources.

\subsection{Review of the sweeping singular-value-thresholding number detection algorithm}
In \cite{liu2021theorylse}, the authors proposed a number detection algorithm called sweeping single-value-thresholding number detection algorithm. It determines the number of sources by thresholding on the singular value of a Hankel matrix formulated from the measurement data. 

To be more specific, suppose the measurement is taken at $M$ evenly-spaced points $\omega_1=-\Omega, \omega_2, \cdots,  \omega_M = \Omega$, that is,  
\[
\mathbf Y(\omega_j) = \mathcal F \mu (\omega_j) + \mathbf W(\omega_j), \quad j=1, \cdots, M.
\]
We choose a partial measurement at the sample points $z_t= \omega_{(t-1)r+1}$ for $t=1,\cdots,2s+1$, where $s\geq n$ and $r=(M-1) \mod 2s$. 
For ease of exposition, assume $r=\frac{M-1}{2s}$. Then $z_t=\omega_{(t-1)\frac{M-1}{2s}+1} =-\Omega+\frac{t-1}{s}\Omega$ (since $\omega_1=-\Omega$, $\omega_{M}=\Omega$) and the partial measurement is
\[
\mathbf Y(z_t)= \mathcal F \mu (z_t) + \mathbf W(z_t), \quad 1\leq t \leq 2s+1.
\]
Assemble the following Hankel matrix by the measurements that 
\begin{equation}\label{hankelmatrix1}
\mathbf H(s)=\left(\begin{array}{cccc}
\mathbf Y(-\Omega)&\mathbf Y(-\Omega+\frac{1}{s}\Omega)&\cdots& \mathbf Y(0)\\
\mathbf Y(-\Omega+\frac{1}{s}\Omega)&\mathbf Y(-\Omega+\frac{2}{s}\Omega)&\cdots&\mathbf Y(\frac{1}{s}\Omega)\\
\cdots&\cdots&\ddots&\cdots\\
\mathbf Y(0)&\mathbf Y(\frac{1}{s}\Omega)&\cdots&\mathbf Y(\Omega)
\end{array}
\right).\end{equation}
We observe that $\mathbf H(s)$ has the decomposition
\[\mathbf H(s)= DAD^{\top}+\Delta,\]
where $A=\text{diag}(e^{-iy_1\Omega}a_1, \cdots, e^{-iy_n\Omega}a_n)$ and $D=\big(\phi_{s}(e^{i y_1 \frac{\Omega}{s}}), \cdots, \phi_{s}(e^{i y_n \frac{\Omega}{s}})\big)$ with $\phi_{s}(\omega)$ being defined as $(1, \omega, \cdots, \omega^s)^\top$ and 
\begin{equation*}
\Delta = \left(\begin{array}{cccc}
\mathbf {W}(-\Omega)&\mathbf {W}(-\Omega+\frac{1}{s}\Omega)&\cdots& \mathbf {W}(0)\\
\mathbf {W}(-\Omega+\frac{1}{s}\Omega)&\mathbf {W}(-\Omega+\frac{2}{s}\Omega)&\cdots&\mathbf {W}(\frac{1}{s}\Omega)\\
\vdots&\vdots&\ddots&\vdots\\
\mathbf {W}(0)&\mathbf {W}(\frac{1}{s}\Omega)&\cdots&\mathbf {W}(\Omega)
\end{array}
\right).
\end{equation*}
We denote the singular value decomposition of $\mathbf H(s)$ as  
\[\mathbf H(s)=\hat U\hat \Sigma \hat U^*,\]
where $\hat\Sigma =\text{diag}(\hat \sigma_1,\cdots, \hat \sigma_n, \hat \sigma_{n+1},\cdots,\hat\sigma_{s+1})$ with the singular values $\hat \sigma_j$, $1\leq j \leq s+1$, ordered in a decreasing manner. From \cite{liu2021theorylse}, we have the following theorem for the threshold to determine the source number. 
\begin{thm}\label{MUSICthm1}
	Let $s\geq n$ and $\mu=\sum_{j=1}^{n}a_j \delta_{y_j}$ with $y_j\in I(n,\Omega), 1\leq j\leq n$. We have 
	\begin{equation}\label{MUSICthm1equ-1}
	\hat \sigma_j\leq  (s+1)\sigma,\quad j=n+1,\cdots,s+1.
	\end{equation}
	Moreover, if the following separation condition is satisfied
	\begin{equation}\label{MUSICthm1equ0}
	\min_{p\neq j}|y_p-y_j|>\frac{\pi s}{\Omega}\Big(\frac{2n(s+1)}{\zeta(n)^2}\frac{\sigma}{m_{\min}}\Big)^{\frac{1}{2n-2}},
	\end{equation}
	where $\zeta(n)= \left\{
\begin{array}{cc}
(\frac{n-1}{2}!)^2,& \text{$n$ is odd,}\\
(\frac{n}{2})!(\frac{k-2}{2})!,& \text{$n$ is even,}
\end{array} 
\right.$ then
	\begin{equation}\label{MUSICthm1equ2}
	\hat\sigma_{n}>(s+1)\sigma.
	\end{equation}
\end{thm}

Based on this theorem, the threshold should be $(s+1)\sigma$ and the following \textbf{Algorithm \ref{algo:singularvaluenumberdetect}} was proposed to detect the source number for fixed $s$. 

\begin{algorithm}
	\caption{\textbf{Singular-value-thresholding number detection algorithm}}
	\textbf{Input:} Number $s$, Noise level $\sigma$;\\
	\textbf{Input:} measurement: $\mathbf{Y}=(\mathbf Y(\omega_1),\cdots, \mathbf Y(\omega_M))^{\top}$;\\	
	1: $r=(M-1)\mod 2s$,  $\mathbf{Y}_{new}=(\mathbf Y(\omega_1), \mathbf Y(\omega_{r+1}), \cdots, \mathbf Y(\omega_{2sr+1}))^{\top}$\;
	2: Formulate the $(s+1)\times(s+1)$ Hankel matrix $\mathbf H(s)$ from $\mathbf{Y}_{new}$, and
	compute the singular value of $\mathbf H(s)$ as $\hat \sigma_{1}, \cdots,\hat \sigma_{s+1}$ distributed in a decreasing manner\;
	4: Determine $n$ by $\hat \sigma_n>(s+1)\sigma$ and $\hat \sigma_{j}\leq (s+1)\sigma, j=n+1,\cdots, s+1$\;
	\textbf{Return:} $n$. 
	\label{algo:singularvaluenumberdetect}
\end{algorithm}

\medskip
In \textbf{Algorithm \ref{algo:singularvaluenumberdetect}}, the $s\geq n$ should be properly chosen to have a good resolution. To address this issue, a sweeping strategy was utilized and the following \textbf{Algorithm \ref{algo:sweepsingularvaluenumberdetect}} was proposed. It was shown in \cite{liu2021theorylse} that the \textbf{Algorithm \ref{algo:sweepsingularvaluenumberdetect}} achieves the optimal resolution order.

\begin{algorithm}
	\caption{\textbf{Sweeping singular-value-thresholding number detection algorithm}}	
	\textbf{Input:} Noise level $\sigma$, measurement: $\mathbf{Y}=(\mathbf Y(\omega_1),\cdots, \mathbf Y(\omega_M))^{\top}$;\\
	\textbf{Input:} $n_{max}=0$\\
	\For{$s=1: \lfloor \frac{M-1}{2}\rfloor$}{
		Input $s,\sigma, \mathbf{Y}$ to \textbf{Algorithm 1}, save the output of \textbf{Algorithm 1} as $n_{recover}$\; 
		\If{$n_{recover}>n_{max}$}{$n_{max}=n_{recover}$}
	}
	{
		\textbf{Return} $n_{max}$.
	}
	\label{algo:sweepsingularvaluenumberdetect}
\end{algorithm}

\subsection{Phase transition} \label{sec-phase}
We know from Section \ref{section:onedresolutionlimit} that the resolution limit to the number detection problem in super-resolution of positive sources is bounded from below and above by  $\frac{C_1}{\Omega}(\frac{\sigma}{m_{\min}})^{\frac{1}{2n-2}}$ and $\frac{C_2}{\Omega}(\frac{\sigma}{m_{\min}})^{\frac{1}{2n-2}}$, respectively for some constants $C_1,C_2$. This indeed implies  
a phase transition phenomenon in the problem. Specifically, recall that the super-resolution factor is 
$SRF=\frac{\pi}{d_{\min} \Omega}$ and the $\frac{m_{\min}}{\sigma}$ can be viewed as the signal-to-noise ratio $SNR$. Taking the logarithm of both sides of the two bounds, we can conclude that the exact number detection is guaranteed if
$$
\log(SNR) > (2n-2)\log(SRF)+(2n-2) \log \frac{C_1}{\pi}, 
$$
and may fail if
$$
\log(SNR) < (2n-2)\log(SRF)+(2n-2) \log \frac{C_2}{\pi}.  
$$
As a consequence, we expect that in the parameter space of $\log SNR-\log SRF$, there exist two lines both with slope $2n-2$ such that the number detection is successful for cases above the first line and unsuccessful for cases below the second. In the intermediate region between the two lines, the number detection can be either successful or unsuccessful from case to case. This is clearly demonstrated in the numerical experiments below. 

We fix $\Omega=1$ and consider $n$ point sources randomly spaced in $\left[-\frac{(n-1)\pi}{2}, \frac{(n-1)\pi}{2} \right]$ with positive amplitudes $a_j$'s. The noise level is $\sigma$ and the minimum separation distance between sources is $d_{\min}$. We perform 10000 random experiments (the randomness is in the choice of $(d_{\min},\sigma, y_j, a_j)$) to detect the source number based on \textbf{Algorithm \ref{algo:sweepsingularvaluenumberdetect}}. Figure \ref{fig:numberphasetransition} shows the results for $n=2,4,$ respectively. In each case, two lines of slope $2n-2$ strictly separate the blue points (successful detection) and red points (unsuccessful detection) and in-between is the phase transition region. It clearly elucidates the phase transition phenomenon of \textbf{Algorithm \ref{algo:sweepsingularvaluenumberdetect}} and is consistent with our theory.

\begin{figure}[!h]
	\centering
	\begin{subfigure}[b]{0.28\textwidth}
		\centering
		\includegraphics[width=\textwidth]{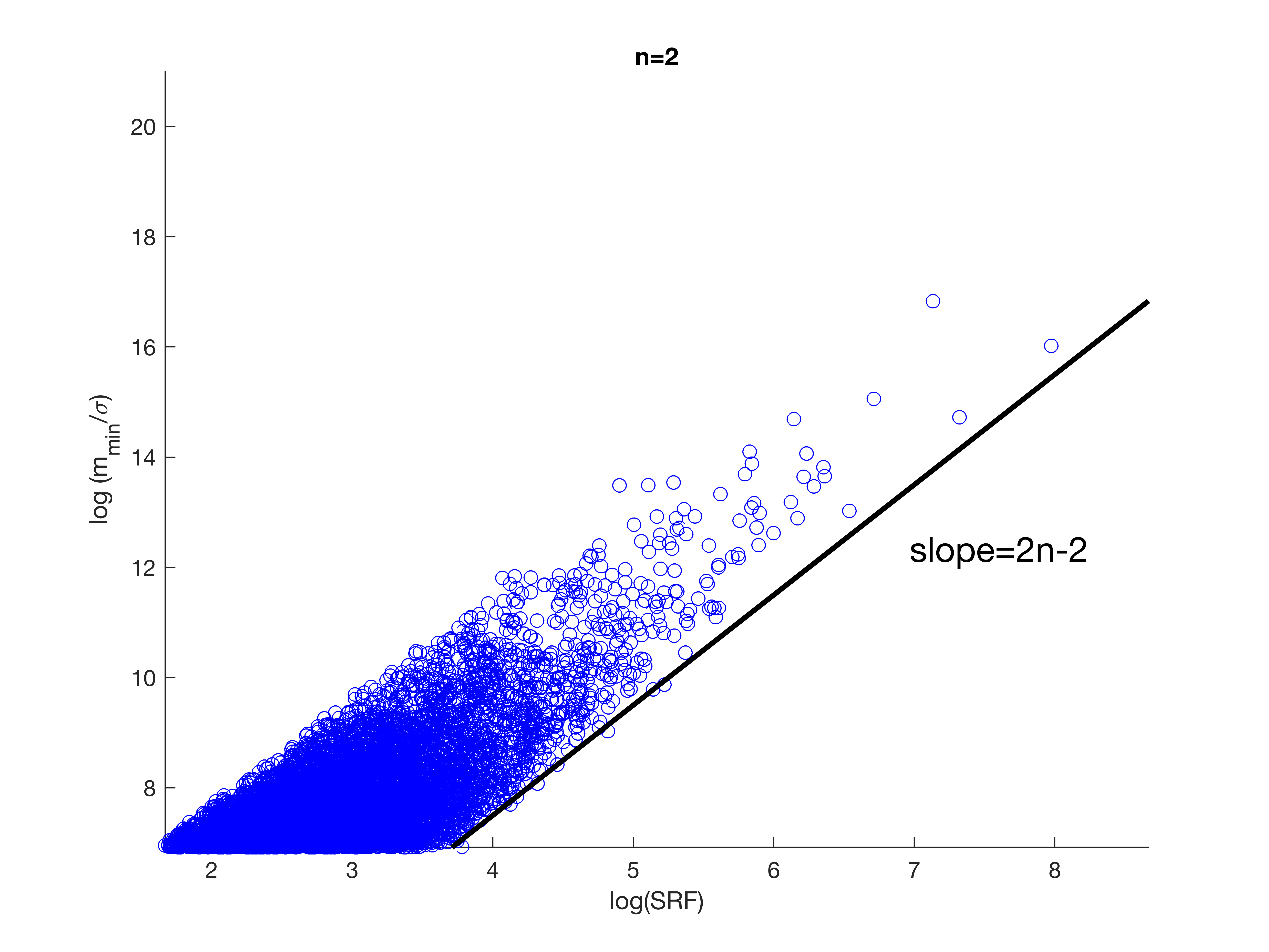}
		\caption{detection success}
	\end{subfigure}
	\begin{subfigure}[b]{0.28\textwidth}
		\centering
		\includegraphics[width=\textwidth]{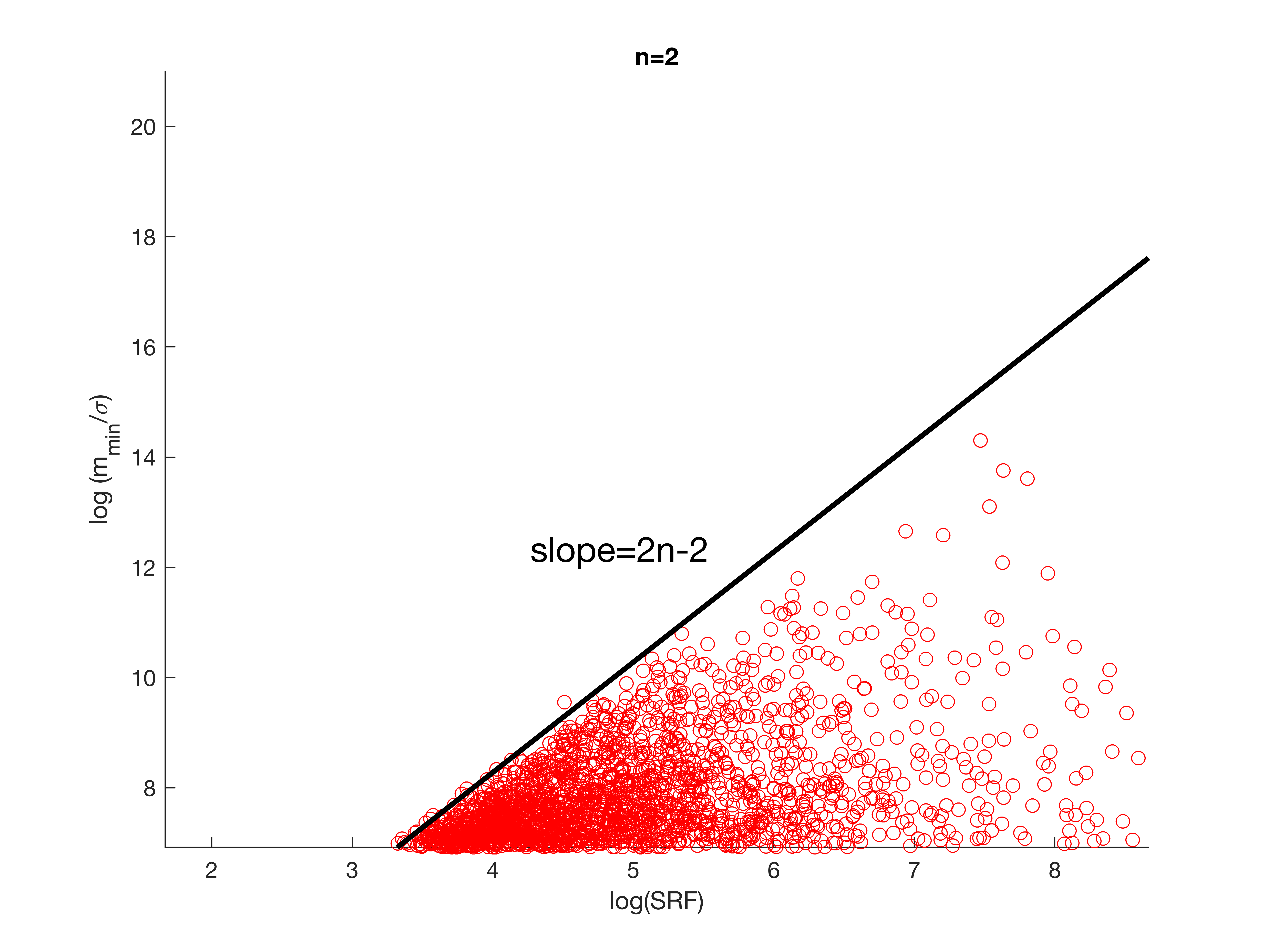}
		\caption{detection fail}
	\end{subfigure}
	\begin{subfigure}[b]{0.28\textwidth}
		\centering
		\includegraphics[width=\textwidth]{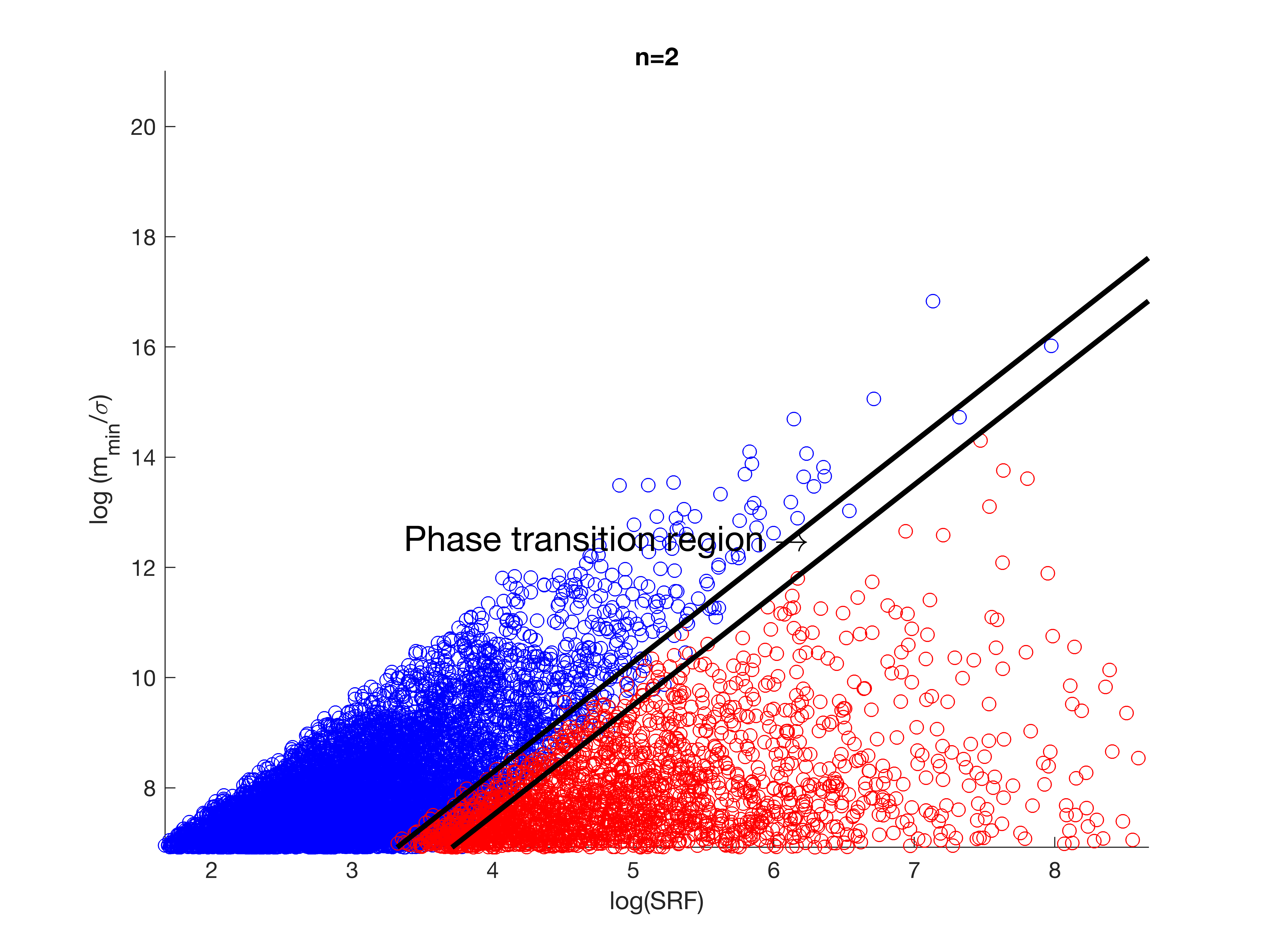}
		\caption{phase transition region}
	\end{subfigure}
	\begin{subfigure}[b]{0.28\textwidth}
		\centering
		\includegraphics[width=\textwidth]{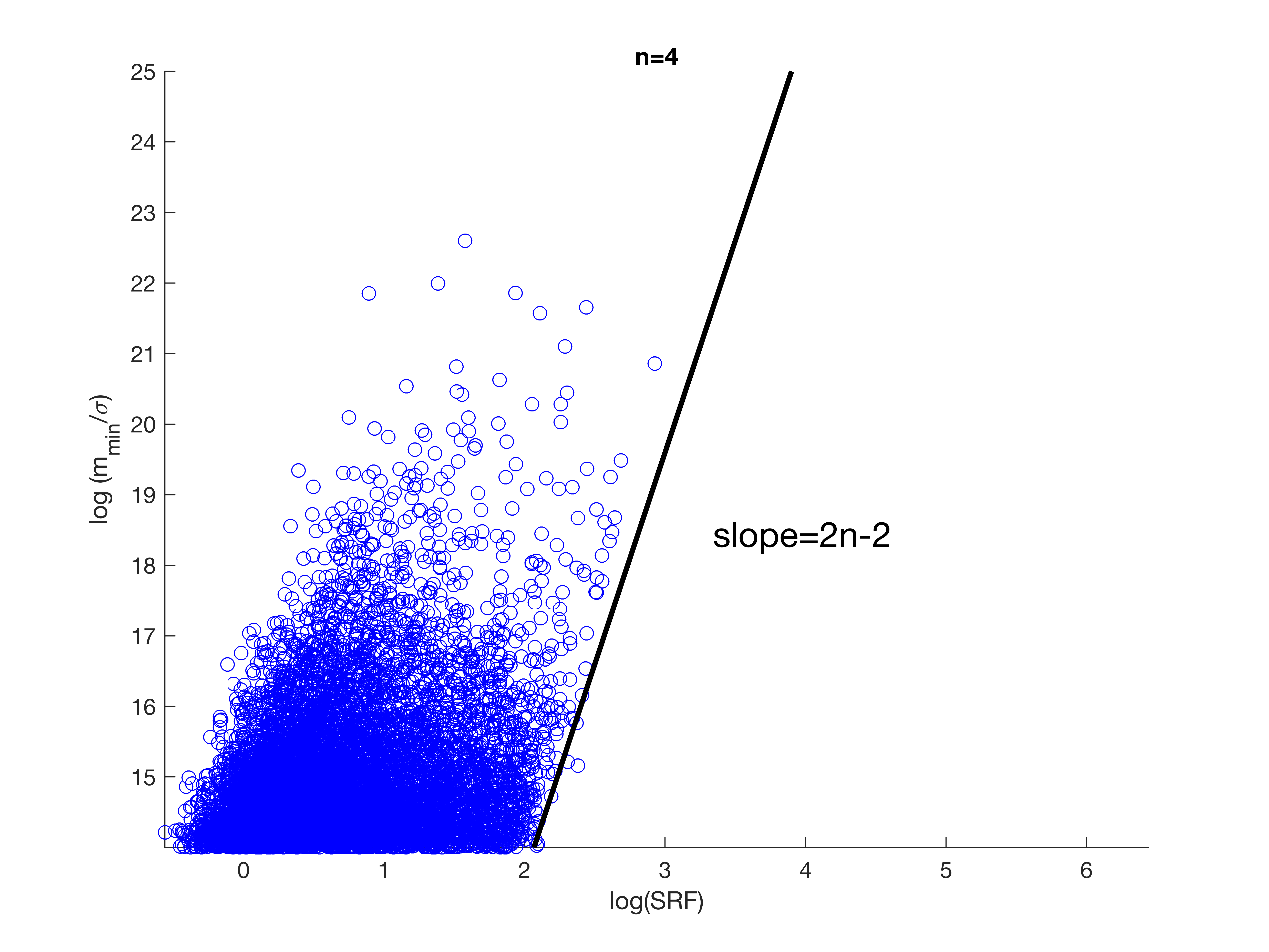}
		\caption{detection success}
	\end{subfigure}
	\begin{subfigure}[b]{0.28\textwidth}
		\centering
		\includegraphics[width=\textwidth]{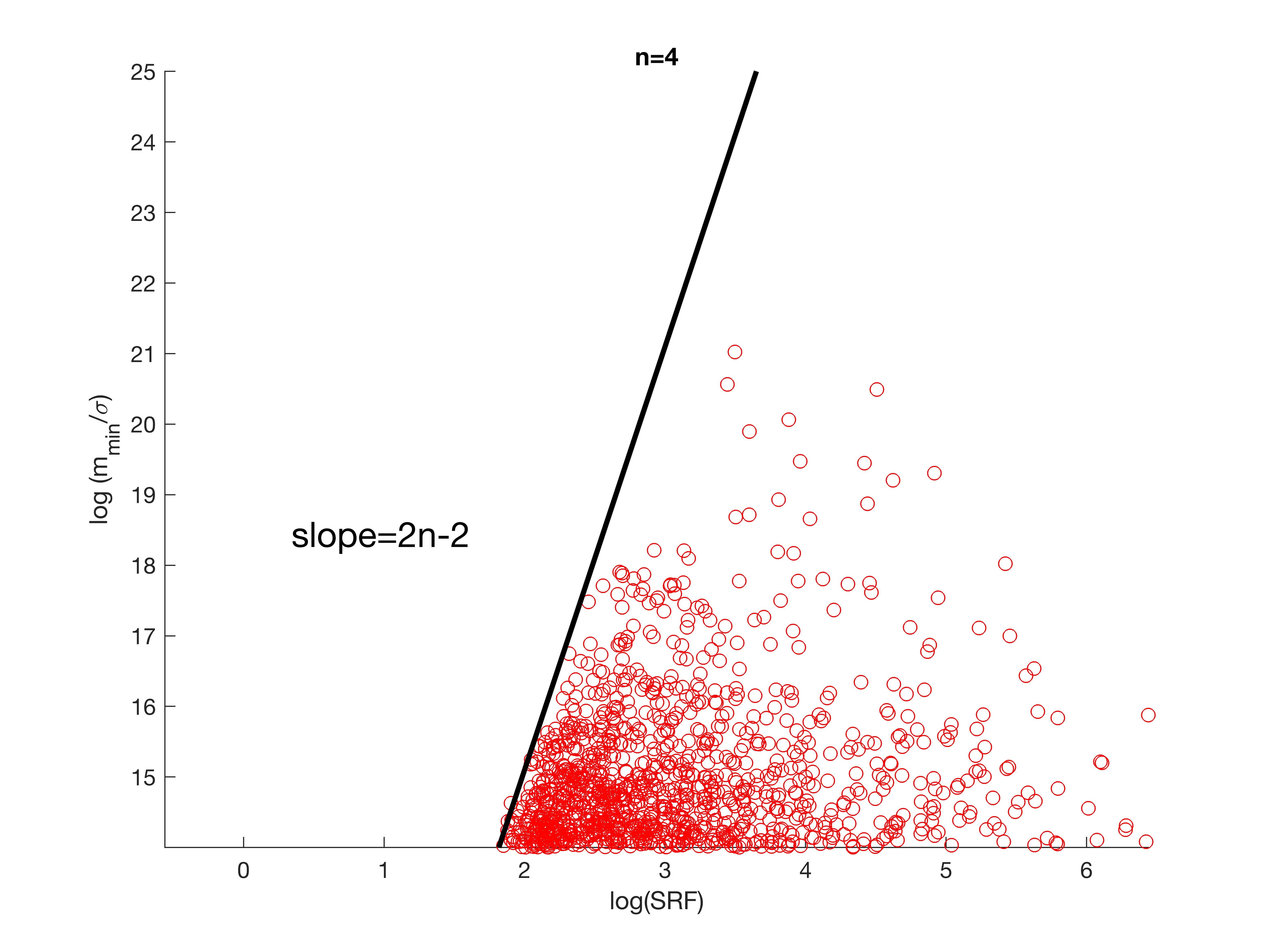}
		\caption{detection fail}
	\end{subfigure}
	\begin{subfigure}[b]{0.28\textwidth}
		\centering
		\includegraphics[width=\textwidth]{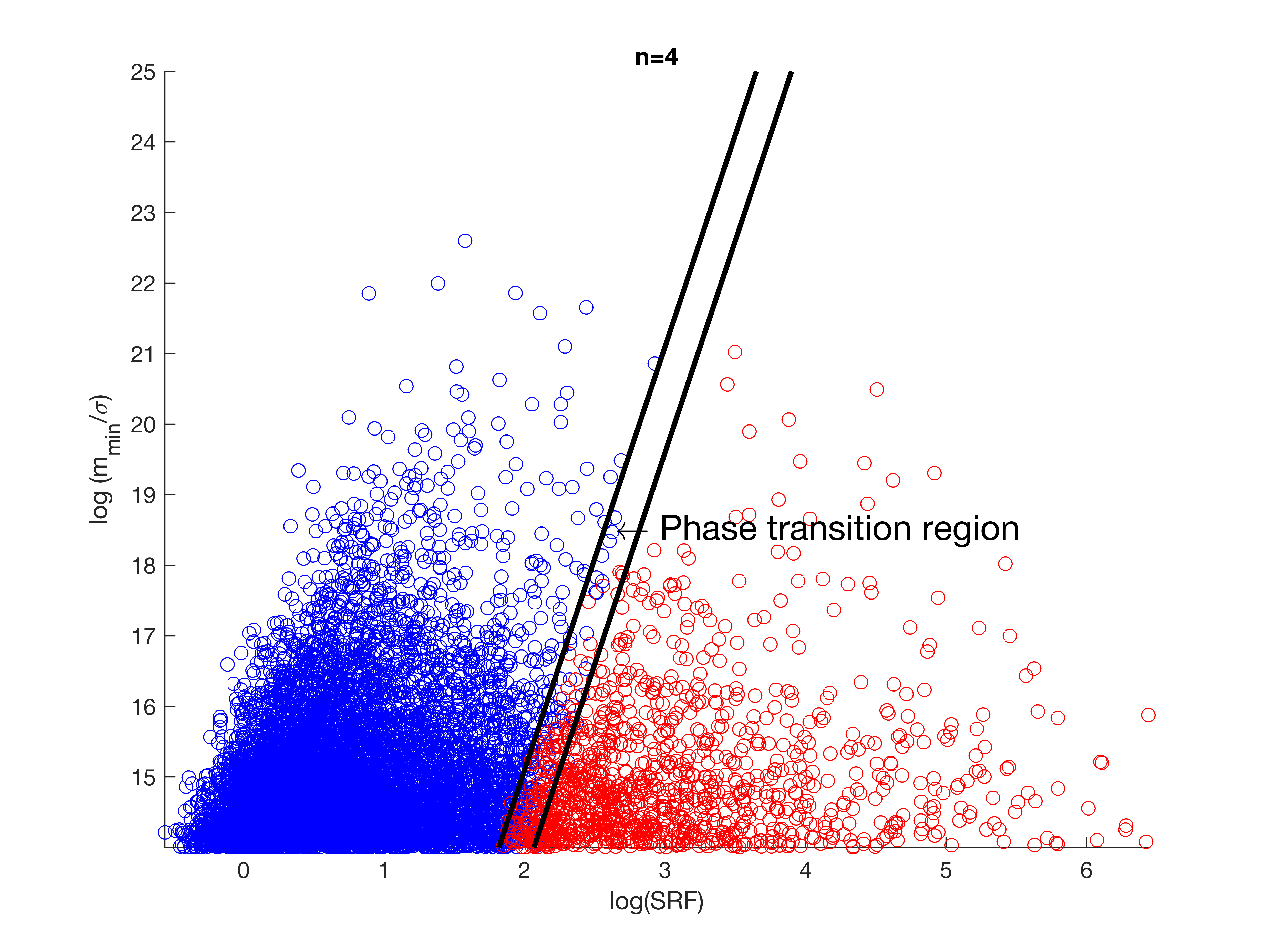}
		\caption{phase transition region}
	\end{subfigure}
	\caption{Plots of the successful and the unsuccessful number detection by \textbf{Algorithm 2} depending on the relation between $\log(SRF)$ and $\log(\frac{m_{\min}}{\sigma})$. (a) illustrates that two positive point source can be exactly detected if $\log(\frac{m_{\min}}{\sigma})$ is above a line of slope $2$ in the parameter space. Conversely, for the same case, (b) shows that the number detection fails if $\log(\frac{m_{\min}}{\sigma})$ falls below another line of slope $2$. (c) highlights the phase transition region which is bounded by the black slashes in (a) and (b). (d),(e) and (f) illustrate parallel results for four positive point sources.}
	\label{fig:numberphasetransition}
\end{figure}

\section{Phase transition in the location recovery}
In this section, by the MUSIC algorithm we verify the phase transition phenomenon for the location recovery in the super-resolution of positive sources.

\subsection{Review of the MUSIC algorithm} 
In this section we review the MUSIC algorithm. From the measurement $\mathbf Y =(\mathbf Y(\omega_1),\mathbf Y(\omega_2),\cdots,\mathbf Y(\omega_M))^\top$ and $\hat M=\lfloor \frac{M-1}{2}\rfloor$, we assemble the $(\hat M +1)\times(\hat M +1)$ Hankel matrix, 
\begin{equation}\label{hankelmatrix1}\vect H=\begin{pmatrix}
\mathbf Y(\omega_1)&\mathbf Y(\omega_2)&\cdots& \mathbf Y(\omega_{\hat M})\\
\mathbf Y(\omega_2)&\mathbf Y(\omega_3)&\cdots&\mathbf Y(\omega_{\hat M+1})\\
\cdots&\cdots&\ddots&\cdots\\
\mathbf Y(\omega_{\hat M})&\mathbf Y(\omega_{\hat M+1})&\cdots&\mathbf Y(\omega_{2\hat M+1})
\end{pmatrix}.
\end{equation}
 We perform the following singular value decomposition for $\vect H$,
\[
\vect H=\hat U\hat \Sigma \hat U^*=[\hat U_1\quad \hat U_2]\text{diag}(\hat \sigma_1, \hat \sigma_2,\cdots,\hat \sigma_n,\hat \sigma_{n+1},\cdots,\hat \sigma_{\hat M+1})[\hat U_1\quad \hat U_2]^*,
\]
where $\hat U_1=(\hat U(1),\cdots,\hat U(n)), \hat U_2=(\hat U(n+1),\cdots,\hat U(\hat M+1))$ with $n$ being the source number. Then we denote the orthogonal projection to the space $\hat U_2$ by $\hat P_2x=\hat U_2(\hat U_2^*x)$. For a test vector $\Phi(x)=(1, e^{ihx},\cdots,e^{i\hat Mhx})^{\top}$ with $h$ being the spacing parameter, we define the MUSIC imaging functional 
\begin{align*}
\hat J(x)=\frac{||\Phi(x)||_2}{||\hat P_2\Phi(x)||_2}=\frac{||\Phi(x)||_2}{||\hat U_2^*\Phi(x)||_2}.
\end{align*}
The local  maximizers of $\hat J(x)$ indicate the locations of the point sources. In practice, we test evenly spaced points in a specified interval and plot the discrete imaging functional and then determine the source locations by detecting the peaks. We present the peak selection algorithm as \textbf{Algorithm \ref{algo:peakselection}} and summarize the MUSIC algorithm in \textbf{Algorithm \ref{algo:standardmusic}} below.

\begin{algorithm}[H]
	\caption{\textbf{MUSIC algorithm}}
	\textbf{Input:} Measurements: $\mathbf{Y}=(\mathbf Y(\omega_1),\cdots, \mathbf Y(\omega_M))^{\top}$, sampling distance $h$, source number $n$\;
	\textbf{Input:} Region of test points $[TS, TE]$ and spacing of test points $TPS$\;
	1: Let $\hat M =\lfloor\frac{M-1}{2}\rfloor$, formulate the $(\hat M +1)\times (\hat M +1)$ Hankel matrix $\hat X$ from $\mathbf{Y}$\;
	2: Compute the singular vector of $\hat X$ as $\hat U(1), \hat U(2),\cdots,\hat U(\hat M +1)$ and formulate the noise space $\hat U_{2}=(\hat U(n+1),\cdots,\hat U(\hat M +1))$\;
	3: For test points $x$'s in $[TS, TE]$ evenly spaced by $TPS$, construct the test vector $\Phi(x)=(1,e^{ihx}, \cdots, e^{i\hat M hx})^{\top}$\; 
	4: Plot the MUSIC imaging functional $\hat J(x)=\frac{||\Phi(x)||_2}{||\hat U_2^*\Phi(x)||_2}$\;
	5: Select the peak locations $\hat y_j$'s in the $\hat J(x)$ by \textbf{Algorithm \ref{algo:peakselection}}\;
	\textbf{Return} $\hat y_j$'s.
	\label{algo:standardmusic}
\end{algorithm}

\begin{algorithm}[H]
    \caption{\textbf{Peak selection algorithm}}
    \textbf{Input:} Image $IMG = (f(x_1), \cdots, f(x_N))$\;
    \textbf{Input:} Peak compare range $PCR$, differential compare range $DCR$,  differential compare threshold $DCT$\;
    1: Initialize the Local maximum points $LMP = [\ ]$, peak points $PP = [\ ]$\;
    2: Differentiate the image $IMG$ to get the $DIMG = (f'(x_1), \cdots, f'(x_N))$\;
    3: \For{$j=1:N$}{
    \If{$f(x_j) = \max(f(x_{j-PCR}), f(x_{j-PCR+1}), \cdots, f(x_{j+PCR}))$}{
    $LMP$ appends $x_j$\;
    }}
    4: \For{$x_j$ in $LMP$}{
    \If{
    $\max(|f'(x_{j-DCR})|, |f'(x_{j-DCR+1})|, \cdots, |f'(x_{j+DCR})|)\geq DCT$
    }{$PP$ appends $x_j$\;}
    }
    \textbf{Return:} $PP$.
    \label{algo:peakselection}
\end{algorithm}

\subsection{Phase transition}
The derived bounds for the resolution limit $\mathcal D_{supp}^+$ of the location recovery in the super-resolution of positive sources implies a phase transition in the problem. Taking the logarithm of both sides of the two bounds, we can draw a conclusion that the location recovery is stable if
$$
\log(SNR) > (2n-1)\log(SRF)+(2n-1) \log \frac{C_3}{\pi}, 
$$
and may be unstable if
$$
\log(SNR) < (2n-1)\log(SRF)+(2n-1) \log \frac{C_4}{\pi},
$$
for certain constants $C_3, C_4$. Similar to the number detection, we expect that in the parameter space of $\log SNR-\log SRF$, there exist two lines both with slope $2n-1$ such that the location recovery is stable for cases above the first line and unstable for cases below the second. This phase transition phenomenon has been demonstrated numerically using the Matrix Pencil method, MUSIC and ESPRIT in \cite{batenkov2019super, liao2016music, li2021stable, li2020super} for resolving general sparse sources.

In what follows, we shall conduct numerical experiments to demonstrate the phase transition phenomenon for the MUSIC algorithm in the super-resolution of positive sources. For simplicity, we fix $\Omega=1$ and consider $n=2$ or $4$ positive point sources separated with minimum separation $d_{\min}$. We perform $10000$ random experiments (the randomness is in the choice of $(d_{\min},\sigma, y_j, a_j)$ to recover the source locations using \textbf{Algorithm \ref{algo:standardmusic}}. The recovery is deemed stable only if $n$ locations $\hat y_j$'s are recovered and they are in a $\frac{d_{\min}}{2}$-neighborhood of the ground truth; see \textbf{Algorithm \ref{algo:singleexperiemnt}} for details in a single experiment. As is shown in Figure \ref{fig:onedsupportphasetransition}, in each case, two lines with slope $2n-1$ strictly separate the blue points (stable recoveries) and red points (unstable recoveries), and in-between is the phase transition region. This is exactly the predicted phase transition phenomenon by our theory. It also demonstrates that the MUSIC can resolve the location of positive point sources with optimal resolution order. 


\begin{algorithm*}[H]\label{algo:singleexperiemnt}
	\caption{\textbf{A single experiment}}	
	\textbf{Input:} Sources $\mu=\sum_{j=1}^{n}a_j \delta_{y_j}$, 
	source number $n$;\\
	\textbf{Input:} Measurements: $\mathbf{Y}$\\
	Input source number $n$ and measurement $\vect Y$ to \textbf{Algorithm \ref{algo:standardmusic}} and save the output as $\hat y_1, \cdots, \hat y_k$, which are ordered in an increasing manner\; 
	\eIf{k==n}{
		Compute the reconstruction error for the source location $y_j$ that $e_j:= |\hat y_j - y_j|$\; 
		\eIf{$\max_{j=1,\cdots,n} e_j< \frac{\min_{p\neq j}|y_p- y_j|}{2}$}
		{Return Stable}
		{Return Unstable}
	}{Return Unstable.}
\end{algorithm*}

\begin{figure}[!h]
	\centering
	\begin{subfigure}[b]{0.28\textwidth}
		\centering
		\includegraphics[width=\textwidth]{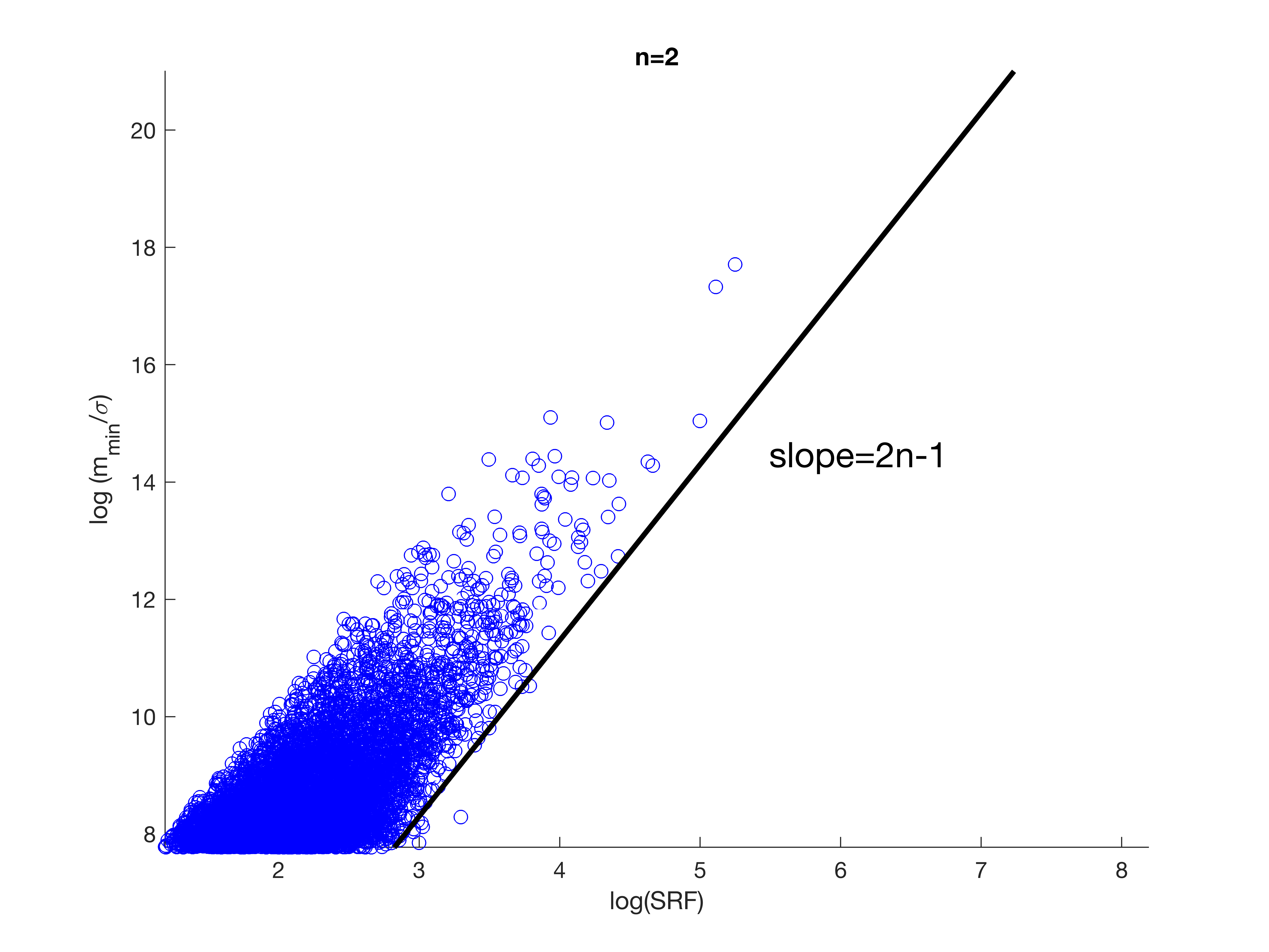}
		\caption{recovery success}
	\end{subfigure}
	\begin{subfigure}[b]{0.28\textwidth}
		\centering
		\includegraphics[width=\textwidth]{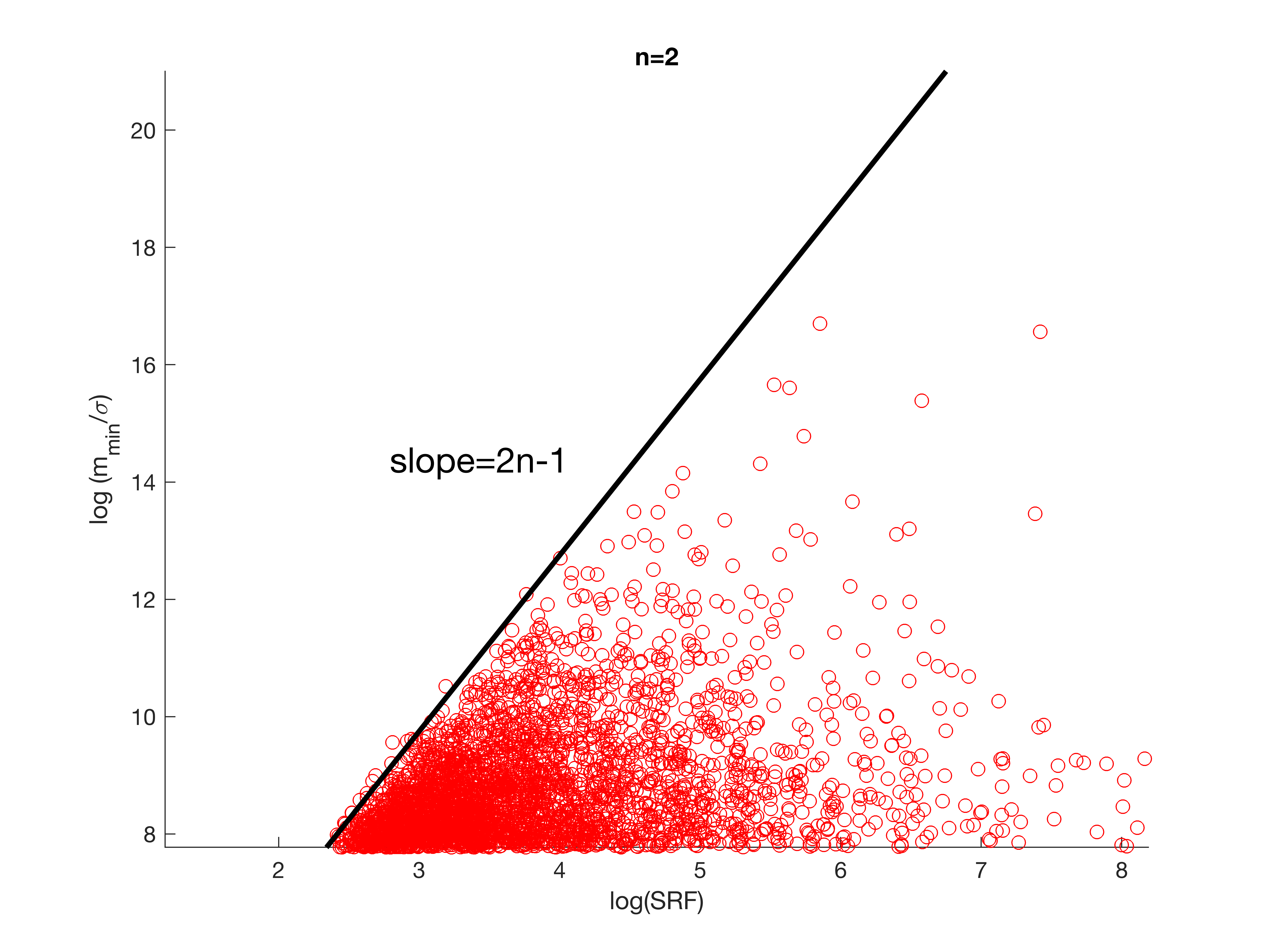}
		\caption{recovery fail}
	\end{subfigure}
	\begin{subfigure}[b]{0.28\textwidth}
		\centering
		\includegraphics[width=\textwidth]{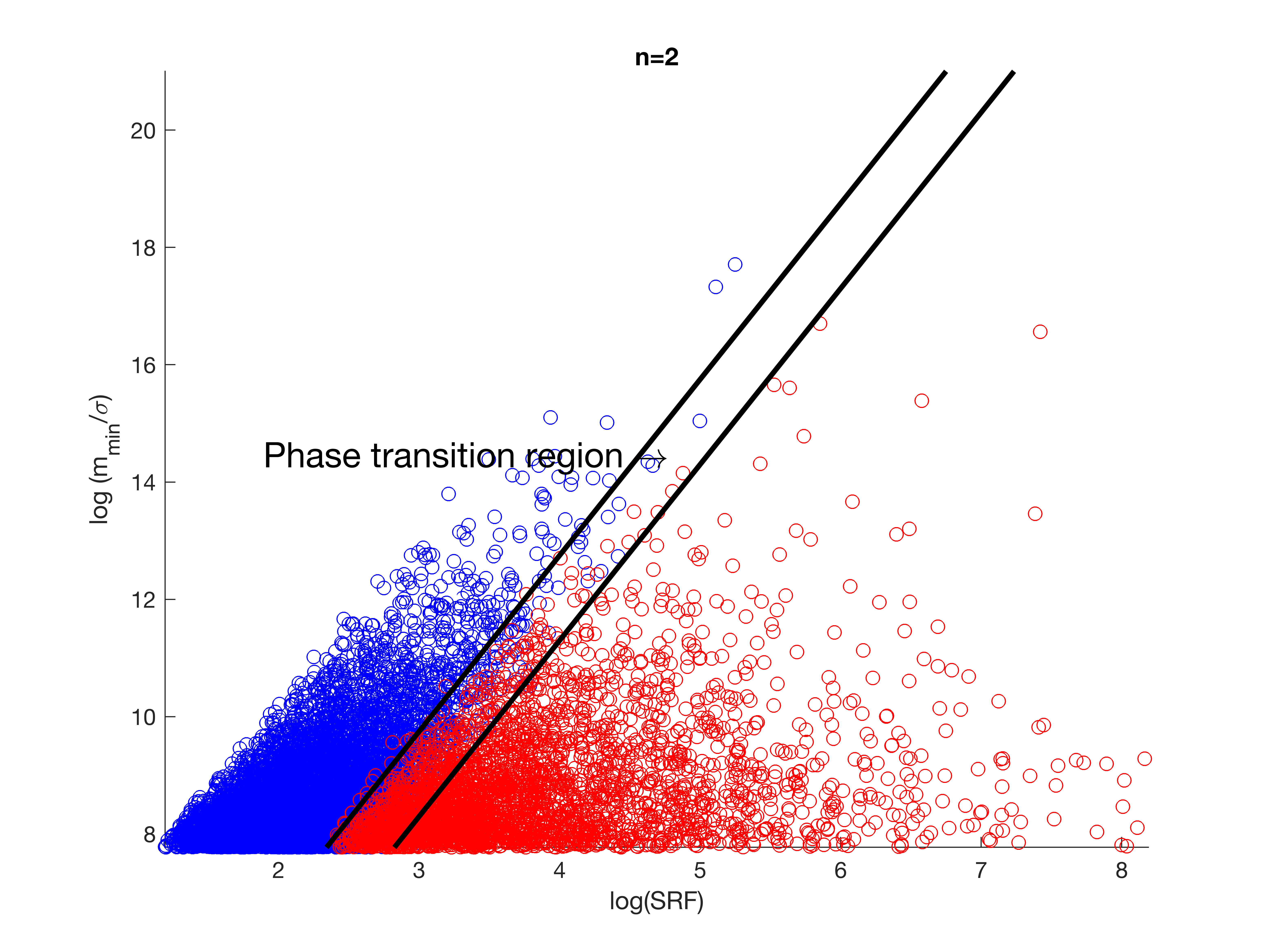}
		\caption{phase transition region}
	\end{subfigure}
	\begin{subfigure}[b]{0.28\textwidth}
		\centering
		\includegraphics[width=\textwidth]{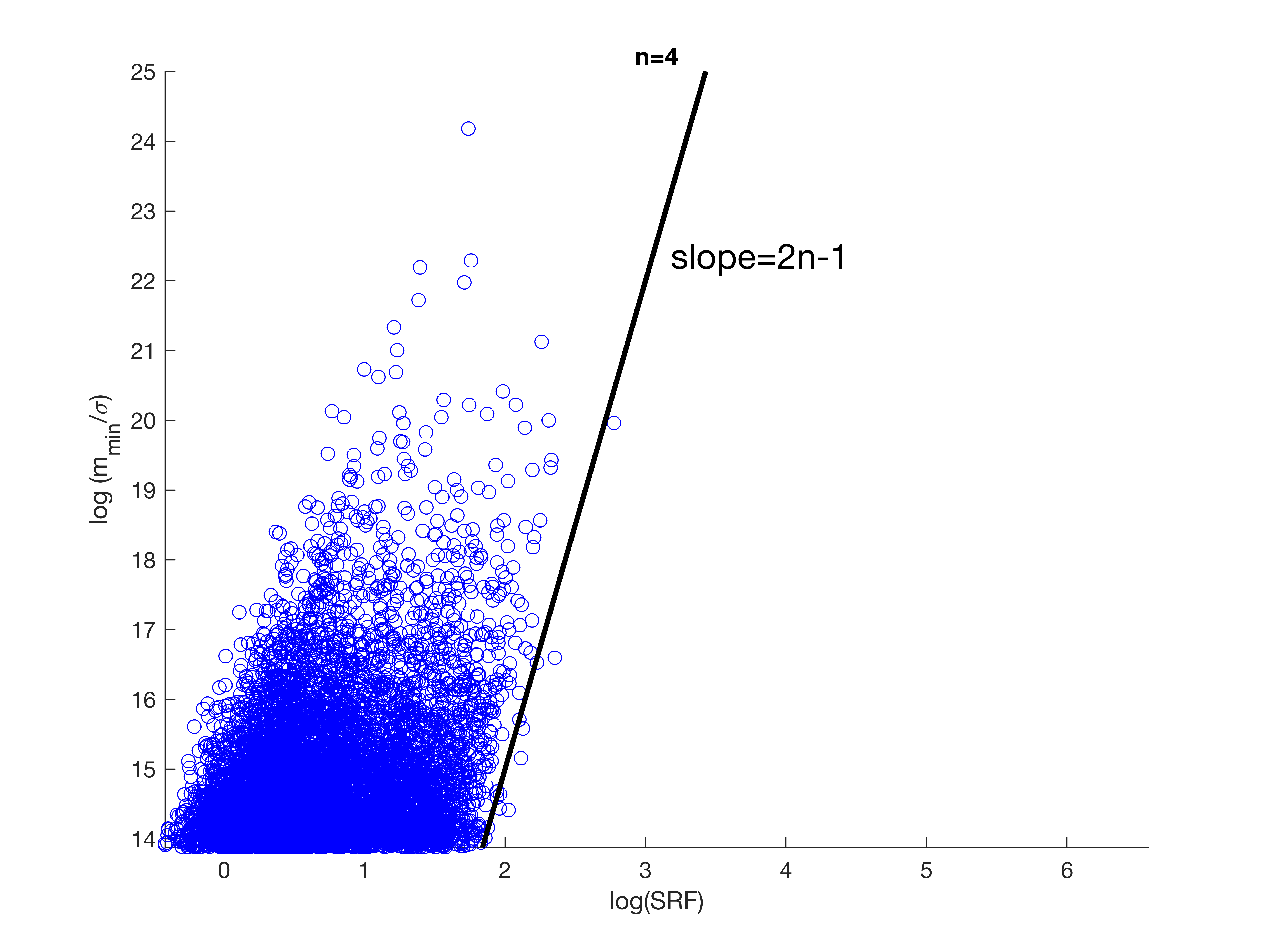}
		\caption{recovery success}
	\end{subfigure}
	\begin{subfigure}[b]{0.28\textwidth}
		\centering
		\includegraphics[width=\textwidth]{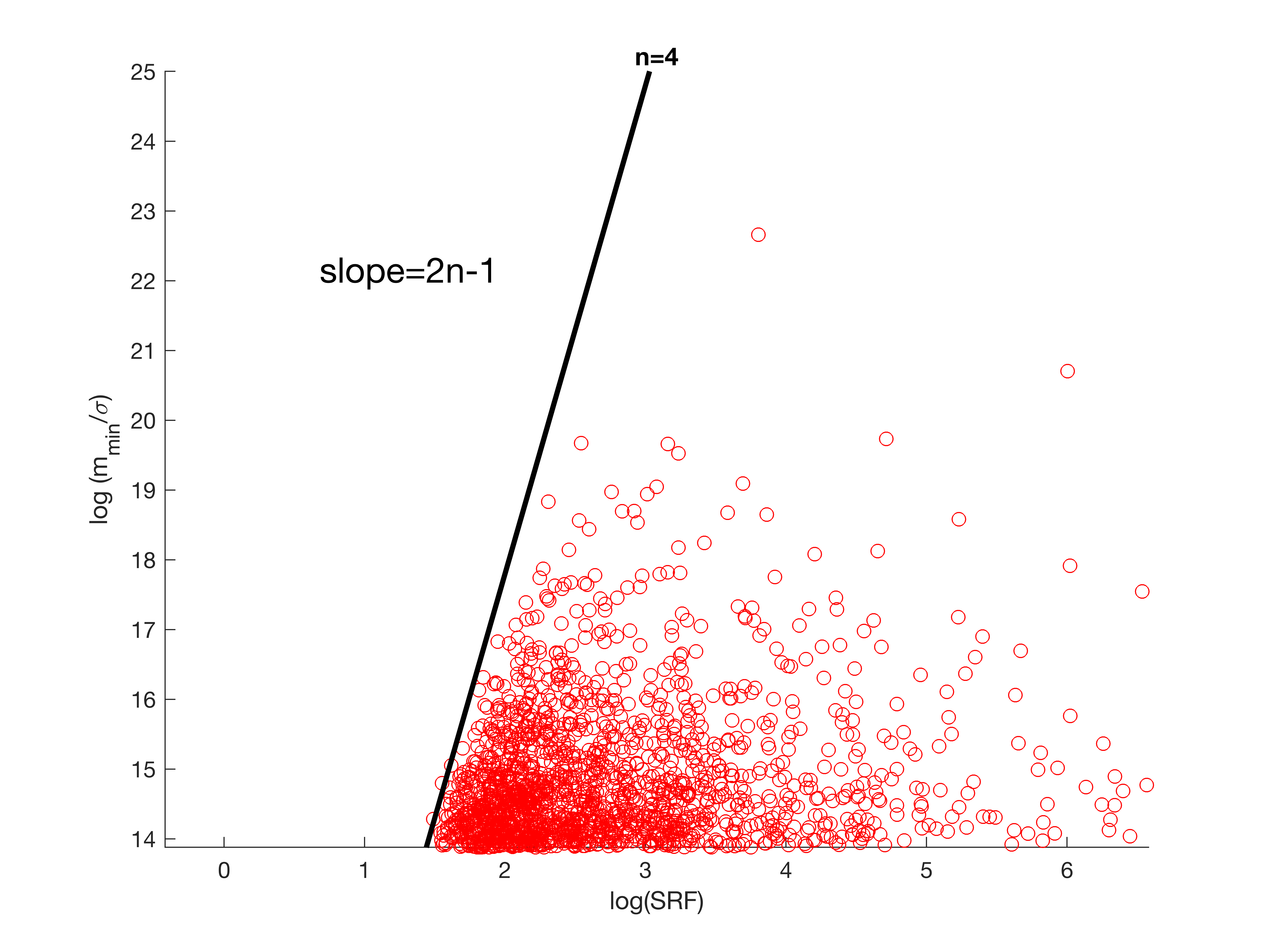}
		\caption{recovery fail}
	\end{subfigure}
	\begin{subfigure}[b]{0.28\textwidth}
		\centering
		\includegraphics[width=\textwidth]{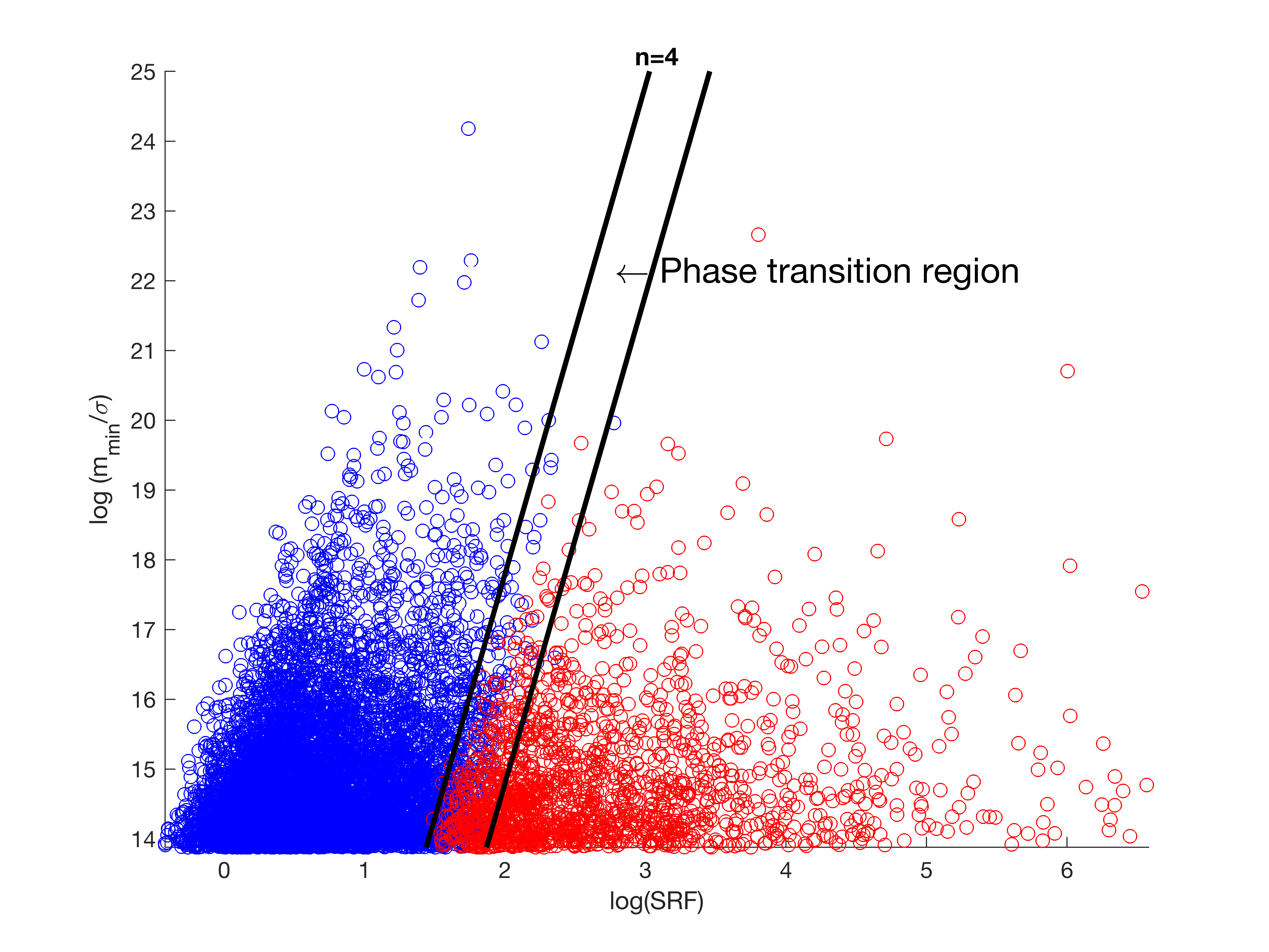}
		\caption{phase transition region}
	\end{subfigure}
	\caption{Plots of the stable and the unstable location recoveries by \textbf{Algorithm \ref{algo:standardmusic}} in view of the relation between $\log(SRF)$ and $\log(\frac{m_{\min}}{\sigma})$. (a) illustrates that the locations of two positive point sources can be stably recovered if $\log(\frac{m_{\min}}{\sigma})$ is above a line of slope $3$ in the parameter space. Conversely, for the same case, (b) shows that the locations cannot be stably recovered by MUSIC if $\log(\frac{m_{\min}}{\sigma})$ falls below another line of slope $3$. (c) highlights the phase transition region which is bounded by the black slashes in (a) and (b). (d),(e) and (f) illustrate parallel results for four positive point sources.}
	\label{fig:onedsupportphasetransition}
\end{figure}

\section{Conclusions and future works}
In this paper, we have introduced the resolution limit for respectively the number detection and the location recovery in the super-resolution of positive sources. We have quantitatively characterized the two limits by establishing their sharp upper and lower bounds. We have also verified the phase transition phenomena that predicted by our theory in the number detection and support recovery problems. 

Our new technique provides a way to analyze the resolving capability of the super-resolution of positive sources. The applications of the technique introduced here to other problems will be presented in a near future.

\section{Proofs of results in Section \ref{section:onedresolutionlimit}}\label{section:proofofmainresult}
We first introduce some notation and lemmas that are used in the following proofs. Set 
\begin{equation}\label{equ:phiformula}
	\phi_{s}(t) = \left(1, t, \cdots, t^s\right)^{\top}.
\end{equation}
We recall the Stirling formula that 
\begin{equation}\label{stirlingformula0}
	\sqrt{2\pi} n^{n+\frac{1}{2}}e^{-n}\leq n! \leq e n^{n+\frac{1}{2}}e^{-n}.
\end{equation}
\begin{lem}{\label{lem:invervandermonde}}
	Let $t_1, \cdots, t_k$ be $k$ different real numbers and let $t$ be a real number. We have
	\[
	\left(D_k(k-1)^{-1}\phi_{k-1}(t)\right)_{j}=\Pi_{1\leq q\leq k,q\neq j}\frac{t- t_q}{t_j- t_q},
	\]
	where $D_k(k-1):=  \big(\phi_{k-1}(t_1),\cdots,\phi_{k-1}(t_k)\big)$ with $\phi_{k-1}(\cdot)$ defined by (\ref{equ:phiformula}). 
\end{lem}
\begin{proof}
	This is \cite[Lemma 5]{liu2021theorylse}. For the reader's convenience, we present a simple proof here. We denote $\left(D_{k}(k-1)^{-1}\right)_{jq}=b_{jq}$. Observe that
	$$
	\left(D_k(k-1)^{-1}\phi_{k-1}(t)\right)_{j}=\sum_{q=1}^{k}b_{jq}t^{q-1}.
	$$
	We have 
	\[\sum_{q=1}^{k}b_{jq}(t_p)^{q-1}=\delta_{jp},\ \forall j,p=1,\cdots,k,\]
	where $\delta_{jp}$ is the Kronecker delta function. Then the polynomial $P_j(x)=\sum_{q=1}^{k}b_{jq}x^{q-1}$ satisfies $P_{j}(t_1)=0,\cdots,P_j(t_{j-1})=0,P_j(t_j)=1,P_j(t_{j+1})=0,\cdots,P_j(t_k)=0$. Therefore, it must be the Lagrange polynomial
	\[
	P_j(x)=\Pi_{1\leq q\leq k,q\neq j}\frac{x-t_q}{t_j-t_q}.
	\]
	It follows that
	\begin{align*}
		\left(D_k(k-1)^{-1}\phi_{k-1}(t)\right)_{j}=\Pi_{1\leq q\leq k,q\neq j}\frac{ t- t_q}{t_j-t_q}.
	\end{align*}
\end{proof}

\subsection{Proof of Theorem \ref{thm:numberlowerboundthm0}}
\begin{proof}\textbf{Step 1.} Let 
	\begin{equation}\label{equ:numberlowerboundequ1}
		\tau = \frac{e^{-1}}{\Omega}\Big(\frac{\sigma}{m_{\min}}\Big)^{\frac{1}{2n-2}}
	\end{equation} 
	and $t_1=-(n-1)\tau, t_2=-(n-2)\tau,\cdots, t_n=0, t_{n+1}=\tau,\cdots, t_{2n-1}=(n-1)\tau$. Consider the following system of linear equations:  
	\begin{equation}\label{equ:numberlowerboundthm0equ1}
		Aa=0, 
	\end{equation}
	where $A=\big(\phi_{2n-3}(t_1),\cdots,\phi_{2n-3}(t_{2n-1})\big)$ with $\phi_{2n-3}(\cdot)$ being defined by (\ref{equ:phiformula}). Since $A$ is underdetermined, there exists a nontrivial solution $a=(a_1,\cdots,a_{2n-1})^{\top}$ to (\ref{equ:numberlowerboundthm0equ1}). By the linear independence of the any $(2n-2)$ column vectors of $A$, we can show that all $a_j$'s are nonzero. By a scaling of $a$, we can assume that $a_{2n-1}>0$ and
	\begin{equation}\label{equ:numberlowerboundthm0equ4}
		\min_{1\leq j\leq n}|a_{2j-1}|=m_{\min}.
	\end{equation}
	We define
	\[
	\mu=\sum_{j=1}^{n}a_{2j-1} \delta_{t_{2j-1}},\quad  \hat \mu=\sum_{j=1}^{n-1}-a_{2j}\delta_{t_{2j}}.
	\]
	We shall show that the intensities in $\hat \mu$ and $\mu$ are all positive and $||\mathcal F [\hat \mu]-\mathcal F [\mu]||_{\infty}<\sigma$ in the subsequent steps. 
	
	\textbf{Step 2.}
	We first analyze the sign of each $a_j, j=1, \cdots, 2n-1,$ based on $a_{2n-1}>0$. The equation (\ref{equ:numberlowerboundthm0equ1}) implies that
	\[
	-a_{2n-1}\phi_{2n-3}(t_{2n-1}) = \big(\phi_{2n-3}(t_{1}), \cdots, \phi_{2n-3}(t_{2n-2})\big)(a_{1},\cdots, a_{2n-2})^{\top},
	\]
	and hence
	\[
	-a_{2n-1} \left(\phi_{2n-3}(t_{1}), \cdots, \phi_{2n-3}(t_{2n-2})\right)^{-1}\phi_{2n-3}(t_{2n-1}) = (a_{1},\cdots, a_{2n-2})^{\top}.
	\]
	Together with Lemma \ref{lem:invervandermonde}, we have 
	\begin{equation}\label{equ:proofnumberlowerequ2}
		-a_{2n-1}\Pi_{1\leq q\leq 2n-2, q\neq j}\frac{t_{2n-1}-t_{q}}{t_{j}-t_{q}}= a_{j},
	\end{equation}
	for $j=1, \cdots, 2n-2$. Observe first that $\Pi_{1\leq q\leq 2n-2, q\neq j}(t_{2n-1}-t_{q})$ is always positive for $1\leq j\leq 2n-2$. For $j=2n-2$, since $a_{2n-1}>0$, $-a_{2n-1} \Pi_{1\leq q\leq 2n-2, q\neq 2n-2}(t_{2n-2}-t_q)$ is negative in (\ref{equ:proofnumberlowerequ2}). Thus we have $a_{2n-2}<0$. In the same fashion, we see that $a_{j}<0$ for even $j$ and $a_j>0$ for odd $j$. Hence the intensities in $\hat \mu$ and $\mu$ are all positive.

	\textbf{Step 3.}
	We demonstrate that $||\mathcal F [\hat \mu]-\mathcal F [\mu]||_{\infty}<\sigma$. Observe that
	\begin{equation}\label{equ:numberlowerboundthm0equ3}
		||\mathcal F[\hat \mu]-\mathcal F[\mu]||_{\infty}= \max_{x\in [-\Omega, \Omega]} |\mathcal F(\gamma)(x)|, 
	\end{equation}
	where $\gamma=\sum_{j=1}^{2n-1}a_j \delta_{t_j}$ and 
	\begin{equation}\label{Taylorseries1}
		\mathcal F[\gamma](x)=\sum_{j=1}^{2n-1} a_j e^{it_j x}=\sum_{j=1}^{2n-1}a_j \sum_{k=0}^{\infty}\frac{(it_jx)^k}{k!}=\sum_{k=0}^{\infty}Q_{k}(\gamma)\frac{(ix)^k}{k!}. 
	\end{equation} 
	Here, $Q_{k}(\gamma)=\sum_{j=1}^{2n-1}a_j t_{j}^k$.
	By (\ref{equ:numberlowerboundthm0equ1}), we have 
	$Q_{k}(\gamma)=0, k=0,\cdots,2n-3$. We next estimate $Q_{k}(\gamma)$ for $k> 2n-3$.

	\textbf{Step 4.}	
	We estimate $\sum_{j=1}^{2n-1}|a_{j}|$ first.
	We begin by ordering $a_j$'s such that
	\[
	m_{\min}=|a_{j_1}|\leq |a_{j_2}|\leq  \cdots \leq |a_{j_{2n-1}}|.
	\]
	Then (\ref{equ:numberlowerboundthm0equ1}) implies that
	\[
	a_{j_1}\phi_{2n-3}(t_{j_1}) = \big(\phi_{2n-3}(t_{j_2}), \cdots, \phi_{2n-3}(t_{j_{2n-1}})\big)(-a_{j_2},\cdots, -a_{j_{2n-1}})^{\top},
	\]
	and hence
	\[
	a_{j_1} \left(\phi_{2n-3}(t_{j_2}), \cdots, \phi_{2n-3}(t_{j_{2n-1}})\right)^{-1}\phi_{2n-3}(t_{j_1}) = (-a_{j_2},\cdots, -a_{j_{2n-1}})^{\top}.
	\]
	Together with Lemma \ref{lem:invervandermonde}, we have 
	\[
	a_{j_1}  \Pi_{2\leq q\leq 2n-2}\frac{t_{j_1}-t_{j_q}}{t_{j_{2n-1}}-t_{j_q}}= -a_{j_{2n-1}}.
	\]
	Further, 
	\begin{equation}\label{equ:proofnumberlower0}
		\begin{aligned}
			\babs{a_{j_{2n-1}}}=& \babs{a_{j_1}}  \Pi_{2\leq q\leq 2n-2}\frac{\babs{t_{j_1}-t_{j_q}}}{\babs{t_{j_{2n-1}}-t_{j_q}}} = \babs{a_{j_1}}  \Pi_{2\leq q\leq 2n-2}\frac{\babs{t_{j_1}-t_{j_q}}}{\babs{t_{j_{2n-1}}-t_{j_q}}} \frac{\babs{t_{j_1}-t_{j_{2n-1}}}}{\babs{t_{j_{2n-1}}-t_{j_{1}}}}\\
			= &\babs{a_{j_1}} \frac{ \Pi_{2\leq q\leq 2n-1}\babs{t_{j_1}-t_{j_q}}}{ \Pi_{1\leq q\leq 2n-2}\babs{t_{j_{2n-1}}-t_{j_q}}}\leq \babs{a_{j_1}} \frac{\max_{j_1=1, \cdots, 2n-1} \Pi_{2\leq q\leq 2n-1}\babs{t_{j_1}-t_{j_q}}}{\min_{j_{2n-1}=1, \cdots, 2n-1} \Pi_{1\leq q\leq 2n-2}\babs{t_{j_{2n-1}}-t_{j_q}}}. 
		\end{aligned}
	\end{equation}
	Thus, based on the distribution of $t_j$'s and (\ref{equ:numberlowerboundthm0equ4}), we have
	\[
	|a_{j_{2n-1}}|\leq \frac{(2n-2)!}{\left((n-1)!\right)^2}|a_{j_1}|\leq  \frac{(2n-2)!}{\left((n-1)!\right)^2}m_{\min},
	\]
	and consequently,
	\begin{equation}\label{equ:numberlowerboundthm0equ2}
		\sum_{j=1}^{2n-1}|a_j| = 	\sum_{q=1}^{2n-1}|a_{j_q}| \leq (2n-1) |a_{j_{2n-1}}|\leq  \frac{(2n-1)!}{\left((n-1)!\right)^2}m_{\min}.
	\end{equation}
	It follows that for $k\geq 2n-2$,
	\begin{align*}
		|Q_{k}(\gamma)|=&|\sum_{j=1}^{2n-1}a_j t_j^{k}|\leq \sum_{j=1}^{2n-1}|a_j| \big((n-1)\tau\big)^{k}\leq  \frac{(2n-1)!}{\left((n-1)!\right)^2}m_{\min} \big((n-1)\tau\big)^{k}.
	\end{align*}
	
	\textbf{Step 5.} Using (\ref{Taylorseries1}), we have
	\begin{align*}
		\max_{x\in [-\Omega, \Omega]}\left|\mathcal F[\gamma](x)\right|
		\leq& \sum_{k\geq 2n-2}\frac{(2n-1)!}{\left((n-1)!\right)^2}m_{\min} \big((n-1)\tau\big)^{k}\frac{\Omega^k}{k!}\\
		= &\frac{(2n-1)!m_{\min}(n-1)^{2n-2}(\tau\Omega)^{2n-2}}{\left((n-1)!\right)^2(2n-2)!}\ \sum_{k=0}^{+\infty}\frac{(\tau\Omega)^{k}(2n-2)!(n-1)^k}{(k+2n-2)!}\\
		< &\frac{(2n-1)m_{\min}(n-1)^{2n-2}(\tau\Omega)^{2n-2}}{\left((n-1)!\right)^2}\ \sum_{k=0}^{+\infty}\left(\frac{\tau\Omega}{2}\right)^k\\
		\leq &\frac{(2n-1)m_{\min}(n-1)^{2n-2}(\tau\Omega)^{2n-2}}{\left((n-1)!\right)^2} \frac{1}{0.8} \quad \left(\text{by (\ref{equ:numberlowerboundequ1}), $\frac{\tau\Omega}{2}\leq 0.2$}\right)\\
		\leq& \frac{(2n-1)m_{\min}}{2\pi(n-1)}(e\tau\Omega)^{2n-2} \frac{1}{0.8}. \quad \Big(\text{by (\ref{stirlingformula0})}\Big)
	\end{align*}
	Finally, using (\ref{equ:numberlowerboundequ1}) and the inequality that $\frac{(2n-1)}{2\pi(n-1)}\frac{1}{0.8}<1$, we obtain
	\[
	\max_{x\in [-\Omega, \Omega]}\left|\mathcal F[\gamma](x)\right|<\sigma.
	\]	
	This completes the proof.
\end{proof}

\subsection{Proof of Theorem \ref{thm:supportlowerboundthm0}}
\begin{proof} Let $t_1=-(n-\frac{1}{2})\tau,t_2=-(n-\frac{3}{2})\tau,\cdots, t_n=-\frac{\tau}{2}, t_{n+1}=\frac{\tau}{2},\cdots, t_{2n}=(n-\frac{1}{2})\tau$. Consider the following system of linear equations: 
	\begin{equation}
		Aa=0,
	\end{equation}
	where $A=\big(\phi_{2n-2}(t_1),\cdots,\phi_{2n-2}(t_{2n})\big)$ with $\phi_{2n-2}(\cdot)$ being defined in (\ref{equ:phiformula}). Since $A$ is underdetermined, there exists a nontrivial solution $a=(a_1,\cdots,a_{2n})^{\top}$. By the linear independence of any $(2n-1)$ column vectors of $A$, all $a_j$'s are nonzero. By a scaling of $a$, we can assume that $a_{2n}<0$ and
	\begin{equation}\label{equ:supplowerboundthm0equ4}
		\min_{1\leq j\leq n}|a_{2j-1}|=m_{\min}.
	\end{equation}
	We define
	\[
	\mu=\sum_{j=1}^{n}a_{2j-1} \delta_{t_{2j-1}},\quad  \hat \mu=\sum_{j=1}^{n}-a_{2j}\delta_{t_{2j}}.
	\]
	Similar to  Step 2 in the proof of Theorem \ref{thm:numberlowerboundthm0}, we can show that $a_{2j-1}>0, j=1, \cdots,n,$ and $a_{2j}<0, j=1, \cdots, n$. Thus both $\hat \mu$ and $\mu$ are positive measures. Similar to  Step 4 in the proof of Theorem \ref{thm:numberlowerboundthm0}, we can show that 
	\begin{equation}\label{equ:supportlowerboundthm0equ1}	\sum_{j=1}^{2n}|a_j|\leq\frac{(2n)!}{n!(n-1)!}m_{\min}.
	\end{equation}
	We now prove that $$||\mathcal F[\hat \mu]-\mathcal F[\mu]||_{\infty}\leq \max_{x\in[-\Omega,\Omega]}|\mathcal F[\gamma](x)|<\sigma,$$
	where $\gamma=\sum_{j=1}^{2n}a_j \delta_{t_j}$. Indeed, (\ref{equ:supportlowerboundthm0equ1}) implies, for $k\geq 2n-1$,
	\begin{align*}
		|\sum_{j=1}^{2n}a_j t_j^{k}|\leq \sum_{j=1}^{2n}|a_j| ((n-1/2)\tau)^{k}\leq \frac{(2n)!}{n!(n-1)!}m_{\min}  \left((n-{1}/{2})\tau\right)^{k}.
	\end{align*}
	On the other hand, similar to expansion (\ref{Taylorseries1}), we can expand $\mathcal F[\gamma]$ and have 
	\begin{align*}
		Q_{k}(\gamma)=0,\  k=0,\cdots,2n-2 \quad \text{and}\ |Q_{k}(\gamma)|\leq\frac{(2n)!}{n!(n-1)!}m_{\min}\left((n-{1}/{2})\tau\right)^{k}, k\geq 2n-1. 
	\end{align*}
	Therefore, for $|x|\leq \Omega$,
	\begin{align*}
		\max_{x\in [-\Omega, \Omega]}\left|\mathcal F[\gamma](x)\right|\leq& \sum_{k\geq 2n-1}\frac{(2n)!}{n!(n-1)!}m_{\min} \left((n-{1}/{2})\tau\right)^{k}\frac{|x|^k}{k!}\leq \sum_{k\geq 2n-1}\frac{(2n)!}{n!(n-1)!}m_{\min} \left((n-{1}/{2})\tau\right)^{k}\frac{\Omega^k}{k!}\\
		=& \frac{(2n)!m_{\min}(n-1/2)^{2n-1}(\tau\Omega)^{2n-1}}{n!(n-1)!(2n-1)!}\ \sum_{k=0}^{+\infty}\frac{(\tau\Omega)^{k}(2n-1)!(n-1/2)^k}{(k+2n-1)!}\\
		<& \frac{2nm_{\min}(n-1/2)^{2n-1}(\tau\Omega)^{2n-1}}{n!(n-1)!}\ \sum_{k=0}^{+\infty}\left(\frac{\tau\Omega}{2}\right)^k\\
		=&\frac{{2n}m_{\min}(n-1/2)^{2n-1}(\tau\Omega)^{2n-1}}{n!(n-1)!} \frac{1}{0.8} \quad \Big(\text{(\ref{supportlowerboundequ0}) implies $\frac{\tau \Omega}{2} \leq 0.2 $}\Big)\\
		\leq& \frac{nm_{\min}(n-1/2)^{2n-1}}{\pi n^{n+\frac{1}{2}}(n-1)^{n-\frac{1}{2}}}(e\tau\Omega)^{2n-1} \frac{1}{0.8}\quad \Big(\text{by (\ref{stirlingformula0})}\Big)\\
		\leq & \frac{n}{\pi(n-1/2)} m_{\min}(e\tau\Omega)^{2n-1} \frac{1}{0.8}\\
		<& \sigma. \quad \Big(\text{by (\ref{supportlowerboundequ0}) and $ \frac{n}{\pi(n-1/2)}\frac{1}{0.8}<1$} \Big)
	\end{align*} 
	It follows that $||\mathcal F[\hat \mu]-\mathcal F[\mu]||_{\infty}<\sigma$.
\end{proof}

\subsection{Proof of Proposition \ref{thm:supportlowerboundthm1}}
\begin{proof} \textbf{Step 1.} For $j\in\{1,2,\cdots,2n\}$, set $t_j=-\frac{sn-2}{2}\tau+\frac{(j-2)s}{2}\tau$ if $j$ is even and $t_j=t_{4\lceil\frac{j+1}{4}\rceil-2}+(-1)^{\frac{j+1}{2}}\tau$ otherwise. Consider the following system of linear equations: 
	\begin{equation*}
		Aa=0,
	\end{equation*}
	where $A=\big(\phi_{2n-2}(t_1),\cdots,\phi_{2n-2}(t_{2n})\big)$ with $\phi_{2n-2}(\cdot)$ defined in (\ref{equ:phiformula}). Since $A$ is underdetermined, there exists a nontrivial solution $a=(a_1,\cdots,a_{2n})^{\top}$. Also, by the linear independence of any $(2n-1)$ column vectors of A, we can show that all $a_j$'s are nonzero. By a scaling of $a$, we can assume that $a_{2n}>0$ and
	\begin{equation}\label{equ:supportlowerboundthm1equ4}
		\min_{1\leq j\leq n}|a_{2j}|=m_{\min}.
	\end{equation}
	We define
	\[
	\mu=\sum_{j=1}^{n}a_{2j} \delta_{t_{2j}},\quad  \hat \mu=\sum_{j=1}^{n}-a_{2j-1}\delta_{t_{2j-1}}.
	\]
	Similar to  Step 2 in the proof of Theorem \ref{thm:numberlowerboundthm0}, we can show that $a_{2j-1}<0, j=1, \cdots,n,$ and $a_{2j}>0, j=1, \cdots, n$. Thus, both $\hat \mu$ and $\mu$ are positive measures.

	\textbf{Step 2.} 	
	We now estimate $\sum_{j=1}^{2n}|a_{j}|$. Reorder $a_j$ such that
	\[
	m_{\min}=|a_{j_1}|\leq |a_{j_2}|\leq  \cdots \leq |a_{j_{2n}}|.
	\]
	Similar to  Step 4 in the proof of Theorem  \ref{thm:numberlowerboundthm0}, we have
	\begin{equation}\label{equ:proofsupplowerdistri1}
		a_{j_1}  \Pi_{2\leq q\leq 2n-1}\frac{t_{j_1}-t_{j_q}}{t_{j_{2n}}-t_{j_q}}= -a_{j_{2n}}.
	\end{equation}
	We next estimate $\Pi_{2\leq q\leq 2n-1}\left|\frac{t_{j_1}-t_{j_q}}{t_{j_{2n}}-t_{j_q}}\right|$. Note that 
	\begin{equation}
		\label{equ:product-of-distance1}
		\begin{split}
			\Pi_{2\leq q\leq 2n-1}\frac{\left|t_{j_1}-t_{j_q}\right|}{\left|t_{j_{2n}}-t_{j_q}\right|} &= \bigg(\Pi_{2\leq q\leq 2n-1}\frac{\left|t_{j_1}-t_{j_q}\right|}{\left|t_{j_{2n}}-t_{j_q}\right|}\bigg)\cdot \frac{\left|t_{j_1}-t_{j_{2n}}\right|}{\left|t_{j_{2n}}-t_{j_1}\right|} = \frac{\Pi_{2\leq q\leq 2n}\left|t_{j_1}-t_{j_q}\right|}{\Pi_{1\leq q\leq 2n-1}\left|t_{j_{2n}}-t_{j_q}\right|}\\
			&\leq\frac{\max_{j\in\{1,2,\dots,2n\}}\Pi_{i\in\{1,2,\dots,2n\}, i\neq j}\left|t_{i}-t_{j}\right|}{\min_{j\in\{1,2,\dots,2n\}}\Pi_{i\in\{1,2,\dots,2n\}, i\neq j}\left|t_{i}-t_{j}\right|}.  
		\end{split}
	\end{equation}
	We separate $\{t_j\}_{j=1,2,\dots,2n}$ into four classes: $C_1 = \{t_{4j-2}\}_{j=1}^{\lceil\frac{n}{2}\rceil}, C_2 = \{t_{4j}\}_{j=1}^{\lfloor\frac{n}{2}\rfloor}, C_3 =\{t_{4j-3}\}_{j=1}^{\lceil\frac{n}{2}\rceil}, C_4 = \{t_{4j-1}\}_{j=1}^{\lfloor\frac{n}{2}\rfloor}$; See Figure \ref{fig:locationdistribution} for an illustration. The points in each class are evenly-spaced, by which we can estimate the right-hand side of (\ref{equ:product-of-distance1}). 
	\begin{figure}[!h]
		\includegraphics[width=0.7\textwidth]{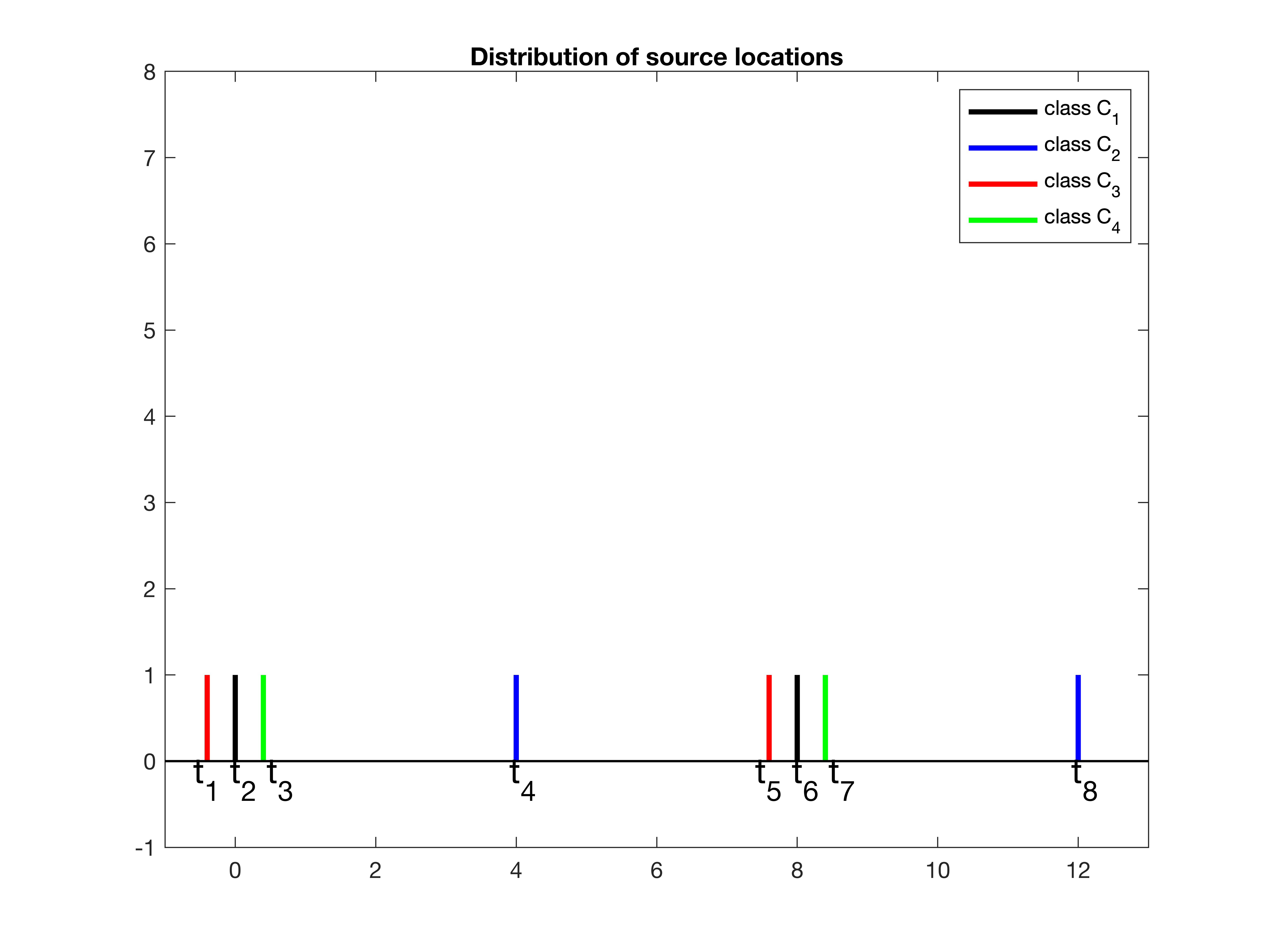}
		\centering
		\caption{Distribution of source locations.}
		\label{fig:locationdistribution}
	\end{figure}
	
	Note that
	\begin{equation}
		\label{equ:nominator-max}
		\begin{split}
			&\max_{j\in\{1,2,\dots,2n\}}\Pi_{p\in\{1,2,\dots,2n\}, p\neq j}|t_{p}-t_{j}|= \max_{k=1,2,3,4}\bigg(\max_{x\in C_k}\Pi_{y\in\{t_1,t_2\dots,t_{2n}\},y\neq x}|x-y|\bigg)\\
			\leq &  \max_{k=1,2,3,4}\bigg(\max_{x\in C_k}\Pi_{p=1,2,3,4}\Pi_{y\in C_p,y\neq x}|x-y|\bigg)=: \max_{k=1,2,3,4} c^{\max}_k,
		\end{split}
	\end{equation}
	and
	\begin{equation}
		\label{equ:denominator-min}
		\begin{split}
			&\min_{j\in\{1,2,\dots,2n\}}\Pi_{p\in\{1,2,\dots,2n\}, p\neq j}|t_{p}-t_{j}|= \min_{k=1,2,3,4}\bigg(\min_{x\in C_k}\Pi_{y\in\{t_1,t_2\dots,t_{2n}\},y\neq x}|x-y|\bigg)\\
			\geq &  \min_{k=1,2,3,4}\bigg(\min_{x\in C_k}\Pi_{p=1,2,3,4}\Pi_{y\in C_p,y\neq x}|x-y|\bigg)=: \min_{k=1,2,3,4} c^{\min}_k.
		\end{split}
	\end{equation}
	The estimates of $\max_{k=1,2,3,4} c^{\max}_k$ and $\min_{k=1,2,3,4} c^{\min}_k$ are detailed in Lemmas \ref{lem:c-max} and \ref{lem:c-min} in Appendix \ref{appendixA}. With the aid of them we control the left-hand side of \eqref{equ:product-of-distance1} that
	\begin{equation}
		\label{equ:product-of-distance}
		\begin{split}
			&\quad \Pi_{2\leq q\leq 2n-1}\frac{|t_{j_1}-t_{j_q}|}{|t_{j_{2n}}-t_{j_q}|} \leq \frac{\tau^{2n-1} s^{2n-1}\left(2\lceil\frac{n}{2}\rceil\right)!\left(2\lceil\frac{n}{2}\rceil-1\right)!}{\tau^{2n-1}s^{2n-3}\cdot \bigg((2\lfloor\frac{\lfloor\frac{n}{2}\rfloor}{2}\rfloor-1)!\bigg)^4}\leq  \frac{s^2e^{11}}{\pi^2}(n+1)^{10}2^{2n-8},
		\end{split}
	\end{equation}
	where the last inequality is obtained by Lemma \ref{numberlowerboundcalculate1} in the following step.\\
	\textbf{Step 3.}
	\begin{lem}\label{numberlowerboundcalculate1}
		For $n\geq 2$, we have
		\begin{align*}
			\frac{\left(2\lceil\frac{n}{2}\rceil\right)!\left(2\lceil\frac{n}{2}\rceil-1\right)!}{ \bigg((2\lfloor\frac{\lfloor\frac{n}{2}\rfloor}{2}\rfloor-1)!\bigg)^4}\leq \frac{e^{11}}{\pi^2}(n+1)^{10}2^{2n-8}.
		\end{align*}
	\end{lem}
	\begin{proof}
		Recall the Stirling approximation of factorial, that is, 
		\begin{equation}\label{stirlingformula}
			\sqrt{2\pi} n^{n+\frac{1}{2}}e^{-n}\leq n! \leq e n^{n+\frac{1}{2}}e^{-n}.
		\end{equation}
		For $n\leq 11$, the inequality can be checked by calculation. For $n>11$, we have 
		\begin{align*}
			&\frac{\left(2\lceil\frac{n}{2}\rceil\right)!\left(2\lceil\frac{n}{2}\rceil-1\right)!}{ \bigg((2\lfloor\frac{\lfloor\frac{n}{2}\rfloor}{2}\rfloor-1)!\bigg)^4}
			\leq\frac{ e^{2-4\lceil\frac{n}{2}\rceil+1}(2\lceil\frac{n}{2}\rceil)^{2\lceil\frac{n}{2}\rceil+\frac{1}{2}}\cdot(2\lceil\frac{n}{2}\rceil-1)^{2\lceil\frac{n}{2}\rceil-\frac{1}{2}}}{(\sqrt{2\pi})^4 (2\lfloor\frac{\lfloor\frac{n}{2}\rfloor}{2}\rfloor-1)^{4(2\lfloor\frac{\lfloor\frac{n}{2}\rfloor}{2}\rfloor-1)+2}e^{-4(2\lfloor\frac{\lfloor\frac{n}{2}\rfloor}{2}\rfloor-1)}}\\
			\leq& \frac{1}{4\pi^2}\frac{e^{3-2n}}{e^{-4(2\cdot\frac{n}{4}-1)}}\cdot\frac{(2(\frac{n}{2}+\frac{1}{2}))^{2n+2}}{(2(\frac{n}{4}-\frac{3}{4})-1)^{4(2(\frac{n}{4}-\frac{3}{4})-1)+2}}= \frac{1}{4e\pi^2}\frac{(n+1)^{2n+2}}{(\frac{n}{2}-\frac{5}{2})^{2n-8}}\\
			= & \frac{1}{4e\pi^2}(n+1)^{10}\left(2+\frac{6}{\frac{n}{2}-\frac{5}{2}}\right)^{2n-8}=  \frac{1}{4e\pi^2}(n+1)^{10}2^{2n-8}\left(1+\frac{6}{n-5}\right)^{12\cdot\frac{n-5}{6}+2}\\
			\leq & \frac{e^{11}}{4\pi^2}\left(1+\frac{6}{n-5}\right)^2(n+1)^{10}2^{2n-8} \leq \frac{e^{11}}{\pi^2}(n+1)^{10}2^{2n-8}.
		\end{align*}
	\end{proof}

	\textbf{Step 4.}
	Thus, combined (\ref{equ:proofsupplowerdistri1}) and (\ref{equ:product-of-distance}), we have
	\[
	\babs{a_{j_{2n}}}\leq \frac{e^{11}s^2}{\pi^2}(n+1)^{10}2^{2n-8}\babs{a_{j_1}},
	\]
	and consequently,
	\begin{equation}
		\sum_{j=1}^{2n}|a_j| \leq \sum_{q=1}^{2n}|a_{j_q}|\leq 2n  \frac{e^{11}s^2}{\pi^2}(n+1)^{10}2^{2n-8}m_{\min}.
	\end{equation}
	It then follows that for $k\geq 2n-1$,
	\begin{align*}
		\left|\sum_{j=1}^{2n-1}a_j t_j^{k}\right|\leq \sum_{j=1}^{2n-1}|a_j| \left(\frac{sn}{2}\tau\right)^{k}
		\leq&\frac{2n(n+1)^{10} e^{11}s^2}{\pi^2}2^{2n-8}m_{\min} \left(\frac{sn}{2}\tau\right)^{k}.
	\end{align*}
	
	\textbf{Step 5.}
	We now prove that $$||\mathcal F[\hat \mu]-\mathcal F[\mu]||_{\infty}\leq \max_{x\in[-\Omega,\Omega]}|\mathcal F[\gamma](x)|<\sigma,$$
	where $\gamma=\sum_{j=1}^{2n}a_j \delta_{t_j}$. On the other hand, similar to expansion (\ref{Taylorseries1}), we can expand $\mathcal F[\gamma]$ and have
	\begin{align*}
		Q_{k}(\gamma)=0,\  k=0,\cdots,2n-2 \quad \text{and}\ |Q_{k}(\gamma)|\leq \frac{2n(n+1)^{10} e^{11}s^22^{2n-8}}{\pi^2}m_{\min} \left(\frac{sn}{2}\tau\right)^{k}, \quad k\geq 2n-1. 
	\end{align*}
	Therefore, for $|x|\leq \Omega$ and $n\geq 6$, we have 
	\begin{align*}
		|\mathcal F[\gamma](x)|\leq& \sum_{k\geq 2n-1}\frac{2n(n+1)^{10} e^{11}s^2}{\pi^2}2^{2n-8}m_{\min} \left(\frac{sn}{2}\tau\right)^{k}\frac{|x|^k}{k!}\\
		\leq& \sum_{k\geq 2n-1}\frac{2n(n+1)^{10} e^{11}s^2}{\pi^2}2^{2n-8}m_{\min} \left(\frac{sn}{2}\tau\right)^{k}\frac{\Omega^k}{k!}\\
		\leq& \frac{2e^{11}}{\pi^2}\cdot\frac{s^2n(n+1)^{10}2^{2n-8}m_{\min}\cdot (\frac{sn\tau\Omega}{2})^{2n-1}}{(2n-1)!} \sum_{k=0}^{+\infty}\frac{(\frac{s\tau\Omega}{2})^{k}(2n-1)!n^k}{(k+2n-1)!}\\
		<&\frac{e^{11}}{2^6\pi^2}\cdot\frac{(\tau\Omega)^{2n-1}s^{2n+1}m_{\min}n(n+1)^{10}n^{2n-1}}{(2n-1)!}\cdot \sum_{k=0}^{+\infty}\left(\frac{s\tau\Omega}{2}\right)^k  \\
		\leq& \frac{e^{11}}{2^6\pi^2} (\tau\Omega)^{2n-1}s^{2n+1}m_{\min} 
		\frac{(n+1)^{10}n^{2n}e^{2n-1}}{\sqrt{2\pi}(2n-1)^{2n-\frac{1}{2}}} \sum_{k=0}^{+\infty}\left(\frac{s\tau\Omega}{2}\right)^k \quad \Big(\text{by Stirling's formula (\ref{stirlingformula0})}\Big)\\
		\leq& 38\cdot 0.2^{2n-1}\sigma\cdot \frac{n(n+1)^{10}}{\sqrt{2n-1}}\frac{n^{2n-1}}{(2n-1)^{2n-1}} \sum_{k=0}^{+\infty}\left(\frac{s\tau\Omega}{2}\right)^k\quad \Big(\text{by separation condition (\ref{supportlowerboundequ1})}\Big)\\
		= &38\cdot 0.1^{2n-1}\sigma\cdot \frac{n(n+1)^{10}}{\sqrt{2n-1}}\left(\frac{n}{n-\frac{1}{2}}\right)^{2n-1} \frac{1}{1-0.04}\quad \Big(\text{(\ref{supportlowerboundequ1}) implies $\frac{s\tau \Omega}{2} \leq 0.04 $}\Big)\\
		\leq&  38e\cdot 0.1^{2n-1}\sigma\cdot \frac{n(n+1)^{10}}{\sqrt{2n-1}} \frac{1}{1-0.04}\\
		<& \sigma \quad \left(\text{$38e\cdot 0.1^{2n-1}\cdot \frac{n(n+1)^{10}}{\sqrt{2n-1}} \frac{1}{1-0.04}<1$ for $n\geq 6$}\right).
	\end{align*} 
	It then follows that $||\mathcal F[\hat \mu]-\mathcal F[\mu]||_{\infty}<\sigma$. 
	
	Now consider the case when $2\leq n\leq 5$. By (\ref{equ:product-of-distance}), we have 
	\[
	\Pi_{2\leq q\leq 2n-1}\frac{|t_{j_1}-t_{j_q}|}{|t_{j_{2n}}-t_{j_q}|} \leq \frac{s^{2}\left(2\lceil\frac{n}{2}\rceil\right)!\left(2\lceil\frac{n}{2}\rceil-1\right)!}{ \bigg((2\lfloor\frac{\lfloor\frac{n}{2}\rfloor}{2}\rfloor-1)!\bigg)^4},
	\]
	and consequently, 
	\[
	\sum_{j=1}^{2n}|a_j|\leq 2n \frac{s^{2}\left(2\lceil\frac{n}{2}\rceil\right)!\left(2\lceil\frac{n}{2}\rceil-1\right)!}{ \bigg((2\lfloor\frac{\lfloor\frac{n}{2}\rfloor}{2}\rfloor-1)!\bigg)^4}.
	\]
	By similar arguments as those for the case when $n\geq 6$, we can show that for $2\leq n\leq 5$,
	\[
	|\mathcal F[\gamma](x)|<\sigma. 
	\]
	
\end{proof}

\section{Proof of results in Section \ref{section:positivitynotenhance}}

\subsection{Proof of Theorem \ref{thm:twopointresolution0}}
\begin{proof}
	\textbf{Step 1.} We first prove the one-dimensional case. Let $\mu = \sum_{j=1}^2a_j \delta_{y_j} , a_j = m_{\min} e^{i\theta_j}$ and $\hat \mu = a \delta_{\hat y}$. A crucial relation is
	\begin{equation}\label{equ:prooftwopointlimit1}
		\mathcal F[\hat \mu](\omega)= \mathcal F[\mu](\omega) +\vect w_1(\omega), \ |\vect w_1(\omega)|< 2\sigma, \ \omega \in [-\Omega, \Omega].
	\end{equation}
	Note that if (\ref{equ:prooftwopointlimit1}) holds, $\hat \mu$ can be a $\sigma$-admissible measure of some $\vect Y$ generated by model (\ref{equ:modelsetting2}). This time, resolving two point sources is impossible. Conversely, if (\ref{equ:prooftwopointlimit1}) does not hold, $\hat \mu$ cannot be any $\sigma$-admissible measure of some $\vect Y$ generated by $\mu$ as in model (\ref{equ:modelsetting2}). Thus the resolution limit $\mathcal R(\theta)$ is the constant such that (\ref{equ:prooftwopointlimit1}) holds when $|y_1-y_2|< \mathcal R(\theta)$ and fails to hold in the opposite case. 
	
	\textbf{Step 2.} Note that for the general source locations $y_1, y_2$, shifting them by $x$ and get that
	\[
	\mathcal F[\hat \mu](\omega)e^{ix\omega}= \mathcal F[\mu](\omega)e^{ix\omega} +\vect w_1(\omega)e^{ix\omega}, \ |\vect w_1(\omega)e^{ix\omega}|< 2\sigma, \ \omega \in [-\Omega, \Omega],
	\]
	we can transform the problem into the case when $y_1=-y_2$. Thus we consider that the underlying source is $\mu = m_{\min}e^{i\frac{\theta}{2}} \delta_{y_1}+ m_{\min}e^{i\frac{-\theta}{2}}\delta_{y_2}$ with $y_1>0, y_1 = -y_2$. The measure $\hat \mu$ is $a\delta_{\hat y}$ with $a$ and $\hat y$ to be determined.  
	
	From (\ref{equ:prooftwopointlimit1}), we get that 
	\begin{align*}
		&\vect w_1(\omega) = ae^{i \hat y\omega}-m_{\min}(e^{i\frac{\theta}{2}}e^{iy_1\omega}+e^{i\frac{-\theta}{2}}e^{i y_2\omega})= ae^{i\hat 
			y\omega}-2m_{\min}\cos\left(y_1\omega+\frac{\theta}{2}\right). 
	\end{align*}
	Note that for two non-negative values $x,y$, we have 
	\begin{equation}\label{equ:prooftwopointlimit-1}
	\left|xe^{iq}-y\right|^2 = (x\cos(q)-y)^2+x^2\sin^2(q) = x^2+y^2-2xy\cos(q) \geq (x-y)^2 
	\end{equation}
	and the equality is attained when $q=0$. We only consider the case when
	\begin{equation}\label{equ:prooftwopointlimit0}
	0\leq y_1\Omega+ \frac{\theta}{2}\leq \frac{\pi}{2}\text{ and }-\frac{\pi}{2}\leq -y_1\Omega+ \frac{\theta}{2}\leq 0
	\end{equation}
and we shall see that this coincides with the case in the theorem. By the above condition, we have $\cos(y_1\omega+\frac{\theta}{2})\geq 0, \omega \in [-\Omega,\Omega]$. Thus by (\ref{equ:prooftwopointlimit-1}), for every $\omega$, 
	\[
	|\vect w_1(\omega)|\geq \left||a|-2m_{\min}\cos\left(y_1\omega+\frac{\theta}{2}\right)\right|
	\]
	and the minimum is attained when $\hat y=0$ and $a$ is a positive number. We now try to find the condition on $y_1$ so that there exists $a$ satisfying
	\[
	\left||a|-2m_{\min}\cos\left(y_1\omega+\frac{\theta}{2}\right)\right|<2\sigma, \quad \omega\in [-\Omega,\Omega]. 
	\]
	This is equivalent to
	\begin{equation}\label{equ:prooftwopointlimit3}
		\max_{ \omega, \omega' \in [-\Omega,\Omega]}\left|2m_{\min}\left(\cos\left(y_1\omega+\frac{\theta}{2}\right)- \cos\left(y_1\omega'+\frac{\theta}{2}\right)\right)\right|<4\sigma.
	\end{equation}

We then analyze the problem for two different cases. We denote $d_{\min}:=|y_1-y_2|$ and  now the condition (\ref{equ:prooftwopointlimit0}) is
\begin{equation}\label{equ:prooftwopointlimit5}
0\leq \frac{d_{\min}\Omega+\theta}{2}\leq \frac{\pi}{2} \text{ and } \frac{-d_{\min}\Omega+\theta}{2}\leq 0.
\end{equation}  
Under this condition, problem (\ref{equ:prooftwopointlimit3}) becomes 
\[
2m_{\min}\left|1-\cos\left(\frac{d_{\min}}{2}\Omega+\frac{\theta}{2}\right)\right|< 4\sigma. 
\]
Thus $4\sin\left(\frac{d_{\min}\Omega+\theta}{4}\right)^2<\frac{4\sigma}{m_{\min}}$, and equivalently
\[
d_{\min} < \frac{4\arcsin\left(\left(\frac{\sigma}{m_{\min}}\right)^{\frac{1}{2}}\right)-\theta }{\Omega}.
\]
By the above discussions, in this case $\mathcal R(\theta)$ is
\[
\mathcal R(\theta) = \frac{4\arcsin\left(\left(\frac{\sigma}{m_{\min}}\right)^{\frac{1}{2}}\right)-\theta }{\Omega}.
\]
Now the condition (\ref{equ:prooftwopointlimit5}) holds when 
\[
\frac{4\arcsin\left(\left(\frac{\sigma}{m_{\min}}\right)^{\frac{1}{2}}\right)-\theta }{\Omega} \Omega + \theta \leq \pi, \text{and} -\left(\frac{4\arcsin\left(\left(\frac{\sigma}{m_{\min}}\right)^{\frac{1}{2}}\right)-\theta }{\Omega}\right)\Omega+\theta \leq 0. 
\]
which further holds when
\[
\sin\left(\frac{\theta}{2}\right)^2 \leq \frac{\sigma}{m_{\min}}\leq \frac{1}{2}.
\]
This proves the first case in the theorem.

 Now, we consider the case when $\frac{\sigma}{m_{\min}}>\frac{1}{2}$. 
	We choose the specific case where $a=m_{\min}e^{i\theta_1}$ and $\hat y =y_1$. Then 
	\[
	\vect w_1(\omega) =m_{\min}e^{i\theta_1}e^{iy_1\omega}-(m_{\min}e^{i\theta_1}e^{iy_1\omega}+m_{\min} e^{i\theta_2}e^{i y_2\omega}) =m_{\min}e^{i\theta_2}e^{iy_2\omega}, \quad \omega \in [-\Omega,\Omega]. 
	\]
	Condition $\frac{\sigma}{m_{\min}}>\frac{1}{2}$ gives
	\[
	|\vect w_1(\omega)|<2\sigma.
	\]
	Thus the case when $\frac{\sigma}{m_{\min}}>\frac{1}{2}$ is meaningless. Indeed, there are always some $\sigma$-admissible measures for some images with only one point source.

	\textbf{Step 3.} Now we consider the case when the  sources $\vect y_j$'s are in $\mathbb R^k$. We still consider the crucial relation that
	\begin{equation}\label{equ:prooftwopointlimit4}
		\mathcal F[\hat \mu](\vect \omega)= \mathcal F[\mu](\vect \omega) +\vect w_1(\vect \omega), \ |\vect w_1(\vect \omega)|< 2\sigma, \ \bnorm{\vect \omega}_{2}\leq \Omega. 
	\end{equation}
	By a similar argument as the one in step 1, we know that the resolution limit $\mathcal R(\theta)$ is the constant such that (\ref{equ:prooftwopointlimit4}) holds when $\bnorm{\vect y_1-\vect y_2}_2< \mathcal R(\theta)$ and fails to hold in the opposite case. Note that by choosing suitable axes or transforming the problem, we can make $\vect y_1= (y_1, 0, \cdots, 0)^\top, \vect y_2= (y_2, 0, \cdots, 0)^\top$. Consider $\hat \mu = a \delta_{\mathbf{\hat y}}, \mathbf{\hat y} \in \mathbb R^k$ with $a$ and $\mathbf{\hat y}$ to be determined. We now have 
	\[
	\mathcal F[\hat \mu](\vect \omega)- \mathcal F[\mu](\vect \omega) = ae^{i \mathbf{\hat y} \cdot \vect \omega} - \sum_{j=1}^{2}a_j e^{i \vect {y}_j \cdot \vect{\omega}}= ae^{i \mathbf{\hat y}_{2:k}\cdot \vect \omega_{2:k}} e^{\mathbf{\hat y}_1 \vect \omega_1} - \sum_{j=1}^{2}a_j e^{i {y}_j \vect{\omega}_1}.  
	\]
	Thus analyzing when (\ref{equ:prooftwopointlimit4}) holds can be reduced to the one-dimensional case and it is not hard to see the result for the one-dimensional space still holds for multi-dimensional spaces.

\end{proof}

\subsection{Proof of Theorem \ref{thm:computatwopointresolution0}}
	\begin{proof}
	\textbf{Step 1.} We only need to analyze the case when $\frac{\sigma}{m_{\min}} \leq \frac{1}{2}$, as the case when $\frac{\sigma}{m_{\min}} >\frac{1}{2}$ is trivial. Also, we only consider the one-dimensional case since the treatment for multi-dimensional spaces is similar to the one in the proof of Theorem \ref{thm:twopointresolution0}. 
	
	Similarly to step 1 in the proof of Theorem \ref{thm:twopointresolution0}, the resolution limit $\mathcal D_{k, num}$ should be the constant such that the following estimate:
	\begin{equation}\label{equ:proofcomplextwopointlimit0}
		\mathcal F[\hat \mu](\omega)= \mathcal F[\mu](\omega) +\vect w_2(\omega), \ |\vect w_2(\omega)|< 2\sigma, \ \omega \in [-\Omega, \Omega],
	\end{equation}
	holds when $|y_1-y_2|< \mathcal D_{k, num}$ and fails to hold in the opposite case. We shall prove that when $	\sin\left(\frac{\theta}{2}\right)^2 \leq \frac{\sigma}{m_{\min}}\leq \frac{1}{2}$, if
	\[
	\left|y_1-y_2\right| \geq \frac{4 \arcsin \left(\left(\frac{\sigma}{m_{\min }}\right)^{\frac{1}{2}}\right)}{\Omega},
	\]
	then (\ref{equ:proofcomplextwopointlimit0}) doesn't hold for any $\hat{\mu}$ consisting of only one source. On the opposite case, Theorem \ref{thm:twopointresolution0} already ensures the existence of such $\hat{\mu}$ making
	\[
	\left|\mathbf{w}_2(\omega)\right|<2 \sigma, \quad \omega \in[-\Omega, \Omega] .
	\]
    This is enough to prove the theorem.
	
	\textbf{Step 2.} 
	Without loss of generality, we assume the underlying source is $$\mu = m_{\min}\alpha e^{-i\beta} \delta_{y_1}+ m_{\min}e^{i\beta} \delta_{y_2}$$ with $y_1 = -y_2$, $0< y_1\leq \frac{\pi}{2\Omega}$, $\alpha \geq 1$ and $0\leq \beta\leq \frac{\pi}{4}$. It is not hard to see that the other cases can all be transformed to the above setting or the same analysis as follows can be applied to. We consider $\hat \mu = ae^{i\gamma}\delta_{\hat y}$ with $a>0$, $\gamma$ and $\hat y$ to be determined. 
	
	From (\ref{equ:proofcomplextwopointlimit0}), we have 
	\begin{align*}
		\vect w_2(\omega) &= ae^{i\gamma}e^{i \hat y \omega}-m_{\min}\left(\alpha e^{-i\beta}e^{iy_1\omega}+ e^{i\beta}e^{-iy_1\omega}\right). 
	\end{align*}
	We rewrite it as 
	\begin{align}\label{equ:proofcomplextwopointlimit-1}
		\vect w_2(\omega) =& ae^{i\gamma}e^{i \hat y \omega}-m_{\min}\left(\alpha e^{i(y_1 \omega-\beta)}+ e^{i(\beta-y_1 \omega)}\right).   \nonumber \\
		=  & ae^{i\gamma}e^{i \hat y \omega} - h m_{\min} e^{i(y_1 \omega-\beta)} - 2m_{\min} \cos(y_1\omega - \beta)
	\end{align}
where $\alpha = m_{\min}(1+h)$.  We next prove the theorem by considering the case when 
	\begin{equation}\label{equ:proofcomplextwopointlimit1}
		y_1\Omega \geq \beta,\quad  0\leq \Omega y_1 +\beta \leq  \frac{\pi}{2}.
	\end{equation} 
We consider the necessary condition for the existence of such $\hat \mu, \mu$ satisfying (\ref{equ:proofcomplextwopointlimit0}) that 
	\[
	\min_{a>0, \alpha\geq 1, \gamma \in \mathbb R,  \hat y\in \mathbb R}\babs{\vect w_2(-\Omega)}+\babs{\vect w_2(\frac{\beta}{y_1})}<4\sigma.    
	\]
By (\ref{equ:proofcomplextwopointlimit-1}), it is
	\begin{align}\label{equ:proofcomplextwopointlimit2}
		&\min_{a>0, h\geq 0, \gamma \in \mathbb R, \hat y\in \mathbb R}\left|ae^{i\gamma}e^{-i\hat y \Omega} -h m_{\min} e^{-i(y_1 \Omega+\beta)}- 2m_{\min} \cos(y_1 \Omega+\beta )\right|+\left|ae^{i\gamma}e^{i \hat y \frac{\beta}{y_1}}-(2m_{\min}+hm_{\min})\right|.
	\end{align}
	A key observation is that if $2m_{\min}\cos(y_1\Omega+\beta )+hm_{\min}\leq a\leq 2m_{\min}+hm_{\min}$, we have 
	\begin{align}
		&\min_{a>0, h\geq 0, \gamma \in \mathbb R, \hat y\in \mathbb R}\left|ae^{i\gamma}e^{-i\hat y\Omega} -h m_{\min} e^{-i(y_1 \Omega+\beta)}- 2m_{\min} \cos(y_1 \Omega+\beta )\right|\nonumber +\left|ae^{i\gamma}e^{i \hat y \frac{\beta}{y_1}}-(2m_{\min}+hm_{\min})\right|\nonumber \\
		\geq &\min_{a>0, h\geq 0, \gamma \in \mathbb R, \hat y\in \mathbb R, \hat x\in \mathbb R}\left|ae^{-i\hat y\Omega} -h m_{\min} e^{-i\hat x \Omega}- 2m_{\min} \cos(y_1 \Omega+\beta)\right|+|ae^{i\gamma}-(2m_{\min}+hm_{\min})|\nonumber \\
		= &\min_{a> 0, h\geq 0}\left|a -h m_{\min}- 2m_{\min} \cos(y_1 \Omega+\beta)\right|+\left|a-(2m_{\min}+hm_{\min})\right|\nonumber \\
		= &\min_{b\in \mathbb R,\ 2m_{\min}\cos(y_1\Omega+\beta )\leq b\leq 2m_{\min}}\left|b- 2m_{\min} \cos(y_1 \Omega+\beta )\right|+|2m_{\min}-b|\nonumber \\
		=&2m_{\min}-2m_{\min} \cos(y_1 \Omega+\beta ), \label{equ:proofcomplextwopointlimit7}
	\end{align}
	where the first equality is because $2m_{\min}\cos(y_1\Omega)+hm_{\min}\leq a\leq 2m_{\min}+hm_{\min}$ and $2m_{\min}\cos(y_1\Omega+\beta )\geq 0$ by (\ref{equ:proofcomplextwopointlimit1}). 
	
	On the other hand, letting $h=0, \hat y =0, \gamma =0$, we have 
	\begin{align*}
		&\min_{a>0}\left|ae^{i\gamma}e^{-i\hat y \Omega} -h m_{\min} e^{-i(y_1 \Omega+\beta)}- 2m_{\min} \cos(y_1 \Omega+\beta )\right|+\left|ae^{i\gamma}e^{i \hat y \frac{\beta}{y_1}}-(2m_{\min}+hm_{\min})\right|.\\
		=&\min_{a>0}\left|a - 2m_{\min} \cos(y_1 \Omega+\beta)\right|+|a-2m_{\min}|\\
		=&\min_{a>0}(a - 2m_{\min} \cos(y_1 \Omega+\beta ))+2m_{\min}-a \quad \big(\text{choose $2m_{\min} \cos(y_1 \Omega+\beta )\leq a\leq 2m_{\min}$}\big)\\ 
		=&2m_{\min}-2m_{\min} \cos(y_1 \Omega+\beta ).
	\end{align*}
	Together with (\ref{equ:proofcomplextwopointlimit7}), this yields 
	\begin{align*}
		&\min_{a>0, h\geq 0, \gamma \in \mathbb R, \hat y\in \mathbb R}\left|ae^{i\gamma}e^{-i\hat y \Omega} -h m_{\min} e^{-i(y_1 \Omega+\beta)}- 2m_{\min} \cos(y_1 \Omega+\beta)\right|+\left|ae^{i\gamma}e^{i \hat y \frac{\beta}{y_1}}-(2m_{\min}+hm_{\min})\right|\\
		=&2m_{\min}-2m_{\min} \cos(y_1 \Omega+\beta),
	\end{align*}
	in the case when $2m_{\min}\cos(y_1\Omega+\beta)+hm_{\min}\leq a\leq 2m_{\min}+hm_{\min}$.
	
	\medskip
	Now, we consider the case when $a< 2m_{\min}\cos(y_1\Omega+\beta)+hm_{\min}$. In this case, we have 
	\begin{align*}
		&\min_{a>0, h\geq 0, \gamma \in \mathbb R, \hat y\in \mathbb R}\left|ae^{i\gamma}e^{-i\hat y \Omega} -h m_{\min} e^{-i(y_1 \Omega+\beta)}- 2m_{\min} \cos(y_1 \Omega+\beta)\right|+\left|ae^{i\gamma}e^{i \hat y \frac{\beta}{y_1}}-(2m_{\min}+hm_{\min})\right|\\
		\geq & \min_{a>0, h\geq 0, \gamma \in \mathbb R, a< 2m_{\min}\cos(y_1\Omega+\beta )+hm_{\min}}\babs{ae^{i\gamma}-(2m_{\min}+hm_{\min})}\\
		\geq& \min_{a>0, h\geq 0, a< 2m_{\min}\cos(y_1\Omega+\beta )+hm_{\min}}\babs{a-(2m_{\min}+hm_{\min})}\\
		=& \min_{a>0, h\geq 0, a< 2m_{\min}\cos(y_1\Omega+\beta )+hm_{\min}}2m_{\min}+hm_{\min}-a\\
		>& 2m_{\min} -2m_{\min}\cos(y_1 \Omega+\beta ).
	\end{align*}
	
	Finally, we consider the case when $a> 2m_{\min}+hm_{\min}$. In this case, we have 
	\begin{align*}
		&\min_{a>0, h\geq 0, \gamma \in \mathbb R, \hat y\in \mathbb R}\left|ae^{i\gamma}e^{-i\hat y \Omega} -h m_{\min} e^{-iy_1 \Omega}- 2m_{\min} \cos(y_1 \Omega+\beta)\right|+\left|ae^{i\gamma}e^{i \hat y \frac{\beta}{y_1}}-(2m_{\min}+hm_{\min})\right|\\
		\geq & \min_{a>0, h\geq 0, \gamma \in \mathbb R, \hat y\in \mathbb R, a> 2m_{\min}+hm_{\min}}\left|ae^{i\gamma}e^{-i\hat y \Omega} -h m_{\min} e^{-iy_1 \Omega}- 2m_{\min} \cos(y_1 \Omega+\beta)\right|\\
		\geq& \min_{a>0, h\geq 0, a> 2m_{\min}+hm_{\min}}a-(2m_{\min}\cos(y_1 \Omega+\beta )+hm_{\min})\\
		>& 2m_{\min} -2m_{\min}\cos(y_1 \Omega+\beta).
	\end{align*}

	Therefore, combining all the above discussions, we arrive at
	\begin{align*}
		&\min_{a>0, h\geq 0, \gamma \in \mathbb R, \hat y\in \mathbb R}\left|ae^{i\gamma}e^{-i\hat y \Omega} -h m_{\min} e^{-iy_1 \Omega}- 2m_{\min} \cos(y_1 \Omega+\beta)\right|+\left|ae^{i\gamma}e^{i \hat y \frac{\beta}{y_1}}-(2m_{\min}+hm_{\min})\right|\\
		=&2m_{\min}-2m_{\min} \cos(y_1 \Omega+\beta).
	\end{align*}
	Thus (\ref{equ:proofcomplextwopointlimit2}) is equivalent to 
	\[
	2m_{\min}-2m_{\min} \cos(y_1 \Omega+\beta)<4\sigma.
	\]
	Similar to the proof of Theorem \ref{thm:twopointresolution0}, this yields
	\[
	d_{\min} < \frac{4\arcsin\left(\left(\frac{\sigma}{m_{\min}}\right)^{\frac{1}{2}}\right)-\theta}{\Omega}.
	\]
	where $\theta = 2\beta$ in our setting. The condition (\ref{equ:proofcomplextwopointlimit1}) now holds when
	\[
	\sin\left(\frac{\theta}{2}\right)^2 \leq \frac{\sigma}{m_{\min}}\leq \frac{1}{2}.
	\]
	This completes the proof.

\end{proof}

\appendix
\section{Auxiliary lemmas} \label{appendixA}
The following results can be easily proved. 
\begin{lem}
\label{lem:produc-min-max-2}
Let $n\in\N^+$ and $\tau >0$ and let $\mathcal{C}=\{j=1,2,\dots,n\}$. Then

(1):
\begin{equation}
\label{equ:argmin-on-product-2}
\lceil\frac{n}{2}\rceil\in\arg\min_{z\in \mathcal{C}}\Pi_{x\in\mathcal{C},x\neq z}|x-z|,
\end{equation} and
\begin{equation}
\label{equ:min-on-product-2}
\min_{z\in \mathcal{C}}\Pi_{x\in\mathcal{C},x\neq z}|x-z|= \left(\lceil\frac{n}{2}\rceil-1\right)!\left(\lfloor\frac{n}{2}\rfloor\right)! ; 
\end{equation}
(2):
\begin{equation}
\label{equ:argmax-on-product-2}
\arg\max_{z\in \mathcal{C}}\Pi_{x\in\mathcal{C},x\neq z}|x-z|=\{1,n\},
\end{equation} and
\begin{equation}
\label{equ:max-on-product-2}
\max_{z\in \mathcal{C}}\Pi_{x\in\mathcal{C},x\neq z}|x-z|=(n-1)!.
\end{equation}
\end{lem}

\begin{lem}\label{lem:product2}
Let $n\in\N^+$ and let $p,q\in\R$ be such that $p>q>0$. For the following three sets of evenly spaced points  $\mathcal{C}_1:=\{x_j=(j-1)p,j=1,\dots,n\}$, $\mathcal{C}_2:=\{y_j=(j-1)p-q,j=1,\dots,n\}$, and $\mathcal{C}_3:=\{y_j=(j-1)p-q,j=1,\dots,n+1\}$, we have\\
(1): $\arg\min_{z\in\mathcal{C}_2}\Pi_{j=1}^n|z-x_j|\subset\{y_{\lfloor\frac{n}{2}\rfloor+1}, y_{\lfloor\frac{n}{2}\rfloor+2}\}
$, and

\begin{equation}
\label{equ: min-on-product-1}
\min_{z\in\mathcal{C}_2}\Pi_{i=1}^n|z-x_i|=\Pi_{j=0}^{\lceil\frac{n}{2}\rceil-1}\left(\min(p-q,q)+pj\right)\cdot\Pi_{j=0}^{\lfloor\frac{n}{2}\rfloor-1}\left(\max(p-q,q)+pj\right);
\end{equation}

\noindent(2): $\arg\min_{z\in\mathcal{C}_3}\Pi_{j=1}^n|z-x_j|\subset\{y_{\lfloor\frac{n}{2}\rfloor+1}, y_{\lfloor\frac{n}{2}\rfloor+2}\}
$, and

\begin{equation}
\label{equ: min-on-product-2}
\min_{z\in\mathcal{C}_3}\Pi_{j=1}^n|z-x_j|=\Pi_{j=0}^{\lceil\frac{n}{2}\rceil-1}\left(\min(p-q,q)+pj\right)\cdot\Pi_{j=0}^{\lfloor\frac{n}{2}\rfloor-1}\left(\max(p-q,q)+pj\right);
\end{equation}

\noindent(3): $y_1 \in \arg\max_{z\in\mathcal{C}_2}\Pi_{j=1}^n|z-x_j|$, and 
\begin{equation}
\label{equ: max-on-product-1}
\max_{z\in\mathcal{C}_2}\Pi_{j=1}^n|z-x_j|=\Pi_{j=0}^{n-1}\left(q+pj\right);  
\end{equation}

\noindent(4): $\arg\max_{z\in\mathcal{C}_3}\Pi_{j=1}^n|z-x_j|\subset\{y_1,y_{n+1}\}$,
and
\begin{equation}
\label{equ: max-on-product-2}
\max_{z\in\mathcal{C}_3}\Pi_{i=j}^n|z-x_j|=\Pi_{j=0}^{n-1}\left(\max(p-q, q)+pj\right).  
\end{equation}
\end{lem}
 
\begin{proof}
Cases (3) and (4) are obvious. For cases (1) and (2), we only need to prove case (2). We verify firstly that for any integer $k$ with $0<k<\lfloor\frac{n}{2}\rfloor+1$,
\begin{equation}
\label{equ:arg-monotonicity}
\Pi^{n}_{j=1}|y_{k}-x_j|>\Pi^{n}_{j=1}|y_{k+1}-x_j|.
\end{equation}
This holds since
\begin{align*}
\Pi^{n}_{j=1}|y_{k}-x_j|=\Pi^{k-1}_{j=1}|y_k-x_j|\cdot\Pi^{n-1}_{j=k}|y_k-x_j|\cdot|y_k-x_n|
\end{align*} and
\begin{align*}
\Pi^{n}_{j=1}|y_{k+1}-x_j|=|y_{k+1}-x_1|\cdot\Pi^{k}_{j=2}|y_{k+1}-x_j|\cdot\Pi^{n}_{j=k+1}|y_{k+1}-x_j|.
\end{align*}
As $\Pi^{k-1}_{j=1}|y_k-x_j|=\Pi^{k}_{j=2}|y_{k+1}-x_j|$, $\Pi^{n-1}_{j=k}|y_k-x_j|=\Pi^{n}_{j=k+1}|y_{k+1}-x_j|$ and $|y_k-x_n|\geq|y_{k+1}-x_1|$ due to the geometrical structure of $\mathcal{C}_1, \mathcal{C}_3$, we have \eqref{equ:arg-monotonicity}. A similar argument gives that for any integer $k$ with $\lfloor\frac{n}{2}\rfloor+2<k<n+2$,
\begin{equation}
\label{equ:arg-monotonicity-2}
\Pi^{n}_{j=1}|y_{k}-x_j|>\Pi^{n}_{j=1}|y_{k-1}-x_j|.
\end{equation}
This proves that  $\arg\min_{z\in\mathcal{C}_3}\Pi_{j=1}^n|z-x_j|\subset\{y_{\lfloor\frac{n}{2}\rfloor+1},y_{\lfloor\frac{n}{2}\rfloor+2}\}$ and thus the minimum value can be checked directly. 
\end{proof} 

\begin{lem}
\label{lem:c-max}
For $c^{\max}_j$, $j=1,2,3,4,$ defined in \eqref{equ:nominator-max}, we have
\begin{equation*}
    \max_{j=1,2,3,4}c^{\max}_j\leq \tau^{2n-1} s^{2n-1}\left(2\lceil\frac{n}{2}\rceil\right)!\left(2\lceil\frac{n}{2}\rceil-1\right)!.
\end{equation*}
\end{lem}

\begin{proof}

We study $c^{\max}_j, j=1,2,3,4,$ term by term. The definition of $c^{\max}_1$ yields

\begin{align*}
c_1^{\max}&=\max_{x\in C_1}\Pi_{k=1,2,3,4}\Pi_{y\in C_k,y\neq x}|x-y| \\
&\leq \max_{x\in C_1}\Pi_{y\in C_1,y\neq x}|x-y|\cdot \max_{x\in C_1}\Pi_{y\in C_2}|x-y|\cdot \max_{x\in C_1}\Pi_{y\in C_3}|x-y|\cdot\max_{x\in C_1}\Pi_{y\in C_4}|x-y|.
\end{align*}
Note that, by Lemmas \ref{lem:produc-min-max-2} and \ref{lem:product2}, we have 
\begin{align*}
    \max_{x\in C_1}\Pi_{y\in C_1,y\neq x}|x-y|&\leq \tau^{\lceil\frac{n}{2}\rceil-1} \bigg(\Pi_{j=1}^{\lceil\frac{n}{2}\rceil-1}(2sj)\bigg) \qquad \left(\text{case (2) in Lemma \ref{lem:produc-min-max-2}}\right) \\ 
    &\leq \tau^{\lceil\frac{n}{2}\rceil-1} s^{\lceil\frac{n}{2}\rceil-1}\left(2\lceil\frac{n}{2}\rceil-1\right)!!.\\ 
    \\
    \max_{x\in C_1}\Pi_{y\in C_2}|x-y|&\leq \tau ^{\lfloor\frac{n}{2}\rfloor}
     \bigg( \Pi^{\lfloor\frac{n}{2}\rfloor-1}_{j=0}(2sj+s)\bigg) \qquad \left(\text{cases (3) or (4) in Lemma \ref{lem:product2}}\right)  \\ 
     &\leq \tau ^{\lfloor\frac{n}{2}\rfloor} s ^{\lfloor\frac{n}{2}\rfloor}\left(2\lfloor\frac{n}{2}\rfloor-2\right)!!.\\
     \\
     \max_{x\in C_1}\Pi_{y\in C_3}|x-y|&\leq \tau ^{\lceil\frac{n}{2}\rceil}\left( \Pi_{j=0}^{\lceil\frac{n}{2}\rceil-1}(2sj+1)\right) \qquad \left(\text{case (3) in Lemma \ref{lem:product2}}\right) \\
     &\leq \tau ^{\lceil\frac{n}{2}\rceil}s ^{\lceil\frac{n}{2}-1\rceil} \left(2\lceil\frac{n}{2}\rceil\right)!!.
     \\
     \\
    \max_{x\in C_1}\Pi_{y\in C_4}|x-y|&\leq \tau ^{\lfloor\frac{n}{2}\rfloor} \left( \Pi_{j=0}^{\lfloor\frac{n}{2}\rfloor-1}(2sj+2s-1)\right) \qquad \left(\text{cases (3) or (4) in Lemma \ref{lem:product2}}\right)\\
    &\leq \tau ^{\lfloor\frac{n}{2}\rfloor} s ^{\lfloor\frac{n}{2}\rfloor}   \left(2\lfloor\frac{n}{2}\rfloor-1\right)!!.
\end{align*}
Thus 
\begin{align*}
c_1^{\max} &\leq \tau^{2n-1} s^{2n-1} \left(2\lceil\frac{n}{2}\rceil-2\right)!!\left(2\lfloor\frac{n}{2}\rfloor-1\right)!!\left(2\lceil\frac{n}{2}\rceil-1\right)!!\left(2\lfloor\frac{n}{2}\rfloor\right)!! \\
&= \tau^{2n-1} s^{2n-2} \left(2\lceil\frac{n}{2}\rceil-1\right)! \left(2\lfloor\frac{n}{2}\rfloor\right)!.  
\end{align*}Regarding $c^{\max}_2$, we have

\begin{align*}
c_2^{\max}&=\max_{x\in C_2}\Pi_{k=1,2,3,4}\Pi_{y\in C_k,y\neq x}|x-y| \\
&\leq \max_{x\in C_2}\Pi_{y\in C_1}|x-y|\cdot \max_{x\in C_2}\Pi_{y\in C_2,y\neq x}|x-y|\cdot \max_{x\in C_2}\Pi_{y\in C_3}|x-y|\cdot\max_{x\in C_2}\Pi_{y\in C_4}|x-y|.
\end{align*}
Note that, by Lemmas \ref{lem:produc-min-max-2} and \ref{lem:product2}, we have 
\begin{align*}
    \max_{x\in C_2}\Pi_{y\in C_1}|x-y|&\leq \tau ^{\lceil\frac{n}{2}\rceil} \bigg(\Pi_{j=0}^{\lceil\frac{n}{2}\rceil-1}(2sj+s)\bigg) \qquad \left(\text{cases (3) in Lemma \ref{lem:product2}}\right) \\ 
    &\leq \tau ^{\lceil\frac{n}{2}\rceil} s ^{\lceil\frac{n}{2}\rceil} \left(2\lceil\frac{n}{2}\rceil-1\right)!!. \\
    \\
    \max_{x\in C_2}\Pi_{y\in C_2,y\neq x}|x-y|&\leq \tau ^{\lfloor\frac{n}{2}\rfloor-1}
     \bigg( \Pi^{\lfloor\frac{n}{2}\rfloor-1}_{j=1}(2sj)\bigg) \qquad \left(\text{case (2) in Lemma \ref{lem:produc-min-max-2}}\right)  \\ 
     &\leq \tau ^{\lfloor\frac{n}{2}\rfloor-1} s ^{\lfloor\frac{n}{2}\rfloor-1}\left(2\lfloor\frac{n}{2}\rfloor-2\right)!!.\\
     \\
    \max_{x\in C_2}\Pi_{y\in C_3}|x-y|&\leq \tau ^{\lceil\frac{n}{2}\rceil}\left( \Pi_{j=0}^{\lceil\frac{n}{2}\rceil-1}(2sj+s+1)\right) \qquad \left(\text{case (3) in Lemma \ref{lem:product2}}\right) \\
    &\leq \tau ^{\lceil\frac{n}{2}\rceil} s ^{\lceil\frac{n}{2}\rceil} \left(2\lceil\frac{n}{2}\rceil\right)!!.\\
    \\
    \max_{x\in C_2}\Pi_{y\in C_4}|x-y|&\leq \tau ^{\lfloor\frac{n}{2}\rfloor} \left( \Pi_{j=0}^{\lfloor\frac{n}{2}\rfloor-1}(2sj+s-1)\right) \qquad \left(\text{case (3) in Lemma \ref{lem:product2}}\right) \\
    &\leq \tau ^{\lfloor\frac{n}{2}\rfloor} s ^{\lfloor\frac{n}{2}\rfloor}\left(2\lfloor\frac{n}{2}\rfloor-1\right)!!.
\end{align*}
Thus 
\begin{align*}
    C_2^{\max} &\leq \tau^{2n-1} s^{2n-1} \left(2\lceil\frac{n}{2}\rceil-1\right)!!\left(2\lfloor\frac{n}{2}\rfloor-2\right)!!\left(2\lceil\frac{n}{2}\rceil\right)!!\left(2\lfloor\frac{n}{2}\rfloor-1\right)!! \\
    &= \tau^{2n-1} s^{2n-1}\left(2\lceil\frac{n}{2}\rceil\right)!\left(2\lfloor\frac{n}{2}\rfloor-1\right)!.
\end{align*}As for $c^{\max}_3$, we have

\begin{align*}
c_3^{\max}&=\max_{x\in C_3}\Pi_{k=1,2,3,4}\Pi_{y\in C_k,y\neq x}|x-y| \\
&\leq \max_{x\in C_3}\Pi_{y\in C_1}|x-y|\cdot \max_{x\in C_3}\Pi_{y\in C_2}|x-y|\cdot \max_{x\in C_3,y\neq x}\Pi_{y\in C_3}|x-y|\cdot\max_{x\in C_3}\Pi_{y\in C_4}|x-y| . \end{align*}
Note that, by Lemmas \ref{lem:produc-min-max-2} and \ref{lem:product2}, we obtain
\begin{align*}
    \max_{x\in C_3}\Pi_{y\in C_1}|x-y|&\leq \tau ^{\lceil\frac{n}{2}\rceil} \bigg(\Pi_{j=0}^{\lceil\frac{n}{2}\rceil-1}(2sj+1)\bigg), \qquad \left(\text{case (3) in Lemma \ref{lem:product2}}\right) \\ 
    &\leq \tau ^{\lceil\frac{n}{2}\rceil} s ^{\lceil\frac{n}{2}\rceil-1} \left(2\lceil\frac{n}{2}\rceil-1\right)!!.\\
    \\
    \max_{x\in C_3}\Pi_{y\in C_2}|x-y|&\leq \tau ^{\lfloor\frac{n}{2}\rfloor}
     \bigg( \Pi^{\lfloor\frac{n}{2}\rfloor-1}_{j=0}(2sj+s+1)\bigg), \qquad \left(\text{case (3) or (4) in Lemma \ref{lem:product2}}\right)  \\
     \\
    \max_{x\in C_3}\Pi_{y\in C_3,y\neq x}|x-y|&\leq \tau ^{\lceil\frac{n}{2}\rceil-1}\left( \Pi_{j=1}^{\lceil\frac{n}{2}\rceil-1}(2sj)\right), \qquad \left(\text{case (2) in Lemma \ref{lem:produc-min-max-2}}\right) \\
    &\leq \tau ^{\lceil\frac{n}{2}\rceil-1} s ^{\lceil\frac{n}{2}\rceil-1}\left(2\lceil\frac{n}{2}\rceil-2\right)!!.\\
    \\
    \max_{x\in C_3}\Pi_{y\in C_4}|x-y|&\leq \tau ^{\lfloor\frac{n}{2}\rfloor} \left( \Pi_{j=0}^{\lfloor\frac{n}{2}\rfloor-1}(2sj+2s-2)\right). \qquad \left(\text{case (3) or (4) in Lemma \ref{lem:product2}}\right) 
\end{align*}
Thus 
\begin{align*}
    C_3^{\max} &\leq \tau^{2n-1} s^{2\lceil\frac{n}{2}\rceil-2} 
     \left(2\lceil\frac{n}{2}\rceil-1\right)!\cdot\left( \Pi_{j=0}^{\lfloor\frac{n}{2}\rfloor-1}(2sj+s+1)(2sj+2s-2)\right)\\
    &\leq \tau^{2n-1} s^{2\lceil\frac{n}{2}\rceil-2} 
     \left(2\lceil\frac{n}{2}\rceil-1\right)!\cdot\left( \Pi_{j=0}^{\lfloor\frac{n}{2}\rfloor-1}(2sj+s)(2sj+2s)\right)\\
    &= \tau^{2n-1} s^{2n-3}\left(2\lceil\frac{n}{2}\rceil-1\right)!\left(2\lfloor\frac{n}{2}\rfloor\right)!.
\end{align*}
Finally,  we study $c^{\max}_4$. We have
\begin{align*}
    \max_{x\in C_4}\Pi_{y\in C_1}|x-y|&\leq \tau ^{\lceil\frac{n}{2}\rceil} (\Pi_{j=0}^{\lceil\frac{n}{2}\rceil-1}(2sj+1)), \qquad \left(\text{case (3) in Lemma \ref{lem:product2}}\right) \\ 
    &\leq \tau^{\lceil\frac{n}{2}\rceil} s ^{\lceil\frac{n}{2}\rceil-1} \bigg( 2\lceil\frac{n}{2}\rceil-1\bigg)!!\\
    \\
    \max_{x\in C_4}\Pi_{y\in C_2}|x-y|&\leq \tau ^{\lfloor\frac{n}{2}\rfloor}
     ( \Pi^{\lfloor\frac{n}{2}\rfloor-1}_{j=0}(2sj+s-1)), \qquad \left(\text{case (3) in Lemma \ref{lem:product2}}\right)  \\ 
     &\leq \tau ^{\lfloor\frac{n}{2}\rfloor} s ^{\lfloor\frac{n}{2}\rfloor}\bigg(2\lfloor\frac{n}{2}\rfloor\bigg)!!\\
     \\
     \max_{x\in C_4}\Pi_{y\in C_3}|x-y|&\leq \tau ^{\lceil\frac{n}{2}\rceil}( \Pi_{j=0}^{\lceil\frac{n}{2}\rceil-1}(2sj+2)), \qquad \left(\text{case (3) in Lemma \ref{lem:product2}}\right)\\ &\leq \tau^{\lceil\frac{n}{2}\rceil} s ^{\lceil\frac{n}{2}\rceil} \bigg( 2\lceil\frac{n}{2}\rceil-1\bigg)!!.\\
    \\
    \max_{x\in C_4}\Pi_{y\in C_4,y\neq x}|x-y|&\leq \tau ^{\lfloor\frac{n}{2}\rfloor-1} ( \Pi_{j=1}^{\lfloor\frac{n}{2}\rfloor-1}(2sj)). \qquad \left(\text{case (2) in Lemma \ref{lem:produc-min-max-2}}\right) \\
    &\leq \tau ^{\lfloor\frac{n}{2}\rfloor-1} s^{\lfloor\frac{n}{2}\rfloor-1}  \bigg( 2\lfloor\frac{n}{2}\rfloor-2\bigg)!!.
\end{align*}
Thus 

\begin{align*}
c_4^{\max}&\leq \tau^{2n-1}s^{2n-1}\bigg( 2\lceil\frac{n}{2}\rceil-1\bigg)!!\cdot\bigg(2\lfloor\frac{n}{2}\rfloor\bigg)!!\cdot\bigg( 2\lceil\frac{n}{2}\rceil-1\bigg)!!\cdot\bigg( 2\lfloor\frac{n}{2}\rfloor-2\bigg)!!\\
&\leq \tau^{2n-1}s^{2n-1}\bigg(2\lceil\frac{n}{2}\rceil\bigg)!\bigg( 2\lceil\frac{n}{2}\rceil-1\bigg)!.
\end{align*}

By concluding above discussions, we have 
\[
\max_{j=1,2,3,4}c_j^{\max} \leq \tau^{2n-1} s^{2n-1}\left(2\lceil\frac{n}{2}\rceil\right)!\left(2\lceil\frac{n}{2}\rceil-1\right)!. 
\]

\end{proof}

\begin{lem}
\label{lem:c-min}
For $c^{\min}_j$, $j=1,2,3,4,$ defined in \eqref{equ:denominator-min}, we have
\begin{equation*}
    \min_{j=1,2,3,4}c^{\min}_j\geq \tau^{2n-1}s^{2n-3}\cdot \bigg((2\lfloor\frac{\lfloor\frac{n}{2}\rfloor}{2}\rfloor-1)!\bigg)^4.
\end{equation*}
\end{lem}

\begin{proof}

We evaluate $c^{\min}_k, k=1,2,3,4$. For $c^{\min}_1$, 
\begin{align*}
\qquad c^{\min}_1&=\min_{x\in C_1}\Pi_{k=1,2,3,4}\Pi_{y\in C_k,y\neq x}|x-y| \\
&\geq \min_{x\in C_1}\Pi_{y\in C_1,y\neq x}|x-y|\cdot \min_{x\in C_1}\Pi_{y\in C_2}|x-y|\cdot \min_{x\in C_1}\Pi_{y\in C_3}|x-y|\cdot\min_{x\in C_1}\Pi_{y\in C_4}|x-y|.
\end{align*}
By using Lemmas \ref{lem:produc-min-max-2} and \ref{lem:product2}, we have 
\begin{align*}
    \min_{x\in C_1}\Pi_{y\in C_1}|x-y|&=\tau^{\lceil\frac{n}{2}\rceil-1} \Pi^{\lceil\frac{\lceil\frac{n}{2}\rceil}{2}\rceil-1}_{j=1}(j\cdot 2s)\cdot\Pi^{\lfloor\frac{\lceil\frac{n}{2}\rceil}{2}\rfloor}_{j=1}(j\cdot 2s), \quad \left(\text{case (1) in Lemma \ref{lem:produc-min-max-2}}\right) \\ 
    & = \tau^{\lceil\frac{n}{2}\rceil-1} s^{\lceil\frac{n}{2}\rceil-1} \bigg( (2\lceil\frac{\lceil\frac{n}{2}\rceil}{2}\rceil-2)!!\cdot(2\lfloor\frac{\lceil\frac{n}{2}\rceil}{2}\rfloor)!! \bigg) \\
    &\geq \tau^{\lceil\frac{n}{2}\rceil-1} s^{\lceil\frac{n}{2}\rceil-1} \bigg( (2\lfloor\frac{\lfloor\frac{n}{2}\rfloor}{2}\rfloor-2)!!\cdot(2\lfloor\frac{\lfloor\frac{n}{2}\rfloor}{2}\rfloor)!! \bigg). \\
\\
 \min_{x\in C_1}\Pi_{y\in C_2,y\neq x}|x-y|&= \tau^{\lfloor\frac{n}{2}\rfloor}\Pi_{j=0}^{\lceil\frac{\lfloor\frac{n}{2}\rfloor}{2}\rceil-1}\left(s+2sj\right)\cdot\Pi_{j=0}^{\lfloor\frac{\lfloor\frac{n}{2}\rfloor}{2}\rfloor-1}\left(s+2sj\right), \quad \left(\text{case (1) or (2) in Lemma \ref{lem:product2}}\right)  \\ 
 &= \tau^{\lfloor\frac{n}{2}\rfloor} s^{\lfloor\frac{n}{2}\rfloor} \bigg((2\lceil\frac{\lfloor\frac{n}{2}\rfloor}{2}\rceil-1)!!\cdot(2\lfloor\frac{\lfloor\frac{n}{2}\rfloor}{2}\rfloor-1)!!\bigg)\\
 &\geq \tau^{\lfloor\frac{n}{2}\rfloor} s^{\lfloor\frac{n}{2}\rfloor} \bigg((2\lfloor\frac{\lfloor\frac{n}{2}\rfloor}{2}\rfloor-1)!!\cdot(2\lfloor\frac{\lfloor\frac{n}{2}\rfloor}{2}\rfloor-1)!!\bigg).\\
\end{align*}
 Furthermore, 
\begin{align*} 
    \min_{x\in C_1}\Pi_{y\in C_3}|x-y|&= \tau ^{\lceil\frac{n}{2}\rceil}\Pi_{j=0}^{\lceil\frac{\lceil\frac{n}{2}\rceil}{2}\rceil-1}\left(1+2sj\right)\cdot\Pi_{k=0}^{\lfloor\frac{\lceil\frac{n}{2}\rceil}{2}\rfloor-1}\left(2s-1+2sk\right) \quad \left(\text{case (1) in Lemma \ref{lem:product2}}\right) \\
    &\geq\tau ^{\lceil\frac{n}{2}\rceil}\Pi_{j=1}^{\lceil\frac{\lceil\frac{n}{2}\rceil}{2}\rceil-1}\left(2sj\right)\cdot\Pi_{k=0}^{\lfloor\frac{\lceil\frac{n}{2}\rceil}{2}\rfloor-1}\left(s+2sk\right) \\
    &= \tau^{\lceil\frac{n}{2}\rceil} s^{\lceil\frac{n}{2}\rceil-1} \bigg((2\lceil\frac{\lceil\frac{n}{2}\rceil}{2}\rceil-2)!!\cdot(2\lfloor\frac{\lceil\frac{n}{2}\rceil}{2}\rfloor-1)!!\bigg)\\
    &\geq \tau^{\lceil\frac{n}{2}\rceil} s^{\lceil\frac{n}{2}\rceil-1} \bigg((2\lfloor\frac{\lceil\frac{n}{2}\rceil}{2}\rfloor-2)!!\cdot(2\lfloor\frac{\lceil\frac{n}{2}\rceil}{2}\rfloor-1)!!\bigg)\\
    &\geq \tau^{\lfloor\frac{n}{2}\rfloor} s^{\lfloor\frac{n}{2}\rfloor-1} \cdot\left(2\lfloor\frac{\lfloor\frac{n}{2}\rfloor}{2}\rfloor-1\right)!.\\
    \\
    \min_{x\in C_1}\Pi_{y\in C_4}|x-y|&= \tau ^{\lfloor\frac{n}{2}\rfloor} \Pi_{j=0}^{\lceil\frac{\lfloor\frac{n}{2}\rfloor}{2}\rceil-1}\left(1+2sj\right)\cdot\Pi_{k=0}^{\lfloor\frac{\lfloor\frac{n}{2}\rfloor}{2}\rfloor-1}\left(2s-1+2sk\right) \quad \left(\text{case (1) or (2) in Lemma \ref{lem:product2}}\right) \\
    &\geq \tau ^{\lfloor\frac{n}{2}\rfloor} \Pi_{j=1}^{\lceil\frac{\lfloor\frac{n}{2}\rfloor}{2}\rceil-1}\left(2sj\right)\cdot\Pi_{k=0}^{\lfloor\frac{\lfloor\frac{n}{2}\rfloor}{2}\rfloor-1}\left(s+2sk\right)\\
    &=\tau^{\lfloor\frac{n}{2}\rfloor}s ^{\lfloor\frac{n}{2}\rfloor-1}\bigg( (2\lceil\frac{\lfloor\frac{n}{2}\rfloor}{2}\rceil-2)!!\cdot(2\lfloor\frac{\lfloor\frac{n}{2}\rfloor}{2}\rfloor-1)!!\bigg)\\
    &\geq\tau^{\lfloor\frac{n}{2}\rfloor}s ^{\lfloor\frac{n}{2}\rfloor-1}\left(2\lfloor\frac{\lfloor\frac{n}{2}\rfloor}{2}\rfloor-1\right)!.
\end{align*}
By combining estimates in the above four cases, it follows that 
\[
c_1^{\min}\geq \tau^{2n-1}s^{2n-3}\left(2\lfloor\frac{\lfloor\frac{n}{2}\rfloor}{2}\rfloor\right)!\left(\left(2\lfloor\frac{\lfloor\frac{n}{2}\rfloor}{2}\rfloor-1\right)!\right)^{3}.
\]
Regarding $c^{\min}_2$, we have
\begin{align*}
\qquad c^{\min}_2&=\min_{x\in C_2}\Pi_{k=1,2,3,4}\Pi_{y\in C_k,y\neq x}|x-y| \\
&\geq \min_{x\in C_2}\Pi_{y\in C_1}|x-y|\cdot \min_{x\in C_2}\Pi_{y\in C_2,y\neq x}|x-y|\cdot \min_{x\in C_2}\Pi_{y\in C_3}|x-y|\cdot\min_{x\in C_2}\Pi_{y\in C_4}|x-y|.
\end{align*}
Note that by Lemmas \ref{lem:produc-min-max-2} and \ref{lem:product2} we have 
\begin{align*}
    \min_{x\in C_2}\Pi_{y\in C_1}|x-y|&= \tau^{\lceil\frac{n}{2}\rceil}\Pi_{j=0}^{\lceil\frac{\lceil\frac{n}{2}\rceil}{2}\rceil-1}\left(s+2sj\right)\cdot\Pi_{j=0}^{\lfloor\frac{\lceil\frac{n}{2}\rceil}{2}\rfloor-1}\left(s+2sj\right) \qquad \left(\text{case (1) or (2) in Lemma \ref{lem:product2}}\right)  \\
    &=\tau^{\lceil\frac{n}{2}\rceil}s^{\lceil\frac{n}{2}\rceil}\bigg( (2\lceil\frac{\lceil\frac{n}{2}\rceil}{2}\rceil-1)!!\cdot(2\lfloor\frac{\lceil\frac{n}{2}\rceil}{2}\rfloor-1)!! \bigg)\\
    &\geq \tau^{\lceil\frac{n}{2}\rceil}s^{\lceil\frac{n}{2}\rceil}\bigg( (2\lceil\frac{\lfloor\frac{n}{2}\rfloor}{2}\rceil-1)!!\cdot(2\lfloor\frac{\lfloor\frac{n}{2}\rfloor}{2}\rfloor-1)!! \bigg).\\
    \min_{x\in C_2}\Pi_{y\in C_2,y\neq x}|x-y|&=\tau^{\lfloor\frac{n}{2}\rfloor-1} \Pi^{\lceil\frac{\lfloor\frac{n}{2}\rfloor}{2}\rceil-1}_{j=1}(j\cdot 2s)\cdot\Pi^{\lfloor\frac{\lfloor\frac{n}{2}\rfloor}{2}\rfloor}_{j=1}(j\cdot 2s) \qquad \left(\text{case (1) in Lemma \ref{lem:produc-min-max-2}}\right) \\ 
    &=\tau^{\lfloor\frac{n}{2}\rfloor-1}s^{\lfloor\frac{n}{2}\rfloor-1}\bigg((2\lceil\frac{\lfloor\frac{n}{2}\rfloor}{2}\rceil-2)!!\cdot(2\lfloor\frac{\lfloor\frac{n}{2}\rfloor}{2}\rfloor)!!\bigg)\\
    &\geq \tau^{\lfloor\frac{n}{2}\rfloor-1}s^{\lfloor\frac{n}{2}\rfloor-1}\bigg((2\lceil\frac{\lfloor\frac{n}{2}\rfloor}{2}\rceil-2)!!\cdot(2\lfloor\frac{\lfloor\frac{n}{2}\rfloor}{2}\rfloor)!!\bigg).
    \end{align*}
    \begin{align*}
    \min_{x\in C_2}\Pi_{y\in C_3}|x-y|&= \tau ^{\lceil\frac{n}{2}\rceil}\Pi_{j=0}^{\lceil\frac{\lceil\frac{n}{2}\rceil}{2}\rceil-1}\left(s-1+2sj\right)\cdot\Pi_{k=0}^{\lfloor\frac{\lceil\frac{n}{2}\rceil}{2}\rfloor-1}\left(s+1+2sk\right) \qquad \left(\text{case (1) or (2) in Lemma \ref{lem:product2}}\right) \\
    &\geq \tau ^{\lceil\frac{n}{2}\rceil}\Pi_{j=1}^{\lceil\frac{\lceil\frac{n}{2}\rceil}{2}\rceil-1}\left(2sj\right)\cdot\Pi_{k=0}^{\lfloor\frac{\lceil\frac{n}{2}\rceil}{2}\rfloor-1}\left(s+2sk\right)\\
    &=\tau ^{\lceil\frac{n}{2}\rceil}s^{\lceil\frac{n}{2}\rceil-1}\bigg((2\lceil\frac{\lceil\frac{n}{2}\rceil}{2}\rceil-2)!!\cdot(2\lfloor\frac{\lceil\frac{n}{2}\rceil}{2}\rfloor-1)!!\bigg)\\
    &\geq\tau ^{\lceil\frac{n}{2}\rceil}s^{\lceil\frac{n}{2}\rceil-1}\bigg((2\lfloor\frac{\lceil\frac{n}{2}\rceil}{2}\rfloor-2)!!\cdot(2\lfloor\frac{\lceil\frac{n}{2}\rceil}{2}\rfloor-1)!!\bigg)\\
    &=\tau ^{\lceil\frac{n}{2}\rceil}s^{\lceil\frac{n}{2}\rceil-1}(2\lfloor\frac{\lceil\frac{n}{2}\rceil}{2}\rfloor-1)!.\\
    \min_{x\in C_2}\Pi_{y\in C_4}|x-y|&= \tau ^{\lfloor\frac{n}{2}\rfloor} \Pi_{j=0}^{\lceil\frac{\lfloor\frac{n}{2}\rfloor}{2}\rceil-1}\left(s-1+2sj\right)\cdot\Pi_{k=0}^{\lfloor\frac{\lfloor\frac{n}{2}\rfloor}{2}\rfloor-1}\left(s+1+2sk\right) \qquad \left(\text{case (1) in Lemma \ref{lem:product2}}\right) \\
    &\geq\tau ^{\lfloor\frac{n}{2}\rfloor} \Pi_{j=1}^{\lceil\frac{\lfloor\frac{n}{2}\rfloor}{2}\rceil-1}\left(2sj\right)\cdot\Pi_{k=0}^{\lfloor\frac{\lfloor\frac{n}{2}\rfloor}{2}\rfloor-1}\left(s+2sk\right) \\
    &=\tau ^{\lfloor\frac{n}{2}\rfloor}s^{\lfloor\frac{n}{2}\rfloor-1}\bigg((2\lceil\frac{\lfloor\frac{n}{2}\rfloor}{2}\rceil-2)!!\cdot(2\lfloor\frac{\lfloor\frac{n}{2}\rfloor}{2}\rfloor-1)!!\bigg)\\
    &\geq \tau ^{\lfloor\frac{n}{2}\rfloor}s^{\lfloor\frac{n}{2}\rfloor-1}\bigg((2\lfloor\frac{\lfloor\frac{n}{2}\rfloor}{2}\rfloor-2)!!\cdot(2\lfloor\frac{\lfloor\frac{n}{2}\rfloor}{2}\rfloor-1)!!\bigg)\\
    &=\tau ^{\lfloor\frac{n}{2}\rfloor}s^{\lfloor\frac{n}{2}\rfloor-1}(2\lfloor\frac{\lfloor\frac{n}{2}\rfloor}{2}\rfloor-1)!.
\end{align*}
Thus,
\begin{align*}
& c^{\min}_2
\geq \tau^{2n-1}s^{2n-3}\cdot\bigg( (2\lceil\frac{\lfloor\frac{n}{2}\rfloor}{2}\rceil-1)!!\cdot(2\lfloor\frac{\lfloor\frac{n}{2}\rfloor}{2}\rfloor-1)!! \bigg) \cdot \bigg((2\lceil\frac{\lfloor\frac{n}{2}\rfloor}{2}\rceil-2)!!\cdot(2\lfloor\frac{\lfloor\frac{n}{2}\rfloor}{2}\rfloor)!!\bigg)\cdot (2\lfloor\frac{\lceil\frac{n}{2}\rceil}{2}\rfloor-1)!\cdot(2\lfloor\frac{\lfloor\frac{n}{2}\rfloor}{2}\rfloor-1)!\\
&=\tau^{2n-1}s^{2n-3}\cdot (2\lceil\frac{\lfloor\frac{n}{2}\rfloor}{2}\rceil-1)!\cdot(2\lfloor\frac{\lfloor\frac{n}{2}\rfloor}{2}\rfloor)!\cdot (2\lfloor\frac{\lceil\frac{n}{2}\rceil}{2}\rfloor-1)!\cdot(2\lfloor\frac{\lfloor\frac{n}{2}\rfloor}{2}\rfloor-1)!\\
&\geq \tau^{2n-1}s^{2n-3}\cdot(2\lfloor\frac{\lfloor\frac{n}{2}\rfloor}{2}\rfloor)!\cdot \bigg((2\lfloor\frac{\lfloor\frac{n}{2}\rfloor}{2}\rfloor-1)!\bigg)^3.\\
\end{align*}
We then turn to estimate $c_3^{\min}$. We have that  
\begin{align*}
\qquad c^{\min}_3&=\min_{x\in C_3}\Pi_{k=1,2,3,4}\Pi_{y\in C_k,y\neq x}|x-y| \\
&\geq \min_{x\in C_3}\Pi_{y\in C_1}|x-y|\cdot \min_{x\in C_3}\Pi_{y\in C_2}|x-y|\cdot \min_{x\in C_3}\Pi_{y\in C_3,y\neq x}|x-y|\cdot\min_{x\in C_3}\Pi_{y\in C_4}|x-y|.
\end{align*}
By Lemmas \ref{lem:produc-min-max-2} and \ref{lem:product2}, it follows that we  
\begin{align*}
    \min_{x\in C_3}\Pi_{y\in C_1}|x-y|&= \tau^{\lceil\frac{n}{2}\rceil}\Pi_{j=0}^{\lceil\frac{\lceil\frac{n}{2}\rceil}{2}\rceil-1}\left(1+2sj\right)\cdot\Pi_{j=0}^{\lfloor\frac{\lceil\frac{n}{2}\rceil}{2}\rfloor-1}\left(2s-1+2sj\right) \qquad \left(\text{case (1) or (2) in Lemma \ref{lem:produc-min-max-2}}\right)  \\
    &\geq \tau^{\lceil\frac{n}{2}\rceil}\Pi_{j=1}^{\lceil\frac{\lceil\frac{n}{2}\rceil}{2}\rceil-1}\left(2sj\right)\cdot\Pi_{j=0}^{\lfloor\frac{\lceil\frac{n}{2}\rceil}{2}\rfloor-1}\left(s+2sj\right)\\
    &= \tau^{\lceil\frac{n}{2}\rceil}s^{\lceil\frac{n}{2}\rceil-1}\bigg( (2\lceil\frac{\lceil\frac{n}{2}\rceil}{2}\rceil-2)!!\cdot(2\lfloor\frac{\lceil\frac{n}{2}\rceil}{2}\rfloor-1)!! \bigg)\\
      &\geq \tau^{\lceil\frac{n}{2}\rceil}s^{\lceil\frac{n}{2}\rceil-1}\bigg( (2\lfloor\frac{\lfloor\frac{n}{2}\rfloor}{2}\rfloor-2)!!\cdot(2\lfloor\frac{\lfloor\frac{n}{2}\rfloor}{2}\rfloor-1)!! \bigg)\\
      &=\tau^{\lceil\frac{n}{2}\rceil}s^{\lceil\frac{n}{2}\rceil-1}(2\lfloor\frac{\lfloor\frac{n}{2}\rfloor}{2}\rfloor-1)!.\\
     \min_{x\in C_3}\Pi_{y\in C_2}|x-y|&= \tau ^{\lfloor\frac{n}{2}\rfloor}\Pi_{j=0}^{\lceil\frac{\lfloor\frac{n}{2}\rfloor}{2}\rceil-1}\left(s-1+2sj\right)\cdot\Pi_{k=0}^{\lfloor\frac{\lfloor\frac{n}{2}\rfloor}{2}\rfloor-1}\left(s+1+2sk\right) \qquad \left(\text{case (1) in Lemma \ref{lem:produc-min-max-2}}\right) \\
     &=\tau ^{\lfloor\frac{n}{2}\rfloor}\bigg(s^{\lfloor\frac{n}{2}\rfloor}\Pi_{j=0}^{\lceil\frac{\lfloor\frac{n}{2}\rfloor}{2}\rceil-1}\left(\frac{s-1}{s}+2j\right)\cdot\Pi_{j=0}^{\lfloor\frac{\lfloor\frac{n}{2}\rfloor}{2}\rfloor-1}\left(\frac{s+1}{s}+2j\right)\bigg)\\
     &\geq\tau ^{\lfloor\frac{n}{2}\rfloor}\bigg(s^{\lfloor\frac{n}{2}\rfloor}\cdot\frac{1}{2}\cdot\Pi_{j=1}^{\lceil\frac{\lfloor\frac{n}{2}\rfloor}{2}\rceil-1}\left(2j\right)\cdot\Pi_{j=0}^{\lfloor\frac{\lfloor\frac{n}{2}\rfloor}{2}\rfloor-1}\left(1+2j\right)\bigg)\\
     &=\frac{1}{2}\tau ^{\lfloor\frac{n}{2}\rfloor}s^{\lfloor\frac{n}{2}\rfloor}\bigg(\Pi_{j=1}^{\lceil\frac{\lfloor\frac{n}{2}\rfloor}{2}\rceil-1}\left(2j\right)\cdot\Pi_{j=0}^{\lfloor\frac{\lfloor\frac{n}{2}\rfloor}{2}\rfloor-1}\left(1+2j\right)\bigg)\\
     &=\frac{1}{2}\tau ^{\lfloor\frac{n}{2}\rfloor}s^{\lfloor\frac{n}{2}\rfloor}\bigg((2\lceil\frac{\lfloor\frac{n}{2}\rfloor}{2}\rceil-2)!!\cdot(2\lfloor\frac{\lfloor\frac{n}{2}\rfloor}{2}\rfloor-1)!!\bigg)\\
     &\geq\frac{1}{2}\tau ^{\lfloor\frac{n}{2}\rfloor}s^{\lfloor\frac{n}{2}\rfloor}\bigg((2\lfloor\frac{\lfloor\frac{n}{2}\rfloor}{2}\rfloor-2)!!\cdot(2\lfloor\frac{\lfloor\frac{n}{2}\rfloor}{2}\rfloor-1)!!\bigg)\\
     &\geq\frac{1}{2}\tau ^{\lfloor\frac{n}{2}\rfloor}s^{\lfloor\frac{n}{2}\rfloor}(2\lfloor\frac{\lfloor\frac{n}{2}\rfloor}{2}\rfloor-1)!.
     \end{align*}
     \begin{align*}
    \min_{x\in C_3}\Pi_{y\in C_3,y\neq x}|x-y|&=\tau^{\lceil\frac{n}{2}\rceil-1} \Pi^{\lceil\frac{\lceil\frac{n}{2}\rceil}{2}\rceil-1}_{j=1}(j\cdot 2s)\cdot\Pi^{\lfloor\frac{\lceil\frac{n}{2}\rceil}{2}\rfloor}_{j=1}(j\cdot 2s) \qquad \left(\text{case (1) in Lemma \ref{lem:product2}}\right) \\
    &=\tau^{\lceil\frac{n}{2}\rceil-1}s^{\lceil\frac{n}{2}\rceil-1}\bigg((2\lceil\frac{\lceil\frac{n}{2}\rceil}{2}\rceil-2)!!\cdot(2\lfloor\frac{\lceil\frac{n}{2}\rceil}{2}\rfloor)!!\bigg)\\
    &\geq\tau^{\lceil\frac{n}{2}\rceil-1}s^{\lceil\frac{n}{2}\rceil-1}\bigg((2\lfloor\frac{\lceil\frac{n}{2}\rceil}{2}\rfloor-2)!!\cdot(2\lfloor\frac{\lceil\frac{n}{2}\rceil}{2}\rfloor-1)!!\bigg)\\
    &=\tau^{\lceil\frac{n}{2}\rceil-1}s^{\lceil\frac{n}{2}\rceil-1}(2\lfloor\frac{\lceil\frac{n}{2}\rceil}{2}\rfloor-1)!.
    \end{align*}
    \begin{align*}
    \min_{x\in C_3}\Pi_{y\in C_4}|x-y|&= \tau ^{\lfloor\frac{n}{2}\rfloor} \Pi_{j=0}^{\lceil\frac{\lfloor\frac{n}{2}\rfloor}{2}\rceil-1}\left(2+2sj\right)\cdot\Pi_{k=0}^{\lfloor\frac{\lfloor\frac{n}{2}\rfloor}{2}\rfloor-1}\left(2s-2+2sk\right) \qquad \left(\text{case (1) or (2) in Lemma \ref{lem:product2}}\right) \\
    &\geq 2\tau ^{\lfloor\frac{n}{2}\rfloor} \Pi_{j=1}^{\lceil\frac{\lfloor\frac{n}{2}\rfloor}{2}\rceil-1}\left(2sj\right)\cdot\Pi_{k=0}^{\lfloor\frac{\lfloor\frac{n}{2}\rfloor}{2}\rfloor-1}\left(s+2sk\right)\\
    &=2\tau ^{\lfloor\frac{n}{2}\rfloor}s^{\lfloor\frac{n}{2}\rfloor-1}\bigg( (2\lceil\frac{\lfloor\frac{n}{2}\rfloor}{2}\rceil-2)!!\cdot(2\lfloor\frac{\lfloor\frac{n}{2}\rfloor}{2}\rfloor-1)!!\bigg)\\
    &\geq2\tau ^{\lfloor\frac{n}{2}\rfloor}s^{\lfloor\frac{n}{2}\rfloor-1}\bigg( (2\lfloor\frac{\lfloor\frac{n}{2}\rfloor}{2}\rfloor-2)!!\cdot(2\lfloor\frac{\lfloor\frac{n}{2}\rfloor}{2}\rfloor-1)!!\bigg)\\
    &=2\tau ^{\lfloor\frac{n}{2}\rfloor}s^{\lfloor\frac{n}{2}\rfloor-1} (2\lfloor\frac{\lfloor\frac{n}{2}\rfloor}{2}\rfloor-1)!.
\end{align*}
Thus,
\begin{align*}
c^{\min}_3&\geq \cdot\tau^{2n-1}s^{2n-3}\cdot (2\lfloor\frac{\lfloor\frac{n}{2}\rfloor}{2}\rfloor-1)!\cdot(2\lfloor\frac{\lfloor\frac{n}{2}\rfloor}{2}\rfloor-1)!\cdot (2\lfloor\frac{\lceil\frac{n}{2}\rceil}{2}\rfloor-1)!\cdot(2\lfloor\frac{\lfloor\frac{n}{2}\rfloor}{2}\rfloor-1)!\\
&\geq \tau^{2n-1}s^{2n-3}\cdot \bigg((2\lfloor\frac{\lfloor\frac{n}{2}\rfloor}{2}\rfloor-1)!\bigg)^4
\end{align*}Finally, regarding $c^{\min}_4$, we write
\begin{align*}
\qquad c^{\min}_4&=\min_{x\in C_4}\Pi_{k=1,2,3,4}\Pi_{y\in C_k,y\neq x}|x-y| \\
&\geq \min_{x\in C_4}\Pi_{y\in C_1}|x-y|\cdot \min_{x\in C_4}\Pi_{y\in C_2}|x-y|\cdot \min_{x\in C_4}\Pi_{y\in C_3}|x-y|\cdot\min_{x\in C_4}\Pi_{y\in C_4,y\neq x}|x-y|.
\end{align*}
Note that, by Lemmas \ref{lem:produc-min-max-2} and \ref{lem:product2}, we obtain 
\begin{align*}
    \min_{x\in C_4}\Pi_{y\in C_1}|x-y|&= \tau^{\lceil\frac{n}{2}\rceil}\Pi_{j=0}^{\lceil\frac{\lceil\frac{n}{2}\rceil}{2}\rceil-1}\left(1+2sj\right)\cdot\Pi_{j=0}^{\lfloor\frac{\lceil\frac{n}{2}\rceil}{2}\rfloor-1}\left(2s-1+2sj\right) \qquad \left(\text{case (1) in Lemma \ref{lem:produc-min-max-2}}\right)  \\
    &\geq \tau^{\lceil\frac{n}{2}\rceil}\Pi_{j=1}^{\lceil\frac{\lceil\frac{n}{2}\rceil}{2}\rceil-1}\left(2sj\right)\cdot\Pi_{j=0}^{\lfloor\frac{\lceil\frac{n}{2}\rceil}{2}\rfloor-1}\left(s+2sj\right)\\
    &= \tau^{\lceil\frac{n}{2}\rceil}s^{\lceil\frac{n}{2}\rceil-1}\bigg( (2\lceil\frac{\lceil\frac{n}{2}\rceil}{2}\rceil-2)!!\cdot(2\lfloor\frac{\lceil\frac{n}{2}\rceil}{2}\rfloor-1)!! \bigg)\\
    &\geq \tau^{\lceil\frac{n}{2}\rceil}s^{\lceil\frac{n}{2}\rceil-1}\bigg( (2\lfloor\frac{\lceil\frac{n}{2}\rceil}{2}\rfloor-2)!!\cdot(2\lfloor\frac{\lceil\frac{n}{2}\rceil}{2}\rfloor-1)!! \bigg)\\
    &= \tau^{\lceil\frac{n}{2}\rceil}s^{\lceil\frac{n}{2}\rceil-1}(2\lfloor\frac{\lceil\frac{n}{2}\rceil}{2}\rfloor-1)!.
\end{align*}
\begin{align*}
     \min_{x\in C_4}\Pi_{y\in C_2}|x-y|&= \tau ^{\lfloor\frac{n}{2}\rfloor}\Pi_{j=0}^{\lceil\frac{\lfloor\frac{n}{2}\rfloor}{2}\rceil-1}\left(s-1+2sj\right)\cdot\Pi_{k=0}^{\lfloor\frac{\lfloor\frac{n}{2}\rfloor}{2}\rfloor-1}\left(s+1+2sk\right) \qquad \left(\text{case (1) in Lemma \ref{lem:produc-min-max-2}}\right) \\
    &= \tau ^{\lfloor\frac{n}{2}\rfloor}\bigg(s^{\lfloor\frac{n}{2}\rfloor}\Pi_{j=0}^{\lceil\frac{\lfloor\frac{n}{2}\rfloor}{2}\rceil-1}\left(\frac{s-1}{s}+2j\right)\cdot\Pi_{j=0}^{\lfloor\frac{\lfloor\frac{n}{2}\rfloor}{2}\rfloor-1}\left(\frac{s+1}{s}+2j\right)\bigg)\\
    &\geq\tau ^{\lfloor\frac{n}{2}\rfloor} \bigg(s^{\lfloor\frac{n}{2}\rfloor}\cdot\frac{1}{2}\cdot\Pi_{j=1}^{\lceil\frac{\lfloor\frac{n}{2}\rfloor}{2}\rceil-1}\left(2j\right)\cdot\Pi_{j=0}^{\lfloor\frac{\lfloor\frac{n}{2}\rfloor}{2}\rfloor-1}\left(1+2j\right)\bigg)\\
    &= \frac{1}{2}\tau ^{\lfloor\frac{n}{2}\rfloor}s^{\lfloor\frac{n}{2}\rfloor}\bigg((2\lceil\frac{\lfloor\frac{n}{2}\rfloor}{2}\rceil-2)!!\cdot(2\lfloor\frac{\lfloor\frac{n}{2}\rfloor}{2}\rfloor-1)!!\bigg)\\
    &\geq \frac{1}{2}\tau ^{\lfloor\frac{n}{2}\rfloor}s^{\lfloor\frac{n}{2}\rfloor}\bigg((2\lfloor\frac{\lfloor\frac{n}{2}\rfloor}{2}\rfloor-2)!!\cdot(2\lfloor\frac{\lfloor\frac{n}{2}\rfloor}{2}\rfloor-1)!!\bigg)\\
    &= \frac{1}{2}\tau ^{\lfloor\frac{n}{2}\rfloor}s^{\lfloor\frac{n}{2}\rfloor}(2\lfloor\frac{\lfloor\frac{n}{2}\rfloor}{2}\rfloor-1)!.\\
    \min_{x\in C_4}\Pi_{y\in C_3}|x-y|&= \tau ^{\lceil\frac{n}{2}\rceil} \Pi_{j=0}^{\lceil\frac{\lceil\frac{n}{2}\rceil}{2}\rceil-1}\left(2+2sj\right)\cdot\Pi_{k=0}^{\lfloor\frac{\lceil\frac{n}{2}\rceil}{2}\rfloor-1}\left(2s-2+2sk\right) \qquad \left(\text{case (1) in Lemma \ref{lem:produc-min-max-2}}\right)\\
    &\geq 2\tau ^{\lceil\frac{n}{2}\rceil} \Pi_{j=1}^{\lceil\frac{\lceil\frac{n}{2}\rceil}{2}\rceil-1}\left(2sj\right)\cdot\Pi_{k=0}^{\lfloor\frac{\lceil\frac{n}{2}\rceil}{2}\rfloor-1}\left(s+2sk\right)\\
    &= 2\tau ^{\lceil\frac{n}{2}\rceil}s^{\lceil\frac{n}{2}\rceil-1} \bigg((2\lceil\frac{\lceil\frac{n}{2}\rceil}{2}\rceil-2)!!\cdot(2\lfloor\frac{\lceil\frac{n}{2}\rceil}{2}\rfloor-1)!!\bigg)\\
    &\geq2 \tau ^{\lceil\frac{n}{2}\rceil}s^{\lceil\frac{n}{2}\rceil-1} \bigg((2\lfloor\frac{\lceil\frac{n}{2}\rceil}{2}\rfloor-2)!!\cdot(2\lfloor\frac{\lceil\frac{n}{2}\rceil}{2}\rfloor-1)!!\bigg)\\
    &= 2\tau ^{\lceil\frac{n}{2}\rceil}s^{\lceil\frac{n}{2}\rceil-1} (2\lfloor\frac{\lceil\frac{n}{2}\rceil}{2}\rfloor-1)!.\\
     \min_{x\in C_4}\Pi_{y\in C_4,y\neq x}|x-y|&=\tau^{\lfloor\frac{n}{2}\rfloor-1} \Pi^{\lceil\frac{\lfloor\frac{n}{2}\rfloor}{2}\rceil-1}_{j=1}(j\cdot 2s)\cdot\Pi^{\lfloor\frac{\lfloor\frac{n}{2}\rfloor}{2}\rfloor}_{j=1}(j\cdot 2s) \qquad \left(\text{case (1) in Lemma \ref{lem:product2}}\right) \\
     &= \tau^{\lfloor\frac{n}{2}\rfloor-1}s^{\lfloor\frac{n}{2}\rfloor-1}\bigg( (2\lceil\frac{\lfloor\frac{n}{2}\rfloor}{2}\rceil-2)!!\cdot(2\lfloor\frac{\lfloor\frac{n}{2}\rfloor}{2}\rfloor)!!\bigg)\\
     &\geq \tau^{\lfloor\frac{n}{2}\rfloor-1}s^{\lfloor\frac{n}{2}\rfloor-1}\bigg( (2\lfloor\frac{\lfloor\frac{n}{2}\rfloor}{2}\rfloor-2)!!\cdot (2\lfloor\frac{\lfloor\frac{n}{2}\rfloor}{2}\rfloor-1)!!\bigg) \\
     &= \tau^{\lfloor\frac{n}{2}\rfloor-1}s^{\lfloor\frac{n}{2}\rfloor-1}(2\lfloor\frac{\lfloor\frac{n}{2}\rfloor}{2}\rfloor-1)!.
\end{align*}
Thus,

\begin{align*}
c^{\min}_4&\geq \tau^{2n-1}s^{2n-3}\cdot (2\lfloor\frac{\lfloor\frac{n}{2}\rfloor}{2}\rfloor-1)!\cdot(2\lfloor\frac{\lfloor\frac{n}{2}\rfloor}{2}\rfloor-1)!\cdot (2\lfloor\frac{\lceil\frac{n}{2}\rceil}{2}\rfloor-1)!\cdot(2\lfloor\frac{\lfloor\frac{n}{2}\rfloor}{2}\rfloor-1)!\\
&\geq \tau^{2n-1}s^{2n-3}\cdot \bigg((2\lfloor\frac{\lfloor\frac{n}{2}\rfloor}{2}\rfloor-1)!\bigg)^4.
\end{align*}
Summarizing all claims above finishes the proof.
\end{proof}

\bibliographystyle{plain}
\bibliography{reference_final}

\end{document}